\documentclass[11pt]{article}

\usepackage[utf8]{inputenc}

\usepackage[margin=1in]{geometry}
\usepackage{setspace}
\usepackage{newtxtext} 
\usepackage{amsmath,amssymb,amsfonts,amsthm}
\usepackage{mathtools}
\usepackage{bbm, soul} 

\usepackage{graphicx, comment}
\usepackage[round,authoryear]{natbib}
\usepackage{caption}
\usepackage{subcaption}
\usepackage{float}
\usepackage{placeins}
\usepackage{makecell}
\usepackage{booktabs}
\usepackage{xcolor}
\usepackage{cite}
\usepackage{hyperref}
\usepackage[export]{adjustbox}
\hypersetup{colorlinks=true, linkcolor=red, citecolor=blue}

\usepackage{tikz}
\usetikzlibrary{arrows.meta,calc}
\usetikzlibrary{patterns}

\tikzset{
  >=Stealth,
  every node/.style={font=\footnotesize},
  rightFill/.style={fill=blue!18, fill opacity=0.35, draw=blue!55!black, line width=0.5pt},
  midFill/.style={fill=blue!10, fill opacity=0.20, draw=blue!55!black, line width=0.5pt},
  leftFill/.style={fill=orange!15, fill opacity=0.25, draw=orange!70!black, line width=0.5pt},
  pt/.style={circle, draw=black, fill=white, inner sep=0pt, minimum size=3pt},
  pt/.style={circle, draw, fill=white, inner sep=0pt, minimum size=3pt},
  ptPurple/.style={circle, draw=purple!80!black, fill=purple!80!black, inner sep=0pt, minimum size=3.2pt},
  ptRed/.style={circle, draw=red!80!black, fill=red!80!black, inner sep=0pt, minimum size=3.2pt},
  polyDashed/.style={draw, dashed, line width=0.5pt},
  polySolid/.style={draw, line width=0.5pt},
  rightFill/.style={fill=blue!15, draw=none},
  arrow/.style={->, line width=0.5pt},
  tinyArrow/.style={->, line width=0.35pt, >=Stealth},
  lightDashed/.style={->, dashed, line width=0.35pt, draw=black!30},
  lightDashedHi/.style={->, dashed, line width=0.35pt, draw=blue!60},
    dashStrong/.style={dashed, line width=0.9pt, draw=purple!85!black},
  dashStrongHi/.style={dashed, line width=1.1pt, draw=orange!90!black},
}
\tikzset{
  rightGridFill/.style={
    pattern = dots,
    pattern color = blue!40,
    draw = blue!60!black,
    line width = 0.5pt,
    pattern = dots,
pattern color = blue!35,
line width = 0.35pt
  }
}
\tikzset{
  ptRightLight/.style={
    circle,
    fill=blue!30,
    draw=none,
    minimum size=6pt
    pattern = dots

  }
}

\usepackage{algorithm}
\usepackage{algpseudocode}

\usepackage{empheq}

\usepackage{tikz}
\usetikzlibrary{arrows.meta}
\usepackage{circuitikz}

\usepackage{pifont}

\newtheorem{lemma}{Lemma}[section]
\newtheorem{theorem}{Theorem}[section]
\newtheorem{corollary}{Corollary}[section]

\usepackage{xcolor}
\definecolor{green1}{RGB}{0,130,0}

\newcommand{\footremember}[2]{%
  \footnote{#2}%
  \newcounter{#1}%
  \setcounter{#1}{\value{footnote}}%
}

\title{Meta-Learning for Repeated Bayesian Persuasion}
\date{}

\author{%
Ata Poyraz Turna\footremember{boundAPT}{Bogazici University; e-mail: \texttt{ata.turna@std.bogazici.edu.tr}} \and 
Asrın Efe Yorulmaz
\footremember{uiucay20}{University of Illinois Urbana--Champaign; e-mail: \texttt{ay20@illinois.edu}}%
\and
Tamer Başar\footremember{uiucb1} {University of Illinois Urbana--Champaign; e-mail: \texttt{basar1@illinois.edu}}%
}

\begin{document}

\maketitle

\begin{abstract}
Classical Bayesian persuasion studies how a sender influences receivers through carefully designed signaling policies within a single strategic interaction. In many real-world environments, such interactions are repeated across multiple games, creating opportunities to exploit structural similarity across tasks. In this work, we introduce Meta-Persuasion algorithms, establishing the first line of theoretical results for both full-feedback and bandit-feedback settings in the Online Bayesian Persuasion (OBP) and Markov Persuasion Process (MPP) frameworks. We show that our proposed meta-persuasion algorithms achieve provably sharper regret rates under natural notions of task similarity, improving upon the best-known convergence rates for both OBP and MPP. At the same time, they recover the standard single-game guarantees when the sequence of games is picked arbitrarily. Finally, we complement our theoretical analysis with numerical experiments that highlight our regret improvements and the benefits of meta-learning in repeated persuasion environments.
\end{abstract}

\section{Introduction}

Information design has become a central tool for understanding how strategic agents behave when information is scarce, costly, or asymmetric. In the classical \emph{Bayesian Persuasion} framework \citep{KamenicaGentzkow2011BayesianPersuasion}, a sender observes the true state of the world and strategically commits to a signaling policy that shapes the posterior beliefs and consequently the actions of a Bayesian receiver. The sender’s goal is to choose an information structure that induces desirable behavior, despite the receivers acting in their own best interest. This model has found applications in economics \citep{KamenicaGentzkow2011BayesianPersuasion}, policy design \citep{Basar2024,Yorulmaz2025}, online marketplaces \citep{Persuasion_market}, and recommendation systems \citep{Persuasion_rec}, where the ability to influence actions through information is often as valuable as the ability to influence incentives.

\vspace{6pt}
\noindent
While the original formulation concerns a single persuasion instance, many real-world scenarios involve \emph{repeated} persuasion problems that share structural similarity. Regulators routinely interact with similar firms, platforms continually seek to guide users with the same backgrounds across repeated recommendation sessions, and automated systems must repeatedly persuade agents whose preferences vary but are not entirely arbitrary. Each persuasion instance is rarely isolated: it is typically one draw from a family of related problems.

\vspace{6pt}
\noindent
This suggests that a principal may benefit from transferring knowledge across persuasion tasks. In other words, persuasion itself may admit a notion of \emph{learning to persuade}. Concept of meta-learning \citep{1998LearningtoLearn} provides such a paradigm: an agent faces a sequence of related tasks and aims to exploit shared structure to improve performance on each new task. Meta-learning has shown substantial gains in multi-task optimization, online learning, and bandit problems, especially when worst-case guarantees are overly pessimistic for structured environments \citep{AdaGrad_Meta-Kodak2019}.

\vspace{6pt}
\noindent
Despite rapid progress in \emph{Repeated Bayesian Persuasion} (RBP), where  a receiver interacts with a sender who aims to minimize his regret over rounds, the literature has almost exclusively treated each persuasion task independently. Existing no-regret algorithms operate from scratch on every new instance, ignoring any similarity across tasks. This creates a gap, although there are RBP algorithms achieving near-optimal worst-case regret, they may be conservative in settings where tasks share a latent structure, a regime where meta-learning would in principle offer substantial improvements.

\vspace{6pt}
\noindent
This work addresses this gap by incorporating meta-learning techniques into two RBP frameworks studied in the literature, namely \emph{Online Bayesian Persuasion} (OBP) \citep{OBP_orig, OptimalRatesOBP} and \emph{Markov Persuasion Processes} (MPPs) \citep{MarkovPersuasionProcess2022, MarkovPersuasionScratch2025}. The main difference between these two frameworks is that, in MPP, agents interact within a \emph{Markov Decision Process} (MDP) environment. We formalize a setting in which the sender repeatedly engages in persuasion tasks drawn from an unknown but structured environment. Motivated by these considerations, we ask:

\begin{center}
\emph{Can we design meta-learning algorithms with full and bandit-feedback for Repeated Bayesian Persuasion?}
\end{center}

\noindent
We answer this question in the affirmative by designing algorithms that achieve provably faster convergence rates for the cumulative regret of the sender when learning a signaling policy over a sequence of ``similar'' games, in both the OBP and MPP frameworks. Moreover, the convergence rate of the regret achieved by our algorithms strictly reduces upon the best-known bounds, when the sequence of games the sender interacts is chosen adversarially.

\subsection{Related Work}

Computational studies of Bayesian persuasion originate with work of 
\citet{DughmiXu2016AlgorithmicPersuasion}, which studies the efficient methods for computing optimal signaling schemes. In particular, \citet{OBP_orig} extended this framework and introduced OBP, where the sender repeatedly interacts with receivers and learns optimal signaling policies over time. 
This line of work was later extended to multiple receivers by \citet{CastiglioniEtAl2021MultiReceiverOBP}, who analyzed learning dynamics when receivers simultaneously react to the sender’s signals. 
More recently, \citet{OptimalRatesOBP}  proposed gradient-based methods operating in the \emph{loss space}, establishing optimal regret rates for OBP.

\vspace{6pt}
\noindent
A complementary direction considers sequential environments. 
The MPP framework was introduced by \citet{MarkovPersuasionProcess2022} to model repeated persuasion in Markovian environments where the sender sequentially interacts with a stream of receivers. 
This framework was further extended by \citet{MarkovPersuasionScratch2025}, who consider settings in which neither the sender nor the receivers have prior knowledge of the environment and must learn the underlying dynamics from interaction. 
In contrast to these works, our setting is the first work that combines Bayesian persuasion with meta-learning across tasks.

\vspace{6pt}
\noindent
From a methodological perspective, our approach is related to gradient-based meta-learning. 
Theoretical foundations of such methods were studied by \citet{AdaGrad_Meta-Kodak2019}, who established convergence guarantees for meta-learning algorithms using tools from online convex optimization and task-similarity assumptions. 
Meta-learning has also been investigated in bandit settings, including Multi-Armed Bandits (MABs) and Bandit Linear Optimization (BLO), as studied in \citet{MetaLearningBandits}. 
These frameworks are particularly relevant to our OBP formulation, which can be viewed as an instance of bandit linear optimization over signaling policies. 
Finally, bridging meta-learning with game-theoretic learning dynamics, \citet{MetaLearningGames} proposed no-regret meta-learning algorithms that improve convergence rates in strategic settings such as zero-sum, general-sum, and Stackelberg games under full-information feedback.

\section{Preliminaries}

For an integer $n\ge 1$, we write $[n]:=\{1,2,\dots,n\}$. For a statement $p$, let $\mathbf 1\{p\}\in\{0,1\}$ denote its indicator.
For a vector $v\in \mathbb R^m$, $v[i]$ denotes its $i$-th coordinate, and
$\langle u,v\rangle$ denotes the Euclidean inner product with $u$.
We use $\widetilde{O}(\cdot)$ to hide factors logarithmic in their argument(s). Unless stated otherwise, all utilities take values in $[0,1]$. We denote each game (task) with $t\in [T]$ and iteration of the each task as $i \in [m]$. Subscripts are to represent the time index $i \in [m]$ while the superscripts are for the task iterations $t \in [T]$.

\subsection{Online Bayesian Persuasion}
In this work, we focus on a single sender–receiver interaction, as this model can be trivially extended to multiple receivers without inter-agent externalities interacting with an information sender \citep{OptimalRatesOBP}. We assume that the receiver is chosen from a finite set $\mathcal{K}$ with $|\mathcal{K}|$ many different types.  Each receiver type chooses her actions from a finite set of of available actions $\mathcal{A}$, in which $a \in \mathcal{A}$ specifying an action and $\mathsf A = |\mathcal A|$. Moreover, the utilities of the sender and the receiver depend on the current \emph{state of nature}, which is drawn from a finite set $\Omega$ according to the publicly-known prior probability distribution $\mu \in \mathrm{int}(\Delta_{\Omega})$. Therefore, we define the utility functions of both sender and the receiver as  $ u^s,u^{r,k} : {\mathcal{A}  \times \Omega \rightarrow}$ $[0,1]$. In OBP setting, the sender gets to know the realized state of the nature $\omega\sim \mu$, and has the ability to signal agents to maximize his own utility. This is done using a publicly announced \emph{signaling scheme} $\phi: \Omega \rightarrow \Delta_{\mathcal{S}} $ 
, where $\mathcal{S}$ is the finite set of signals.  We define $\phi_\omega(s)$ as the probability of sending $s \in \mathcal{S}$ when the realized state of nature is  $\omega \in \Omega$. The repeated interaction for OBP is as follows :

{\floatname{algorithm}{Protocol}  
\begin{algorithm}[H]
\caption{Sender-Receiver Interaction at $i \in [m]$ for Task $t \in [T]$ in OBP}
\begin{algorithmic}[1]
    \State At each round $i \in [m]$, the sender publicly announces a signaling scheme $\phi_i$, and the state of nature (outcome) $\omega \sim \mu$ is realized.
    \State The sender samples a signal $s \sim \phi_{i,\omega}$ and shares it with the receiver.
    \State Upon receiving the signal $s$, the receiver updates her posterior according to Bayes' rule based on the prior $\mu$ and the announced signaling policy $\phi_{i,\omega}$.
    \State The receiver chooses an action $a \in \mathcal{A}$ that maximizes her utility under the posterior belief.
\end{algorithmic}
\end{algorithm}
} 

\noindent
The posterior $\rho^s \in \Delta_{\Omega}$ after each interaction is calculated as,
\( 
\rho_\omega^{s} \coloneqq
\frac{\mu_{\omega}\,\phi_{\omega}(s)}
     {\sum_{\omega' \in \Omega} \mu_{\omega'} \phi_{\omega'}(s)}
\)  \text{for every } \( \omega \in \Omega.
\) 
Given a posterior $\rho \in \Delta_{\Omega}$, the set of \emph{best-response actions} for the receiver of type $k \in \mathcal{K}$  is defined as, $\mathcal{B}^{k}_{\rho}  \coloneqq  
\arg\max_{a \in \mathcal{A}} \; \sum_{\omega \in \Omega} \rho_{\omega}\, u^{r}_{k}(a,\omega).$
Moreover, assuming receiver break ties in favor of the sender, 
the sender’s expected utility for signaling scheme $\phi$ 
and receiver's type $k \in \mathcal{K}$ is provided as
\( 
u^{s}(\phi, k) 
\coloneqq
\sum_{s \in \mathcal{S}}
\big(
    \arg\max_{a \in \mathcal{B}^{k}_{\rho^{s}}}
    \;
    \sum_{\omega \in \Omega} \mu_{\omega}\, \phi_{\omega}(s)\, u^{s}(a,\omega)
\big).
\) 

\vspace{6pt}
\noindent
We will be focusing on computing a sequence $\{\phi_i\}_{i \in [m]}$ of signaling schemes in an online manner which can be employed by the sender in order to maximize his utility.  We assume that the sequence of receiver’s type profiles 
$\{k_i\}_{i \in [m]}$, with $k_i \in \mathcal{K}$, is selected by an oblivious adversary. At each round $i\in [m]$ of the repeated interaction, the sender receives a payoff $u^{s}(\phi_i, k_i)$ and receives some feedback about the receiver types. In the \emph{full feedback} setting, the sender gets to know the receiver’s type profile $k_i$, while in the \emph{bandit feedback} setting the sender only observes the action profile $a_i \in \mathcal{A}$ played by the receiver at round $i$.

\vspace{6pt}
\noindent
In this work, without loss of generality, we focus on signaling schemes that are direct and persuasive, since the \emph{Revelation Principle} holds in our setting. Particularly, a signaling scheme $\phi$ is \emph{direct} if signals correspond to action recommendations where the set of signals of a receiver is $\mathcal S = \mathcal A^{|\mathcal{K}|}$, with each signal defining an action recommendation to the receiver type. Moreover, a direct signaling scheme is \emph{persuasive} if receiver type is incentivized to follow the action recommendations issued by the sender.
Formally, the set of direct and persuasive signaling schemes $\mathcal{P}$ is the set of all $\phi : \Omega \to \Delta_{A^{|\mathcal{K}|}}$ such that, for each
receiver's type $k \in \mathcal{K}$, and each action $a \in \mathcal A$, it holds
\( 
\sum_{\omega \in \Omega} \sum_{s \in \mathcal A^{|\mathcal K|}} 
\mu_\omega \phi_\omega(s)\big( u^{r,k}(a_k, \omega) - u^{r,k}(a, \omega) \big) \ge 0
\) 
where, we let $a_k$ be the action in direct signal $s$ corresponding to type $k \in \mathcal{K}$.

\vspace{6pt}
\noindent 
We note that the set $\mathcal{P}$ can be represented as a polytope, due 
to the persuasiveness constraints. We impose the conditions ensuring that 
$\phi$ is a valid signaling rule. Namely
\( 
\sum_{s \in \mathcal A^{\mathcal{|K|}}} \phi_{\omega}(s) = 1 \) \( \text{for all } \omega \in \Omega .
\) 
Finally, given any direct and persuasive signaling scheme $\phi \in \mathcal{P}$,
the sender’s utility under type profile $k \in \mathcal{K}$ is
\( 
u^s(\phi,k) 
:= \sum_{\omega \in \Omega} \sum_{s \in \mathcal A^{\mathcal{|K|}}} 
\mu_\omega \phi_\omega(s) 
u^s\!\big( a_{k}, \omega \big).
\)

\subsection{Markov Persuasion Processes} 
MPPs extend the classical one-shot Bayesian persuasion model to dynamic environments where a sender interacts sequentially with multiple receivers within an MDP. In this setting, the sender encounters a sequence of myopic receivers who choose actions based solely on immediate payoffs, without considering future ones. Formally an episodic MPP is defined by a tuple for the task $t \in [T]$ as $M \coloneqq ( X,A,\Omega,\mu, P,\{u_i^s \}^{m}_{i=1},\{u_i^r \}^{m}_{i=1}) $ where:

\begin{itemize}
  \item $m$ is the number of episodes
  \item $X$, $A$, and $\Omega$ are finite sets of states, actions, and outcomes, respectively.
  \item $\mu: X \to \Delta(\Omega)$ is a prior function defining a probability distribution over outcomes at each state. We let $\mu(\omega\mid x)$ be the probability with which outcome $\omega\in\Omega$ is sampled in state $x\in X$.

  \item $P: X\times\Omega\times A \to \Delta(X)$ is a transition function. We let $P(x'\mid x,\omega,a)$ be the probability of moving from $x\in X$ to $x'\in X$ by taking action $a\in A$, when the outcome sampled in state $x$ is $\omega\in\Omega$.

  \item $\{u_i^s \}^{m}_{i=1}$ is a sequence specifying a sender's reward function
  $u^s_{i}: X\times\Omega\times A \to [0,1]$ at each episode $i$.
  Given $x\in X$, $\omega\in\Omega$, and $a\in A$, each $u^{s,t}_i(x,\omega,a)$ for $i\in[m]$ is sampled independently from a bounded distribution between $[0,1]$
   with mean $u^{s,t}(x,\omega,a)$ for task $t \in [T]$.

  \item $\{u_i^r \}^{m}_{i=1}$ is a sequence defining a receiver's reward function
  $u_i^{r,t}: X\times\Omega\times A \to [0,1]$ at each episode $t$.
  Given $x\in X$, $\omega\in\Omega$, and $a\in A$, each $u^{r,t}_i(x,\omega,a)$ for $i\in[m]$ is sampled independently from a bounded distribution between $[0,1]$
   with mean $u^{r,t}(x,\omega,a)$ for task $t \in [T]$.
\end{itemize}

\noindent
As we are interested with episodic MPPs, we focus on MPPs enjoying the \emph{loop-free} property, as justified in case of online learning in MDPs \citep{MarkovPersuasionScratch2025, ConvexoptimizationAdversarialMDP}. In a loop-free MPP, states are partitioned into $L+1$ layers $X_0 \cdots X_L$ such that $X_0 \coloneqq \{ x_0 \}$ and $X_L \coloneqq \{x_L\}$ with $x_0$ being the initial state and $x_L$ being the final one, in which the episode ends. Moreover, by letting $\mathcal{K} = [0\cdots L-1]$, $P(x'|x,\omega,a) > 0$ only when $x' \in X_{k+1}$ and $x \in X_k$ for some $k \in \mathcal{K}$.

\vspace{6pt}
\noindent
In the MPP framework, the sender publicly commits to a signaling policy 
\(
\phi : X \times \Omega \to \Delta(\mathcal{S}),
\)
which specifies, for every state \( x \in X \) and outcome \( \omega \in \Omega \), 
a distribution over signals. We write
\( 
\phi(\cdot \mid x, \omega) \in \Delta(\mathcal{S}),
\) 
where \( \phi(s \mid x, \omega) \) denotes the probability of sending signal 
\( s \in \mathcal{S} \) when the system is in state \( x \) and the realized 
outcome is \( \omega \). Analogous to each round of OBP, a myopic receiver who observes state 
\( x \in \mathcal{X} \) and receives signal \( s \in \mathcal{S} \) 
updates her belief over outcomes via Bayes’ rule and selects a best-response 
action. We denote by
\( 
b^{\phi}(s, x) \in A
\) 
the action chosen as a best response under the signaling policy \( \phi \). Furthermore, we assume that neither the sender nor the receivers have any prior knowledge of the transition kernel $P$, the prior distribution $\mu$, or the reward functions $u_i^{s,t}(x,\omega,a)$ and $u_i^{r,t}(x,\omega,a)$.

\vspace{6pt}
\noindent
As is in the OBP setting, \emph{Revelation Principle} allows us to focus 
on signaling policies that are \emph{direct} and \emph{persuasive}. 
Formally, a signaling policy is \emph{direct} if the set of signals coincides with 
the set of actions, and a signaling policy 
\( 
\phi : X \times \Omega \to \Delta(A)
\) 
is persuasive if, for every state $x \in X$ and action recommendation 
$a \in A$, the inequality,  
\( 
\sum_{\omega \in \Omega} 
\mu(\omega \mid x)\phi(a \mid x, \omega)
\Big(
u^{r,t}(x, \omega, a) 
- u^{r,t}\big(x, \omega, b_{\phi}(a, x)\big)
\Big)
\ge 0, 
\) holds. 

\vspace{6pt}
\noindent
To enable meta-learning across repeated games, we assume an across-task model for the
task-dependent primitives. For each task $t\in[T]$, let
\(
P^t(\cdot\mid x,\omega,a),\; \mu^t(\cdot\mid x),\;
u^{s,t}(x,\omega,a),\; u^{r,t}(x,\omega,a)
\) 
denote the task-specific mean transition kernel, prior, and sender/receiver reward means, respectively.
We posit the existence of global, across-task, reference means
\(
P_G(\cdot\mid x,\omega,a),\; \mu_G(\cdot\mid x),\;
u_G^{s}(x,\omega,a),\; u_G^{r}(x,\omega,a),
\)
such that, for every fixed coordinate, $(x,\omega,a,x')\in X\times\Omega\times A\times X$ and
$(x,\omega,a)\in X\times\Omega\times A$, the corresponding task parameter is drawn i.i.d.\ across tasks
around its global mean with bounded inter-task variance. Concretely, for each $(x,\omega,a,x')$,
the scalar random variable $P^t(x'\mid x,\omega,a)$ is i.i.d.\ over $t$ with
\( 
\mathbb{E}\!\left[P^t(x'\mid x,\omega,a)\right]=P_G(x'\mid x,\omega,a),
\;
\mathrm{Var}\!\left(P^t(x'\mid x,\omega,a)\right) \le \iota_P^2,
\) 
and for each $(x,\omega)$,
\( 
\mathbb{E}\!\left[\mu^t(\omega\mid x)\right]=\mu_G(\omega\mid x),
\;
\mathrm{Var}\!\left(\mu^t(\omega\mid x)\right) \le \iota_\mu^2.
\) 
Similarly, for each $(x,\omega,a)$,
\( 
\mathbb{E}\!\left[u^{s,t}(x,\omega,a)\right]=u_G^{s}(x,\omega,a),
\;
\mathrm{Var}\!\left(u^{s,t}(x,\omega,a)\right) \le \iota_{s}^2,
\;
\mathbb{E}\!\left[u^{r,t}(x,\omega,a)\right]=u_G^{r}(x,\omega,a),
\;
\mathrm{Var}\!\left(u^{r,t}(x,\omega,a)\right) \le \iota_{r}^2.
\) 

\vspace{6pt}
\noindent
Within each task $t$, episode-wise observations are generated with bounded support and bounded
within-task variance. In particular, conditioned on the task means above, rewards are sampled
independently across episodes with support in $[0,1]$ and
\[
\mathbb{E}\!\left[u_{i}^{s,t}(x,\omega,a)\mid u^{s,t}(x,\omega,a)\right]=u^{s,t}(x,\omega,a),
\quad
\mathrm{Var}\!\left(u_{i}^{s,t}(x,\omega,a)\mid u^{s,t}(x,\omega,a)\right) \le \sigma_s^2,
\]
\[
\mathbb{E}\!\left[u_{i}^{r,t}(x,\omega,a)\mid u^{r,t}(x,\omega,a)\right]=u^{r,t}(x,\omega,a),
\quad
\mathrm{Var}\!\left(u_{i}^{r,t}(x,\omega,a)\mid u^{r,t}(x,\omega,a)\right) \le \sigma_r^2,
\]
for all $(x,\omega,a)$ and all episodes $i\in[m]$. 
Finally, the sender-receiver interaction for  time $i \in [m]$ for task $t \in [T]$ is as follows from \citep{MarkovPersuasionScratch2025}:
{\floatname{algorithm}{Protocol}  
\begin{algorithm}[H]
\caption{Sender-Receivers Interaction at $i \in [m]$ for Task $t \in [T]$ in MPPs} \label{protoc2}
\begin{algorithmic}[1] 
\State All the rewards $u_i^{s,t}(x,\omega,a)$, $u^{r,t}_{i}(x,\omega,a)$ are sampled
\State Sender publicly commits to $\phi^t_i : X \times \Omega \to \Delta(A)$
\State The state of the MPP is initialized to $x_0$
\For{$k = 0, \ldots, L-1$}
  \State Sender observes outcome $\omega_k \sim \mu(x_k)$
  \State Sender draws recommendation $a_k \sim \phi(\cdot \mid x_k,\omega_k)$
  \State A new Receiver observes $a_k$ and plays it
  \State The MPP evolves to state $x_{k+1} \sim P(\cdot \mid x_k,\omega_k,a_k)$
  \State Sender observes the next state $x_{k+1}$
\EndFor
\State Sender observes \textit{feedback} for every $k \in [0 \ldots L-1]$:
\begin{itemize}
  \item \textit{full} $\to u_i^{s,t}(x_k,\omega_k,a),\, u_i^{r,t}(x_k,\omega_k,a)\ \ \forall a \in A$
  \item \textit{partial} $\to u_i^{s,t}(x_k,\omega_k,a_k),\, u_i^{r,t}(x_k,\omega_k,a_k)$
\end{itemize}
\end{algorithmic}
\end{algorithm}
}

\vspace{6pt}
\noindent
We emphasize that neither the sender nor the receiver types have any knowledge of the transition kernel $P$, the prior distribution $\mu$, or the reward functions $u_i^{s,t}(x,\omega,a)$ and $u_i^{r,t}(x,\omega,a)$, including any prior information about their underlying distributions. 

\vspace{6pt}
\noindent 
Under this assumption, Protocol~\ref{protoc2} prescribes that receiver types always follow the recommended actions. The rationale is that, in the oblivious MPP setting, learning algorithms ensure that the average per-round \emph{violation} of persuasiveness constraints converges to zero as the number of episodes increases. Since no algorithm can guarantee persuasiveness in every single episode, if algorithm wants to be a no-regret algorithm, the violation necessarily vanishes only asymptotically. Consequently, it is optimal for receiver types to adhere to the recommendations in the long run \citep{MarkovPersuasionScratch2025}.

\subsection{Meta-Learning Across Repeated Games}

We consider the problem of meta-learning across  tasks $t \in [T]$
over some compact and convex action set $\mathcal{J} \subseteq \mathbb{R}^K$.  
On each round, $i \in [m]$ of task $t \in [T]$ we play action $\mathbf{x}^t_{i} \in \mathcal{J}$ and receive feedback
$\mathcal{L}^t_{i}(\mathbf{x}^t_{i})$ for some loss function 
$\mathcal{L}^t_{i} : \mathcal{J} \mapsto [0,1]$.
For the class of loss functions, we assume that they have a linear form of 
$\mathcal{L}^t_{i}(\mathbf{x}) = \langle \ell^t_{i}, \mathbf{x} \rangle$.

\vspace{6pt}
\noindent
In online learning, the goal in a single task is to play actions 
$\mathbf{x}^t_{1}, \ldots, \mathbf{x}^t_{m}$ that minimize the regret
$\sum_{i=1}^m \mathcal{L}^t_{i}(\mathbf{x}^t_{i}) - \mathcal{L}^t_{i}(\mathbf{x}^{\ast t})$,
with respect to $\mathbf{x}^{\ast t} \in \arg\min_{\mathbf{x} \in \mathcal{J}}
\sum_{i=1}^m \mathcal{L}^t_{i}(\mathbf{x})$.
Lifting this to the meta-learning setting, we define our goal as minimizing the \emph{task-averaged regret} given as:
\begin{equation}
R^T_m \coloneqq\frac{1}{T} \sum_{t=1}^T \sum_{i=1}^m 
\mathcal{L}^t_{i}(\mathbf{x}^t_{i}) - \mathcal{L}^t_{i}(\mathbf{x}^{\ast t})
\end{equation}
\noindent
In particular, we aim to leverage multi-task data to improve the average performance. Formally, our goal is to achieve a task-averaged regret of
\( 
\tilde{\mathcal{O}}(\sqrt{V m}),
\) 
where $V \in \mathbb{R}_{\ge 0}$ is a task-similarity measure that remains small when tasks are highly similar, while still recovering the worst-case single-task performance when they are heterogeneous.

\vspace{6pt}
\noindent
To this end, we adopt a meta-learning perspective. Specifically, we aim to learn a within-task algorithm (or base learner), i.e., a parameterized method that is deployed independently on each task $t$. The objective is to learn improved \emph{initializations} and meta-hyperparameters that minimize the task-averaged regret across tasks \citep{MetaLearningGames,MetaLearningBandits}. The underlying premise is that the task-specific optimal parameters are close to one another; hence, a suitably meta-learned initialization enables rapid adaptation, yielding strong performance after only a few within-task updates.

\vspace{6pt}
\noindent
The base learner we choose for the \emph{meta-persuasion} algorithms for OBP framework is the Online Mirror Descent (OMD). For a strictly convex \emph{regularizer} $\mathcal{R}: \mathcal{\bar{J}} \mapsto \mathbb{R}$ and step-size $\eta > 0$, the OMD update is being performed as,
\[
\mathbf{x}^t_{i+1}
= \Big\{\arg\min_{\mathbf{x} \in \overline{\mathcal{J}}}
D_{\mathcal{R}}(\mathbf{x} \| \mathbf{x}^t_{i})
+ \eta \langle \ell^t_{i}, \mathbf{x} \rangle \Big\}
\]

\noindent where 
$D_{\mathcal{R}}(\mathbf{x} \| \mathbf{y}) 
= \mathcal{R}(\mathbf{x}) - \mathcal{R}(\mathbf{y}) 
- \langle \nabla \mathcal{R}(\mathbf{y}), \mathbf{x} - \mathbf{y} \rangle$ 
is the \emph{Bregman divergence} of $\mathcal{R}$. It is notable that, 
OMD recovers online gradient descent (OGD) when 
$\mathcal{R}(\mathbf{x}) = \tfrac{1}{2}\|\mathbf{x}\|_2^2$, 
in which case 
$D_{\mathcal{R}}(\mathbf{x} \| \mathbf{y}) = \tfrac{1}{2}\|\mathbf{x} - \mathbf{y}\|_2^2$. Specifically, we utilize OGD and OMD with a self-concordant barrier $\mathcal{R}$ as the regularizer, which serve as the base learners in the full-feedback and bandit-feedback settings, respectively. We remark that OMD and follow the regularized leader (FTRL) methods recover the same iterates under the condition of the regularizer is both convex and differentiable as in our case \citep{CompetingIntheDark}. On the other hand, for the MPP framework, we consider carefully designed estimators for enabling meta-learning task, as specified in the Section \ref{sec:MPP}.

\section{The Meta-Learning for Online Bayesian Persuasion}
In the online learning problem for OBPs, at each round 
$i \in m$, of task $t \in T$ an agent takes a decision $\phi^t_i$ from a set $\mathcal{P} \subseteq \mathbb{R}^M$, 
and, then, an adversary selects an element $k_i$ from a finite set 
$\mathcal{K}$ of $K := |\mathcal{K}|$ elements. 
Then, the loss suffered by the agent is $\mathcal{L}^t_{k_i}(\phi^t_i)$, where functions 
$\mathcal{L}_k : \mathcal{P} \rightarrow [0,1]$ are loss functions indexed by the elements $k \in \mathcal{K}$. 
Thus, the performance of the agent over the $m$ rounds of $T$ tasks is evaluated in terms of task averaged regret given as
\begin{equation}
    R^T_m \coloneqq \frac{1}{T} \sum_{t=1}^T \sum_{i=1}^m 
\mathbb{E}[\mathcal{L}_{k^t_i}(\phi^t_i)] - 
\mathcal{L}_{k^t_i}\big({(\phi^{*})}^{t}\big),
\end{equation}
where the expectation is with respect to the (possible) randomization that the agent adopts in choosing $\phi^t_i$. 

\vspace{6pt} 
\noindent
Next, we introduce a general no-regret algorithm that works by exploiting the linear structure of the 
online learning problem described above. 
In order to do so, we introduce a vector-valued function 
$\nu : \mathcal{P} \to \mathbb{R}^K$ defined as 
\( 
\nu(\phi) := [-u^s(\phi,k)]_{k \in \mathcal{K}}=  [\mathcal{L}_{k}(\phi)]_{k \in \mathcal{K}} \)  and \( \mathcal{L}_{k^t_i}(\phi^t_i) :=\nu(\phi)^\top \mathbf{1}_{k^t_i} \) \( \text{ for all } \phi \in \mathcal{P} . 
\) We denote the convex hull of such functions as \(\bar \nu (\cdot)\).
Furthermore, we assume that $\nu$ is a linear map, i.e. there exists $\mathbf{M} \in \mathbb{R}^{K \times M}$ such that $\nu(\phi) = \mathbf{M\phi}$ for all $\phi \in \mathcal{P}$. Knowing that inverse map $\nu^\dagger$ exists, we can map our signaling scheme $\phi$ to the loss space $\nu(\mathcal{P})$ to perform the iterations and map back the result to our set of direct and persuasive signaling schemes $\mathcal{P}$ to play an actual scheme $\phi^t_i$. However, since the set $\nu(\mathcal{P})$ is not guaranteed to be convex, we make use of  \emph{Carath\'eodory's Theorem}.

\begin{theorem}[Carath\'eodory's Theorem]
    For any set $\mathcal{J} \subset \mathbb{R}^{K}$ and any point $\bar{z}$ in its convex hull $\bar{\mathcal{J}}$, there exist at most $K+1$ points $z^{1},\ldots,z^{n} \in \mathcal{J}$ with $n \le K+1$ such that $z = \sum_{i=1}^{n} \lambda_{i}\, x^{i}$,
where $\lambda_{i} \ge 0$ and $\sum_{i=1}^{n} \lambda_{i} = 1$.     
\end{theorem}

\vspace{6pt} 
\noindent
Departing from the theorem, we now describe the \emph{Carath\'eodory Oracle}, which takes  
$\nu(\mathcal{P})$, the loss space, and a point in its convex hull, as inputs and returns a sparse representation 
of the input point using elements from the original set. Concretely, given a point 
$\bar{z}_i^t \in \bar \nu(\mathcal{P})$, the oracle returns 
at most $K+1$ points in $\nu(\mathcal{P})$ together with associated weights such that 
$\bar{z}_i^t$ can be expressed as their convex combination. By Carath\'eodory’s Theorem, such a representation always exists in $\mathbb{R}^K$. 
Formally, 
\[
\{(z_{i,j}^t, \lambda_{i,j}^t)\}_{j=1}^{m}
\gets 
\mathrm{Carath\acute{e}odory}\bigl(\bar{z}_i^t,\, \nu(\mathcal{P})\bigr),
\]
where $m \le K+1$, $z_{i,j}^t \in \nu(\mathcal{P})$, 
$\lambda_{i,j}^t \ge 0$, $\sum_{j=1}^{m} \lambda_{i,j}^t = 1$, and
\( 
\bar{z}_i^t = \sum_{j=1}^{m} \lambda_{i,j}^t z_{i,j}^t.
\) We then sample one of the points $z_{i,j}^t$ with probability 
$\lambda_{i,j}^t$ and play the corresponding signaling scheme 
\( 
\phi_i^t = \nu^\dagger(z_{i,j}^t),
\) 
where $\nu^\dagger$ denotes the inverse map. Since $\nu^\dagger$ exists and is efficiently computable, 
this procedure yields a signaling scheme whose expected loss vector 
matches $(\bar{z}_i)^t$, thereby implementing the desired policy in expectation. Then, by equivalence in expectation, our algorithm performs OGD over the convex domain $\bar \nu(\mathcal{P})$. At each iteration, the resulting iterate in the lifted space is mapped back to the signaling space $\mathcal{P}$, from which the sender samples and implements a signaling policy. We depict this process in Figure \ref{fig:1}. Importantly, to our knowledge such an algorithm first introduced in \citet{OptimalRatesOBP}.


\vspace{6pt} 
\noindent
A crucial observation is that computing an optimal direct and persuasive signaling scheme is NP-hard, even when the distribution over the receiver type is known \citep{OBP_orig}. This hardness result implies that the polytope $\mathcal{P}$ of feasible signaling schemes has exponential size. 
Moreover, reductions from offline to online optimization indicate that no computationally efficient algorithm with polynomial per-iteration running time can exist for this problem \citep{OBP_orig}. 
Consequently, in the OBP setting, the primary objective is to improve the sender’s sample complexity, rather than to address the computational complexity of identifying direct and persuasive signaling schemes $\phi \in \mathcal{P}$.

\vspace{6pt} 
\noindent
Furthermore, the reason our optimization procedure is carried out in the convex hull of the loss space, $\bar{\nu}(\mathcal{P})$, rather than directly over the set of direct and persuasive signaling schemes, $\bar{\mathcal{P}}$, stems from the formulation of OBP. In this setting, the receiver type set is significantly smaller than the set of signaling schemes, which is exponential in size; that is, $|\mathcal{K}| \ll |\mathcal{P}|$. Consequently, performing optimization in the loss space reduces the task-averaged regret, since the cardinality of the decision set appears in the regret upper bound. Hence, operating in a lower-dimensional space yields improved performance guarantees.

\vspace{6pt} 
\noindent
In the OBP framework, task heterogeneity is introduced by allowing the player type set $\mathcal{K}$ to vary across tasks, as well as the prior $\mu$, receiver utilities $u^{r,k}$, for each type $k \in \mathcal{K}$, and the sender utility $u^{s}$. Meanwhile, the action set $\mathcal{A}$ and the signal set $\mathcal{S}$ remain fixed across tasks.

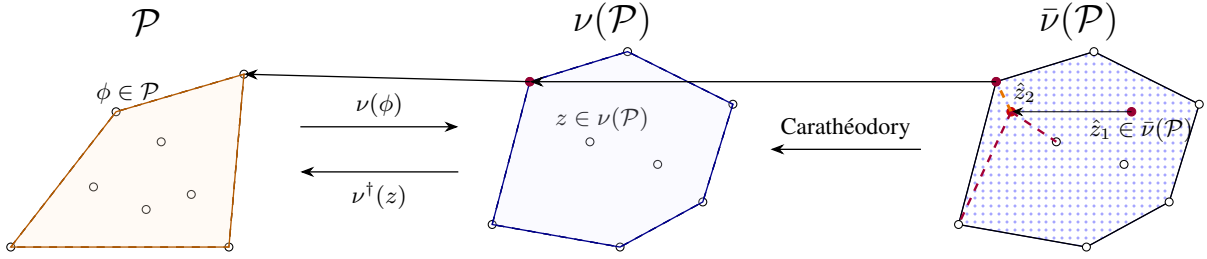
\begin{figure}[t] 
\centering
\begin{tikzpicture}[x=1cm,y=1cm] 

\coordinate (SLeft)  at (0,0);
\coordinate (SMid)   at (6.2,0);
\coordinate (SRight) at (12.4,0);

\node[font=\Large] at (0,3.2) {$\mathcal{P}$};

\coordinate (LA) at (-1.8,0.2);
\coordinate (LB) at (-0.4,2.0);
\coordinate (LC) at ( 1.3,2.5); 
\coordinate (LD) at ( 1.1,0.2);

\draw[polyDashed] (LA)--(LB)--(LC)--(LD)--cycle;
\foreach \p in {LA,LB,LC,LD} \node[pt] at (\p) {};

\coordinate (Li1) at (-0.7,1.0);
\coordinate (Li2) at ( 0.2,1.6);
\coordinate (Li3) at ( 0.6,0.9);
\coordinate (Li4) at ( 0.0,0.7);
\foreach \p in {Li1,Li2,Li3,Li4} \node[pt] at (\p) {};

\node[anchor=west] at ($(Li2)+(-1,0.65)$) {$\phi \in \mathcal{P}$};

\node[font=\Large] at ($(SMid)+(0,3.2)$) {$\nu(\mathcal{P})$};

\coordinate (MA) at ($(SMid)+(-1.6,0.5)$);
\coordinate (MB) at ($(SMid)+(-1.1,2.4)$); 
\coordinate (MC) at ($(SMid)+( 0.2,2.8)$);
\coordinate (MD) at ($(SMid)+( 1.6,2.1)$);
\coordinate (ME) at ($(SMid)+( 1.2,0.8)$);
\coordinate (MF) at ($(SMid)+( 0.1,0.2)$);

\draw[polyDashed] (MA)--(MB)--(MC)--(MD)--(ME)--(MF)--cycle;
\foreach \p in {MA,MB,MC,MD,ME,MF} \node[pt] at (\p) {};

\coordinate (Mi1) at ($(SMid)+(-0.3,1.6)$);
\coordinate (Mi2) at ($(SMid)+( 0.6,1.3)$);
\foreach \p in {Mi1,Mi2} \node[pt] at (\p) {};

\node[anchor=west] at ($(Mi2)+(-1.5,0.6)$) {$z \in \nu(\mathcal{P})$};

\node[font=\Large] at ($(SRight)+(0,3.2)$) {$\bar{\nu}(\mathcal{P})$};

\coordinate (RA) at ($(SRight)+(-1.6,0.5)$);
\coordinate (RB) at ($(SRight)+(-1.1,2.4)$); 
\coordinate (RC) at ($(SRight)+( 0.2,2.8)$);
\coordinate (RD) at ($(SRight)+( 1.6,2.1)$);
\coordinate (RE) at ($(SRight)+( 1.2,0.8)$);
\coordinate (RF) at ($(SRight)+( 0.1,0.2)$);

\path[rightGridFill] (RA)--(RB)--(RC)--(RD)--(RE)--(RF)--cycle;
\draw[polySolid] (RA)--(RB)--(RC)--(RD)--(RE)--(RF)--cycle;

\path[leftFill] (LA)--(LB)--(LC)--(LD)--cycle;
\draw[polyDashed, draw=orange!70!black] (LA)--(LB)--(LC)--(LD)--cycle;

\path[midFill] (MA)--(MB)--(MC)--(MD)--(ME)--(MF)--cycle;
\draw[polyDashed, draw=blue!55!black] (MA)--(MB)--(MC)--(MD)--(ME)--(MF)--cycle;

\coordinate (Ri1) at ($(SRight)+(-0.3,1.6)$);
\coordinate (Ri2) at ($(SRight)+( 0.6,1.3)$);
\foreach \p in {Ri1,Ri2} \node[pt] at (\p) {};

\coordinate (z1) at ($(SRight)+(0.7,2.0)$);
\node[ptPurple] at (z1) {};
\node[anchor=west] at ($(z1)+(-0.7,-0.25)$) {$\hat z_1 \in \bar{\nu}(\mathcal{P})$};

\coordinate (z2) at ($(z1)+(-1.6,0)$);
\node[ptRed] at (z2) {};
\node[anchor=south] at ($(z2)+(0.18,0)$) {$\hat z_2$};

\draw[tinyArrow] (z1) -- (z2);

\draw[dashStrongHi] (z2) -- (RB);
\draw[dashStrong]   (z2) -- (RA);
\draw[dashStrong]   (z2) -- (Ri1);

\foreach \p in {RA,RB,RC,RD,RE,RF}
  \node[pt] at (\p) {};

\draw[arrow] (2.05,1.8) -- (4.15,1.8) node[midway,above] {$\nu(\phi)$};
\draw[arrow] (4.15,1.2) -- (2.05,1.2) node[midway,below] {$\nu^{\dagger}(z)$};
\draw[arrow] (10.3,1.5) -- (8.3,1.5) node[midway,above] {Carath\'eodory};

\node[ptPurple] at (RB) {};
\node[ptPurple] at (MB) {};

\draw[arrow] (RB) -- (MB);
\draw[arrow] (MB) -- (LC);

\end{tikzpicture}
\caption{Illustration of the Carath\'eodory oracle used in Algorithms \ref{alg:omd-tuning-fb} and \ref{alg:omd-tuning}.}
\label{fig:1}
\end{figure}

\subsection{Full Feedback Setting}

For the full-feedback setting, the sender observes the types of the receiver encountered at each interaction. Then, we can employ OGD, where gradients are with respect to $\nu(\phi^t_i)$, as specified in Algorithm~\ref{alg:omd-tuning-fb}. Our ultimate goal is to learn the hyperparameter $\eta$ and to identify a suitable initialization for subsequent tasks.

\begin{algorithm}[ht]
\caption{$\epsilon$-EWOO}
\label{alg:epsilon-ewoo}

\begin{algorithmic}[1]
\State \textbf{Require:} meta-hyperparameter $\beta > 0, \ \{ \tilde U^{(s)}(\eta) \}_{s=1}^{t} $
\State \textbf{Initialize:}  \ $\eta^1 \in [\epsilon, \sqrt{A^{2} + \epsilon^{2}}] $

 
\State 
    \( 
        \eta^{(t+1)} 
        \gets 
        \frac{\displaystyle\int_{\epsilon}^{\sqrt{A^{2}+\epsilon^{2}}} 
            \eta \exp\Bigl(-\beta \sum_{s \le t} \tilde{U}^{(s)}(\eta)\Bigr)\,d\eta}
        {\displaystyle\int_{\epsilon}^{\sqrt{A^{2}+\epsilon^{2}}}
            \exp\Bigl(-\beta \sum_{s \le t} \tilde{U}^{(s)}(\eta)\Bigr)\,d\eta}
    \) 

\end{algorithmic}
\end{algorithm}

\begin{algorithm}
\caption{Full-feedback Meta-Persuasion}
\label{alg:omd-tuning-fb}

\begin{algorithmic}[1]
\State \textbf{Require:}  inverse-map $\nu^{\dagger}$, meta-hyperparameter $\beta > 0$,
\State \textbf{Initialize:} $ (\bar z_1)^1 = \arg\min_{z\in\bar\nu(\mathcal P)}\ \frac{1}{2} \|z\|^2_2 $
\For{\textbf{task} $t = 1,\dots,T$}
    \For{$i = 1,\dots,m$}
         \State $\{ (z^t_{i,j}\,, \lambda^t_{i,j})\}_{j \in [K+1]}
{ \gets} {\mathrm{Carath\acute{e}odory \ }\!\bigl((\bar z_i)^{t},\, \ \nu(\mathcal{P})\bigr)}
$
        \State Draw $j' \in [K+1]$ with probabilities $\lambda^t_{i,1},\lambda^t_{i,2},\cdots,\lambda^t_{i,K+1}$
        \State Play $\phi^t_i \gets \nu^{\dagger}(z^t_{i,j'})$
        \State Observe $k^t_i \in \mathcal{K}$ and suffer loss $\mathcal{L}_{k^t_i}(\cdot)$
        \State $(\bar z_{i+1})^t \gets \Pi_{\bar{\nu}(\mathcal{P})}\big((\bar z_i)^{t} - \eta^t\nabla\mathcal{L}_{k^t_i}(\phi^t_i) \big)$
    \EndFor
    \State ${(z^*)}^{t} \gets \arg\min_{z\in\bar{\nu}(\mathcal P)} \langle \sum^m_{i=1}\textbf{1}_{k^t_i},z\rangle$ 
    \State $(\bar z_{1})^{t+1} \gets \frac{1}{t}\sum_{s\leq t} {(z^*)}^s$ 
    \State $ \tilde U^{(t)}(\eta) \gets \left(
     \frac{\| (z^*)^{t} - (\bar z_1)^{t} \|_2^2}{m\eta} \!+\! \frac{\rho^2 \!A^2}{\eta} \!
    +\! \eta
\right)\!
\frac{m}{2} $
    \State $\eta^{t+1} \gets \epsilon\text{-EWOO}\bigl(\beta,\{\tilde U^{(s)}(\eta)\}_{s=1}^{t}\bigr)$
\EndFor
\end{algorithmic}
\end{algorithm}

\noindent
To achieve this goal, the Algorithm \ref{alg:epsilon-ewoo} learns a sequence of losses for each task $t$ of the form 
\[ U^{(t)}(\eta) \coloneqq\left( \frac{(B^{(t)})^2}{\eta} + \eta \right) \gamma^{(t)} = \left( \frac{\| (z^\ast)^{t} - (z_1)^{t} \|_2^2}{m\eta} + \eta \right)\frac{m}{2} \]
\noindent
and applies the main idea of Exponentially Weighted Online Optimization (EWOO) method~\citep{LogRegretAlgos} to obtain an updated value of $\eta$ for the next task $t+1$. However, $U^{(t)}(\eta)$ functions we derive are not exp-concave and are ill-conditioned 
near $\eta=0$. Therefore, we employ the modified 
version of the algorithm, $\epsilon$-EWOO \citep{AdaGrad_Meta-Kodak2019}. For $\epsilon$-EWOO we define $\tilde U^t(\cdot) $ as:
\[
\begin{aligned}
\tilde{U}^t(\eta) \coloneqq \left( \!
    \frac{(B^{(t)})^2 \!+\! \epsilon^2}{\eta} \!+\! \eta 
\right)\gamma^{(t)}
\!=\!
\left(
    \frac{ \frac{\| (z^\ast)^{t} - (z_1)^{t} \|_2^2}{m} \!+\! \rho^2 \!A^2 \!\!}{\eta}
    \!+\! \eta
\right)\!
\frac{m}{2}.
\end{aligned}
\]
where we specify $\epsilon = \rho A = \sqrt{\frac{K}{m}}\, T^{-1/4}$. As shown in Algorithm \ref{alg:epsilon-ewoo}, $\epsilon$-EWOO differs from EWOO~\citep{LogRegretAlgos} only through this modified objective and the corresponding adjusted integration limits, which together ensure that the loss functions are smooth and convex. Then, using $\epsilon$-EWOO as a subroutine for meta-learning, we iterate over each \emph{persuasion task} as described for the single instance of an OBP process and formalized in Algorithm~\ref{alg:omd-tuning-fb}. Next, we present our theorem for the \emph{meta-persuasion} for OBP framework within the full-feedback setting.

\begin{theorem}\label{thm:OBP-full-fb}
Algorithm~\ref{alg:omd-tuning-fb} with parameters
\( 
\epsilon = \rho A, 
\) \(  
\rho = T^{-1/4}, 
\) \(  
A = \sqrt{\frac{K}{m}}, 
\) \( 
\beta = \frac{4}{mA}\min\{\frac{\epsilon^2}{A^2},1\}
\) 
achieves task-averaged regret
\( 
R_m^T
= O\!\left(
\sqrt{\, m \, \mathrm{Var}\!\left(\{z^{*(t)}\}_{t=1}^T\right)}
\;\right)
\)  \( \text{as } T \to \infty,
\) 
where $z^{*(t)} = \nu(\phi^{*(t)})$ denotes the optimal sender strategy at task $t$, and $\mathrm{Var}\!\left(\{z^{*(t)}\}_{t=1}^T\right)$ denotes the empirical variance of $\{z^{*(t)}\}_{t=1}^T$.
\end{theorem}

\subsection{Partial Feedback Setting}

In our second setting, we assume that only a scalar loss value is revealed to the sender after interacting with the environment, as in the standard Bandit Linear Optimization (BLO) framework. Formally, at each round $i$ of task $t$ we observe loss $ \langle \nu(\phi^t_i), \mathbf{1}_{k^t_i}\rangle = \langle \ell^t_{i}, \mathbf{\phi}^t_{i} \rangle \in [0,1]$, where we defined $\ell^t_{i}=\mathbf{M}_{k^t_i,\cdot}$ denoting the ${k^t_i}^{th}$ row of $\mathbf{M}$. Before, introducing the algorithm, we first define the related tools.

\vspace{6pt}
\noindent 
We define the $b$-restricted space 
\( 
\bar \nu_b(\mathcal{P}) = \{ z \in \mathbb{R}^K : \pi_{z_1}(z) \le 1/(1+b) \} \subset \bar{\nu}(\mathcal{P})\},
\) 
where $z_{1} = \arg\min_{z \in \bar{\nu}(\mathcal{P})}\mathcal{R}(z)$ and 
\( 
\pi_{z_1}(z) = \inf\{ \lambda > 0 : z_1 + \lambda^{-1} (z - z_1) \in \bar{\nu}(\mathcal{P}) \}
\) 
is the Minkowski function. For such a task, we employ $\vartheta$-self-concordant barriers $\mathcal{R}$ as the regularizer in OMD iterations. A convex function 
\( 
\mathcal{R} : \operatorname{int}\big(\bar{\nu}(\mathcal{P})\big) \to \mathbb{R}
\) 
is called \emph{self-concordant} if it is $C^{3}$ and satisfies
\( 
\left| d^{3} \mathcal{R}(z)[h,h,h] \right|
\le
2 \left( d^{2}\mathcal{R}(z)[h,h] \right)^{3/2},
\) 
which relates the second and third-order differentials. In addition, it must satisfy
\( 
\left| d\mathcal{R}(z)[h] \right|
\le
\vartheta^{1/2}
\left( d^{2}\mathcal{R}(z)[h,h] \right)^{1/2},
\) 
which relates the first and second-order differentials.

\vspace{6pt}
\noindent 
We define the \emph{local} norm of a vector with respect to a given $z \in \bar \nu(\mathcal{P})$, as $\|h\|_z = \bigl( \langle h, h \rangle_z \bigr)^{1/2}$, where we define \( \langle g, h \rangle_z = g^{\top} \nabla^{2} \mathcal{R}(z)\, h\). Furthermore, we denote the dual local norm as, $\|g\|_{z_i,*}:=\sqrt{g^\top H_i^{-1}g}.$ Finally, using the local norms we define the  \emph{Dikin ellipsoid} of radius $r$ centered at $z\in  \bar \nu(\mathcal{P})$ where the sampling procedure takes place as,
\( 
W_r(z) = \{\, y \in \bar \nu(\mathcal{P}) : \| y - z \|_z < r \,\}.
\) 

\begin{algorithm}[ht]
\caption{CTOMD}
\label{alg:omd-tuning}

\begin{algorithmic}[1]
\State \textbf{Require:} inverse mapping $\nu^{\dagger}$, $\vartheta$-self-concordant $\mathcal{R}$, $\eta$, and $z^t_1$ 

    \For{$i = 1,\dots,m$}
         
        \State Let $\{\mathbf{e}_1, \ldots, \mathbf{e}_K\} \text{ and } \{v_1, \ldots, v_K\} $
be the set of eigenvectors and eigenvalues of $\nabla^{2}\mathcal{R}(\bar z^t_i)$.
        \State Choose $j''$ uniformly at random from $\{ 1,\cdots,K\}$ and $\varepsilon^t_i = \pm1$ with probability $1/2$
        \State Predict $ (\bar y_i)^t \gets (\bar z_i)^t + \varepsilon^t_iv^{-1/2}_{j''}\mathbf{e}_{j''}$
        \State $\{ (y^t_{i,j}\,, \lambda^t_{i,j})\}_{j \in [K+1]}\;
{\gets} \;{\mathrm{Carath\acute{e}odory \ }\!\bigl( (\bar y_i)^t,\, \nu(\mathcal{P})\bigr)}
$
        \State Draw $j' \in [K+1]$ with probabilities $\lambda^t_{i,1},\lambda^t_{i,2},\cdots,\lambda^t_{i,K+1}$
        \State Play $\phi^t_i \gets \nu^{\dagger}(y^t_{i,j'})$
        \State Suffer loss $\langle \ell^t_i,\phi^t_i\rangle \in \mathbb{R}$
        \State Define $\tilde{\ell^t_i} \coloneqq K\langle \ell^t_i,\phi^t_i\rangle\varepsilon^t_iv^{1/2}_{j''}\mathbf{e}_{j''}$
        \State $(\bar z_{i+1})^t \gets  \arg\min_{z \in \bar{\nu}(\mathcal{P})} \{\eta^t \langle \tilde{\ell}^t_i, z \rangle + D_R(z,(\bar z_i)^t) \}$
    \EndFor
    \State $\tilde\ell^t = \sum_{i=1}^m\tilde{\ell^t_i}$

\end{algorithmic}
\end{algorithm}

\vspace{6pt}
\noindent 
Conceptually, Algorithm~\ref{alg:omd-tuning} is the \emph{Bandit Online Linear Optimization} algorithm of the \citet{CompetingIntheDark}, applied over $\bar{\nu}(\mathcal{P})$, and augments them with a Carath\'eodory oracle to map iterates back to the signaling space, whose existence was first implied in \citet{OptimalRatesOBP}. One can observe that another difference between our algorithm and the one proposed in \citet{CompetingIntheDark} is that we iterate over policies using an OMD procedure, whereas they employ an FTRL-based formulation. However, as noted in the same work, these two policy update methods yield identical iterates when self-concordant regularizers are used.  

\vspace{6pt}
\noindent
At round $i$ of task $t$, Algorithm~\ref{alg:omd-tuning} maintains an interior point $ (\bar z_i)^t\in\bar\nu(\mathcal P)$ and only uses the
scalar loss $\langle \ell_i^t,\phi_i^t\rangle=\nu(\phi_i^t)^\top \mathbf{1}_{k_i^t}\in[0,1]$ after playing $\phi_i^t$. To obtain a low-variance estimator while staying feasible, Algorithm~\ref{alg:omd-tuning} explores inside the \emph{Dikin ellipsoid} induced by the
$\vartheta$-self-concordant barrier $\mathcal R$.
Let $\nabla^2\mathcal R(\bar z_i^t)$ have eigenpairs $\{(\mathbf e_j,v_j)\}_{j=1}^K$.
Sampling $j''\sim\mathrm{Unif}([K])$ and $\varepsilon_i^t\in\{\pm1\}$ uniformly, the algorithm forms the perturbed point
\( 
(\bar y_i)^t \;=\; (\bar z_i)^t+\varepsilon_i^t\,v_{j''}^{-1/2}\mathbf e_{j''},
\) 
which lies in $W_1(\bar z_i^t)\subset\bar\nu(\mathcal P)$.
Geometrically, the eigenvectors $\mathbf e_j$ are the principal axes of the Dikin ellipsoid, and the scaling
$v_j^{-1/2}$ moves one unit in the \emph{local} norm, producing exploration that is adapted to the curvature of
$\mathcal R$.

\vspace{6pt}
\noindent
Although, the sampled $(\bar y_i)^t$ is a valid point in the convex hull $\bar\nu(\mathcal P)$, the sender must play an actual
$\phi_i^t\in\mathcal P$.
Then, the Carath\'eodory oracle decomposes $(\bar y_i)^t$ as
\( 
\{(y_{i,j}^t,\lambda_{i,j}^t)\}_{j\in[K+1]}
\) 
and Algorithm~\ref{alg:omd-tuning} samples $j'$ with $\mathbb{P}(j'=j)=\lambda_{i,j}^t$, then plays
$\phi_i^t=\nu^\dagger(y_{i,j'}^t)$.
This guarantees the \emph{implementation-in-expectation},
\( 
\mathbb{E}\big[\nu(\phi_i^t)\mid (\bar y_i)^t\big]=(\bar y_i)^t,
\) 
so linear losses evaluated at the played policy match the losses at the mapped point in expectation.

\vspace{6pt}
\noindent
From the scalar observation $\langle \ell_i^t,\phi_i^t\rangle$, Algorithm~\ref{alg:omd-tuning} forms the estimator
\( 
\tilde\ell_i^t
:=
K\,\langle \ell_i^t,\phi_i^t\rangle\,\varepsilon_i^t\,v_{j''}^{1/2}\mathbf e_{j''}.
\) 
along with the Carath\'eodory implementation, this yields an unbiased estimator of the loss direction. Finally, Algorithm~\ref{alg:omd-tuning} performs the mirror descent step on $\bar\nu(\mathcal P)$ with barrier regularizer $\mathcal R$. We refer readers to
\citep{CompetingIntheDark} for more detailed discussion on the core algorithm.

\begin{algorithm}[t]
\caption{Partial-Feedback Meta-Persuasion}
\label{alg:omd}

\begin{algorithmic}[1]

\State \textbf{Input:} compact $\bar{\nu}({\mathcal{P})} \subset \mathbb{R}^K$, 
meta-hyperparameters $\alpha>0$ , 
$\mathcal{G} \subset \mathbb{R}^2$ over $(\eta,b)$, $\mathrm{OPT}_b $.

\For{$g = (\eta,b) \in \mathcal{G}$}
    \State $\mathbf{z}^{1,(g)} \gets  \arg\min_{\mathbf{z}\in\bar{\nu}({\mathcal{P})}} \mathcal{R}(\phi)$ 
\EndFor

\State $\mathbf{p}_1 \gets \mathbf{1}_{|\mathcal{G}|}/|\mathcal{G}|$

\For{\textbf{task} $t = 1,\ldots,T$}

    \State sample $g^t = (\eta^t,b^t) \sim \mathbf{p}^t$ from $\mathcal{G}$

    \State $\tilde{\ell}^t \gets 
    \mathrm{CTOMD}(\mathbf{z}^{t,(g^t)})$

    \For{$g = (\eta,b) \in \mathcal{G}$}

        \State $\mathbf{z}^{t+1,(g)} \gets 
        \frac{1}{t}\sum_{s=1}^t \mathrm{OPT}_b(\tilde{\ell}^s)$

        \State $\mathbf{p}^{t+1}(g) \gets 
        \mathbf{p}^{t+1}(g)\exp\!\big(-\alpha 
        U^t(\mathbf{z}^{t,(g)},g)\big)$

    \EndFor

    \State $\mathbf{p}^{t+1} \gets 
    \mathbf{p}^{t+1}/\|\mathbf{p}^{t+1}\|_1$

\EndFor

\end{algorithmic}
\end{algorithm}
\vspace{6pt}
\noindent
We now explain how we leverage repeated \emph{persuasion tasks} to tune the bandit learner in
Algorithm~\ref{alg:omd-tuning}.
The inner-loop objective of the sender within each task is to perform well against the best fixed signaling
rule for that task, but the similarity of the task compared to previous tasks can vary across $t$.
In particular, the performance of Algorithm~\ref{alg:omd-tuning} depends critically on two quantities: (i) the \emph{within-task step size}
$\eta$, and (ii) the \emph{boundary offset} $b$ which affects the barrier geometry and the estimator variance. Rather than choosing $(\eta,b)$ a priori, Algorithm~\ref{alg:omd} learns these hyperparameters online
across tasks using an experts-style meta-procedure, following the meta-algorithm given in \citet{MetaLearningBandits}.

\vspace{6pt}
\noindent
To learn a better initialization of the parameters, the meta-algorithm utilizes the the cumulative estimated loss vector
\( 
\tilde{\ell}^t := \sum_{i=1}^m \tilde{\ell}_i^t,
\) 
outputted by Algorithm~\ref{alg:omd-tuning} at the end of task $t$, which serves as an unbiased proxy for the (unknown) cumulative loss in the \emph{loss} space. From $\tilde\ell^t$ we form the $b$-restricted optimum-in-hindsight
\( 
\mathrm{OPT}_{b}(\tilde\ell^t)\in\arg\min_{z\in\bar\nu_b(\mathcal{P})}\langle \tilde\ell^t, z\rangle,
\) summarizes the \emph{task-specific best response} of the sender in the \emph{loss}
space. It can be intuitively seen that when tasks are similar these optima tend to cluster; when tasks are more heterogeneous, they tend to be more dispersed.

\vspace{6pt}
\noindent
To exploit task similarity, Algorithm~\ref{alg:omd} maintains, \emph{for each hyperparameter pair}
$g=(\eta,b)\in\mathcal{G}$, a meta-initialization $z^{t,(g)}\in\bar\nu(\mathcal{P})$ given by the running average
of past optima at that offset
\( 
\bar z^{t+1,(g)} \gets \frac{1}{t}\sum_{s=1}^t \mathrm{OPT}_{b}(\tilde\ell^{s}).
\) 
This can be seen as a principled warm-start, as if the task-wise optima $\mathrm{OPT}_{b}(\tilde\ell^{t})$ concentrate around a common
center, then $\bar z^{t,(g)}$ quickly approaches that center and the divergence term shrinks, improving the average regret.

\vspace{6pt}
\noindent
Thus, Algorithm~\ref{alg:omd} discretizes a continuous admissible range of $(\eta,b)$ into a finite grid
$\mathcal{G}\subset\mathbb{R}_{>0}\times(0,1)$.
Each $g=(\eta,b)\in\mathcal{G}$ is treated as an \emph{expert}.
At the beginning of task $t$, the meta-learner samples $g^t=(\eta^t,b^t)\sim p^t$ and runs Algorithm~\ref{alg:omd} with the
corresponding initialization $z^{t,(g^t)}$.
After observing $\tilde\ell^t$, the meta-learner can evaluate, for every $g=(\eta,b)\in\mathcal{G}$, the
task-level upper bound
\[
U^{t}\!\big(z^{t,(g)},g\big)
\;:=\;
\frac{1}{\eta}\,D_{\mathcal R}\!\Big(\mathrm{OPT}_{b}(\tilde\ell^t)\,\Big\|\, z^{t,(g)}\Big)
\;+\;
(32K^2\eta + b)\,m.
\]
Finally, the distribution over \emph{experts} is updated by multiplicative weights
followed by normalization.
Next, we provide the task averaged regret guaranteed by  Algorithm \eqref{alg:omd} in the following theorem.

\begin{theorem}\label{thm:obp-bandit-meta1}
For each $b$, define the constants
\( 
D_b^2 := \max_{x,y\in\mathcal{\bar \nu}_b}D_{\mathcal R}(x\|y),
\) \( 
S_b := \max_{x\in\mathcal{\bar \nu}_b}\|\nabla^2\mathcal{R}(x)\|_2,
\) \( 
\mathsf{K} := \max_{x,y\in\mathcal{\bar \nu}}\|x-y\|_2,
\) 
and the barrier-divergence, by
\( 
\widehat V_b^2
\;:=\;
\min_{z\in\mathcal{\bar \nu}}\;
\mathbb{E}\!\left[\frac{1}{T}\sum_{t=1}^T D_{\mathcal R}\!\big(\mathrm{OPT}_b(\hat\ell^t)\,\|\,z\big)\right].
\)
Then, choosing $\eta = \frac{\widehat V_b}{2\sqrt{32}\mathsf{K}m}$, there exist a grid size
$k=\widetilde O(D_{\underline b}^2K\sqrt{mT})$ and a meta step-size $\alpha$ such that for Algorithm~\ref{alg:omd}, the expected task-averaged regret satisfies
\begin{equation}\label{eq:meta-simplified1}
R_m^T
\;\le\;
\widetilde O\!\left(\frac{D_{\underline b}K m}{T^{1/4}}+\frac{S_{\underline b}\mathsf{K}^2K\sqrt m}{D_{\underline b}T^{3/4}}\right)
\;+\;
\min_{b\in[\underline b,\bar b]}
\Big(4K\widehat V_b\sqrt{2m}+bm\Big).
\end{equation}
\end{theorem}
\noindent
Following Theorem~\ref{thm:obp-bandit-meta1}, we present Corollary~\ref{cor:ball-barrier-obp} to show how task similarity improves the averaged regret.  

\begin{corollary}[Corollary 5.1 in \citet{MetaLearningBandits}]\label{cor:ball-barrier-obp} 
Assume that the feasible region is
\( 
\bar\nu(\mathcal P) = \{z\in\mathbb{R}^K:\ \|z\|_2\le 1\},
\) 
and take the self-concordant barrier
\( 
\mathcal R(z) = -\ln\bigl(1-\|z\|_2^2\bigr).
\) 
Let the $b$-restricted set be the
\( 
\bar\nu_b(\mathcal P) \;:=\; \{z\in\mathbb{R}^K:\ \|z\|_2\le 1-b\},
\; b\in(0,1).
\) 
Then, running Algorithm~\ref{alg:omd}, the expected task-averaged regret satisfies
\[
R_m^T
\;\le\;
\widetilde O\!\left(\frac{K\,m^2}{T^{1/4}}\right)
\;+\;
\min_{1/m\le b\le 1/\sqrt m}
\left\{
4K\,\mathbb{E}\!\left[\sqrt{\,2m\ln\left(\frac{1-\|\hat{\bar z}^{(b)}\|_2^2}{2b-b^2}\right)}\right]
\;+\;
bm
\right\},
\]
where
\( 
\hat{\bar z}^{(b)} := \frac{1}{T}\sum_{t=1}^T \mathrm{OPT}_b(\tilde\ell^t).
\) 
Moreover, in this geometry the constrained optimizer has the closed form
\( 
\mathrm{OPT}_b(\tilde\ell^t)
\in
\arg\min_{\|z\|_2\le 1-b}\langle \tilde\ell^t,z\rangle
\;=\;
-(1-b)\frac{\tilde\ell^t}{\|\tilde\ell^t\|_2}
\) .
\end{corollary}

\noindent
It can be seen that the Corollary \ref{cor:ball-barrier-obp} captures task similarity, as the quantity $\hat{\bar z}^{(b)}$ is an average of estimated loss directions across tasks.
If tasks are similar, these directions align, and thus $\|\hat{\bar z}^{(b)}\|_2$ is close to $1-b$.
Then $1-\|\hat{\bar z}^{(b)}\|_2^2 \approx 2b-b^2$, making the logarithmic term close to $1$ and hence
$\widehat V_b \approx 0$.
Consequently, as $T\to\infty$ the dominant term becomes $bm$; choosing $b=1/m$ yields constant asymptotic
task-averaged regret.
If tasks are dissimilar, the normalized directions cancel out and $\|\hat{\bar z}^{(b)}\|_2$ is small, making the
logarithmic term large and recovering a worst-case scaling.

\section{The Meta-Learning for Markov Persuasion Processes} \label{sec:MPP}

For the MPP setting, we first define the occupancy measures. Given a transition function $P$, a signaling policy $\phi$, and a prior function $\mu$, 
the occupancy measure induced by $(P,\phi,\mu)$ is a vector 
$q^{P,\phi,\mu} \in [0,1]^{|X \times \Omega \times A \times X|}$ 
whose entries are defined as follows. 
For every $x \in X_k$, $\omega \in \Omega$, $a \in A$, and 
$x' \in X_{k+1}$ with $k \in \mathcal{K}$, we define
\[
q^{P,\phi,\mu}(x,\omega,a,x') 
:= 
\mathbb{P}\Bigl((x_k,\omega_k,a_k,x_{k+1}) = (x,\omega,a,x') \,\big|\, P,\phi,\mu\Bigr),
\]
which is the probability that the next state is $x'$ after playing action $a$ in state $x$ when the realized outcome is $\omega$, under
transition function $P$, signaling policy $\phi$, and prior function $\mu$. Moreover, we let:
\[
q^{P,\phi,\mu}(x,\omega,a)\! := \!\! \!\! \!\!\sum_{x'\in X_{k+1}} \!\!\!\!q^{P,\phi,\mu}(x,\omega,a,x'), \quad
q^{P,\phi,\mu}(x,\omega)\! := \!\!\sum_{a\in A} \! q^{P,\phi,\mu}(x,\omega,a),\quad
q^{P,\phi,\mu}(x) \! := \!\!  \sum_{\omega\in\Omega}  \! q^{P,\phi,\mu}(x,\omega).
\]
\noindent
As it is the case in standard MDPs, a valid occupancy measure $q \in [0,1]^{|X \times \Omega \times A \times X|}$ induces a transition function $P^{q}$, a signaling policy $\phi^{q}$, and prior function  $\mu^{q}$ defined as follows:
\[
P^{q}(x' \mid x,\omega,a) := \frac{q(x,\omega,a,x')}{q(x,\omega,a)}, \qquad
\phi^{q}(a \mid x,\omega) := \frac{q(x,\omega,a)}{q(x,\omega)},
\qquad
\mu^{q}(\omega \mid x) := \frac{q(x,\omega)}{q(x)}.
\]
We denote by $\mathcal{Q} \subseteq [0,1]^{|X \times \Omega \times A \times X|}$ the set of all the valid occupancy measures of an MPP.  The following lemma characterizes the set of \textit{valid} occupancy measures and it is a generalization to the MPP setting.

\begin{lemma}[Lemma 1, \citet{MarkovPersuasionScratch2025}]
    A vector $q\in[0,1]^{|X\times\Omega\times A\times X|}$ is a valid occupancy measure of an MPP if and only if it holds:
\[
\left\{
\begin{array}{ll}
\text{1 -}~~ \displaystyle\sum_{x\in X_k}\sum_{\omega\in\Omega}\sum_{a\in A}\sum_{x'\in X_{k+1}} q(x,\omega,a,x') = 1
& \forall k\in\mathcal{K} \\[1.2ex]
\text{2 -}~~ \displaystyle\sum_{x'\in X_{k-1}}\sum_{\omega\in\Omega}\sum_{a\in A} q(x',\omega,a,x) = q(x)
& \forall k\in[1\ldots L-1],\forall x\in X_k \\[1.2ex]
\text{{3 -}}~~ P^{q} = P \\[0.6ex]
\text{{4 -}}~~ \mu^{q} = \mu,
\end{array}
\right.
\]
where $P$ is the transition function of the MPP and $\mu$ its prior function, while $P^{q}$ and $\mu^{q}$ are the transition and prior
functions, respectively, induced by occupancy measure $q$.
\end{lemma}

\vspace{6pt}
\noindent
Our objective is to construct algorithms that produce sequences of signaling policies 
$\{\phi_i^t\}$ which maximize the sender’s cumulative reward over $m$ episodes across $T$ tasks, 
while ensuring that violations of the \emph{persuasiveness} constraints remain controlled. Crucially, we do not attempt to enforce that each policy $\phi_i^t$ be persuasive at every episode $t$, 
as such a guarantee is unattainable since the sender does not have access to the receiver types’ reward distributions 
\citep{MarkovPersuasionScratch2025}.  Accordingly, our goal is to design algorithms that achieve vanishing average 
\emph{regret} together with vanishing average \emph{constraint violations}. We now introduce the benchmark offline optimization problem that the sender would solve for each task $t \in [T]$:
\begin{align}
\max_{q^t\in\mathcal{Q}} \quad
& \sum_{x\in X}\sum_{\omega\in\Omega}\sum_{a\in A} q^t(x,\omega,a)\, u^{s,t}(x,\omega,a)
\label{eq:baseline_obj} \tag{1a}\\
\text{s.t.}\quad
& \sum_{\omega\in\Omega} q^t(x,\omega,a)\bigl(u^{r,t}(x,\omega,a)-u^{r,t}(x,\omega,a')\bigr)\ge 0
\notag\\
& \qquad \forall x\in X,\ \forall \omega\in\Omega,\ \forall a\in A,\ \forall a'\in A\setminus\{a\}.
\label{eq:baseline_pers} \tag{1b}
\end{align}
\vspace{6pt}
\noindent 
It follows that Problem~\eqref{eq:baseline_obj} determines the optimal occupancy measure—and, by correspondence, the associated optimal signaling policy—subject to the persuasiveness constraints in~\eqref{eq:baseline_pers}. Since, in the MPP setting, the players’ rewards are stochastic, we let 
$u^{s,t},u^{r,t} \in [0,1]^{|X \times \Omega \times A|}$ denote the random vectors whose components represent the mean sender and receiver types' rewards. We define the optimal benchmark value for task $t$ as \( \mathrm{OPT}^t := (u^{s,t})^\top q^t_\star,\) where $q^t_\star \in \mathcal{Q}$ is an optimal solution to Problem~\eqref{eq:baseline_obj}. Throughout, we denote by $\phi^t_\star$ an optimal signaling policy for task $t \in [T]$, induced by $q^t_\star$, i.e.,
\( 
\phi^t_\star := \phi^{q^t_\star}.
\) 
We evaluate learning performance using two standard metrics. 
The first metric is the task-averaged cumulative regret $R_m^T$, defined as
\[
R_m^T 
:= 
\frac{1}{T}\sum_{t \in [T]}
\Big(
m \cdot \mathrm{OPT}^t 
- 
\sum_{i \in [m]} (u^{s,t})^\top q_i^t
\Big)
=
\frac{1}{T}
\sum_{t \in [T]}
\sum_{i \in [m]}
(u^{s,t})^\top (q^t_\star - q^t_i),
\]
where $q^t_i := q^{P^t,\phi^t_i,\mu^t}$ denotes the occupancy measure induced by the signaling policy $\phi^t_i$ with known prior function $\mu^t$ and transition function $P^t$ for task $t \in [T]$. The second metric is the task-averaged cumulative violation $V_m^T$, which measures deviations from persuasiveness. Since the sender does not observe the receivers’ reward distributions or \emph{types}, unlike in the OBP framework, we evaluate violations cumulatively. Formally,
\[
V_m^T 
:=
\frac{1}{T}
\sum_{t \in [T]}
\sum_{i \in [m]}
\sum_{x \in X}
\sum_{\omega \in \Omega}
\sum_{a \in A}
q_i^t(x,\omega,a)
\Bigl(
u^{r,t}\bigl(x,\omega,b^{\phi^t_i}(a,x)\bigr)
-
u^{r,t}(x,\omega,a)
\Bigr),
\]
where $b^{\phi^t_i}(a,x)$ denotes the receiver’s best response under policy $\phi^t_i$ upon receiving signal $a \in A$ in state $x \in X$. We emphasize that, the violation metric is only considered for the MPP part. As in the OBP part, all of the signaling scheme's $\phi^t_i \in \mathcal{P}$ are already persuasive. Accordingly, our objective is to design learning algorithms that generate signaling policies $\phi_i^t$ while ensuring that both regret and constraint violations grow sublinearly in the number of episodes $m$. Although the persuasiveness constraints may be violated in some episodes, such violations occur only in a vanishing fraction of rounds. Consequently, in the long run, it remains optimal for receivers to follow, i.e., be obedient to, the sender’s recommendations.

\subsection{Estimators and Confidence Bounds for Meta-Learning in Markov Persuasion Processes}

Before presenting the learning algorithms, we first construct estimators and confidence sets for the stochastic components of the MPP model, namely the transition dynamics, the prior distribution, the sender's rewards, and the receiver types' rewards.

\vspace{6pt}
\noindent
For each task \(t \in [T]\) and episode index \(i \in [m]\), we introduce empirical visitation counts. Specifically, for every
\((x,\omega,a,x') \in X \times \Omega \times A \times X\), let
\( 
N_i^t(x,\omega,a,x')
:=
\sum_{j=1}^{i-1}
\mathbf{1}\!\Big\{(x_j^t,\omega_j^t,a_j^t,x_{j+1}^t)=(x,\omega,a,x')\Big\}.
\) 
Similarly, we define the lower-order counts by marginalization:
\( 
N_i^t(x,\omega,a)
:=
\sum_{x' \in X_{k(x)+1}} N_i^t(x,\omega,a,x'),
\) \(
N_i^t(x,\omega)
:=
\sum_{a \in A} N_i^t(x,\omega,a),
\;
N_i^t(x)
:=
\sum_{\omega \in \Omega} N_i^t(x,\omega).
\) 
Thus, each counter records how many times the corresponding coordinate has been observed strictly before episode \(i\) in task \(t\).

\vspace{6pt}
\noindent
For every scalar entry of the unknown primitives that we estimate in task \(t\in[T]\)—namely,
a transition probability \(P^t(x' \mid x,\omega,a)\),
a prior entry \(\mu^t(\omega \mid x)\),
or a reward entry \(u^{s,t}(x,\omega,a)\) or \(u^{r,t}(x,\omega,a)\)—
we use the same meta-learning template. For a fixed coordinate \(c\), let \(n:=N_i^t(c)\) denote the number of observations of that coordinate collected up to, but excluding, episode \(i\) in task \(t\), and let \(\bar Q_i^t(c)\) denote the corresponding within-task empirical estimator. We set
\( 
\bar Q_i^t(c) := 0
\; \text{whenever}\;
N_i^t(c)=0
\)
by convention.

\vspace{6pt}
\noindent
To leverage information gathered from previous tasks, we define an across-task meta-mean using past tasks. For each \(\tau \in [t-1]\), let \(\bar Q^\tau(c)\) denote the terminal within-task empirical estimator of coordinate \(c\) in task \(\tau\), computed from all observations collected up to episode \(m\), and let
\( 
N_\tau(c) := N_m^\tau(c)
\) 
be the corresponding terminal count. Since a coordinate may be unobserved in some tasks, we average only over tasks in which that coordinate has been observed. Formally, define the within-task empirical mean for the new task and the task-wise terminal empirical means for past tasks as
\[
\bar Q_i^{\,t+1}
:=
\frac{1}{\max\{1,n\}}\sum_{j=1}^{n} Q_j^{t+1},
\qquad
\bar Q^\tau
:=
\frac{1}{\max\{1,N_\tau\}}\sum_{j=1}^{N_\tau} Q_j^\tau,
\quad \forall \tau \in [t].
\]
If we have \( N_m^\tau(c) = 0 \), we set \( \bar Q^\tau(c) = 0\) by convention. Define the active-task indicator and active-task count by
\( 
I_\tau(c) := \mathbf{1}\{N_\tau(c)>0\},
\) \( 
M_{t-1}(c) := \sum_{\tau=1}^{t-1} I_\tau(c).
\) 
Then, the across-task meta-mean is
\( 
\bar Q_G^{\,t-1}(c)
:=
\frac{1}{\max\{1,M_{t-1}(c)\}}
\sum_{\tau=1}^{t-1}
I_\tau(c)\,\bar Q^\tau(c).
\) 

\vspace{6pt}
\noindent
To make proposed estimator well-defined for all values of the count and similarity parameter, including the degenerate case \(N_i^t(c)=0\) and \(\kappa=0\), we define the weights piecewise as
\[
w_\kappa(n)
:=
\begin{cases}
1, & \kappa = 0,\\[4pt]
\dfrac{n}{n+\kappa}, & \kappa > 0,
\end{cases}
\qquad
\bar w_\kappa(n)
:=
1-w_\kappa(n)
=
\begin{cases}
0, & \kappa = 0,\\[4pt]
\dfrac{\kappa}{n+\kappa}, & \kappa > 0.
\end{cases}
\]
Using these weights, the meta-estimator for coordinate \(c\) at episode \(i\) of task \(t\) is defined by
\[
\hat Q_i^t(c)
:=
w_\kappa\!\big(N_i^t(c)\big)\,\bar Q_i^t(c)
+
\bar w_\kappa\!\big(N_i^t(c)\big)\,\bar Q_G^{\,t-1}(c).
\]
Hence, when \(\kappa=0\), the estimator reduces exactly to the within-task empirical estimator, while for \(\kappa>0\) it interpolates between the within-task estimate and the across-task meta-mean.

\vspace{6pt}
\noindent
Such an idea of estimator can be traced back to empirical Bayesian estimation of different types of distributions \citep{Shrinkage_est, Shrinkage_book}. Its main advantage is that it preserves \emph{within-task consistency}: for any fixed $\kappa$, as $n \to \infty$, we have $w_t \to 1$, and therefore $\hat Q_i^t - \bar Q_i^t \to 0$. At the same time, when $n$ is small, the estimator can substantially reduce variance by borrowing strength from previous tasks.

\vspace{6pt}
\noindent
Moreover, the proposed estimator interpolates smoothly between pure within-task learning and aggressive transfer across tasks. When \(\kappa\) is small, the weight on the current task is close to one; when \(\kappa\) is large, the estimator places more weight on the across-task meta-mean. A natural similarity parameter is of the form
\( 
\kappa \asymp \frac{\text{within-task noise level}}{\text{across-task variability}}.
\) 
To formalize this construction, we first introduce the following assumption, which is standard in the meta-learning and Bayesian persuasion literatures.

\vspace{6pt}
{\hypersetup{linkcolor=black}
\noindent\textbf{Assumption 4.1}
}\label{rem:known-variance0}
For every coordinate $(x,\omega,a,x') \in X \times \Omega \times A \times X$
and $(x,\omega,a) \in X \times \Omega \times A$,
we assume that the sender knows the within-task variance $\sigma^2$
and the across-task variance $\iota^2$ associated with the task parameters
\(P^t(x' \mid x,\omega,a)\),
 \(\mu^t(\omega \mid x)\),
 \(u^{s,t}(x,\omega,a)\), \(u^{r,t}(x,\omega,a)\) .

\vspace{6pt}
\noindent 
In many meta-learning work, the observation noise variance is assumed to be known or estimated offline from abundant data, while the primary goal is to estimate task-specific means \citep{KnownVarianceBanditGeneral, MetaGaussianKnownVariancebandit}.
By contrast, in classical Bayesian persuasion models, the sender typically assumed to know the prior over the state and the payoff functions, and thus the analogous ``mean'' and ``variance'' parameters are not themselves learned, \citep{BasTopcu, Info_theory_Comm}.
Our repeated setting sits between these extremes: we learn task-dependent quantities online while using $\sigma^2$ and $\iota^2$ only as quantities that summarize within-task noise and cross-task similarity through $\kappa$.

\vspace{6pt}
\noindent
In our setting, the precise meaning of the within-task noise depends on the primitive being estimated, as formalized next. For every transition coordinate \((x,\omega,a,x')\), the task-dependent parameter
\( 
P^t(x' \mid x,\omega,a)
\)
is drawn i.i.d.\ across tasks with mean
\( 
P_G(x' \mid x,\omega,a)
\) 
and across-task variance bounded by \(\iota_P^2\). Likewise, for every prior coordinate \((x,\omega)\), the task-dependent parameter
\( 
\mu^t(\omega \mid x)
\) 
is drawn i.i.d.\ across tasks with mean
\( 
\mu_G(\omega \mid x)
\) 
and across-task variance bounded by \(\iota_\mu^2\). 

\vspace{6pt}
\noindent
Within a fixed task \(t\), the quantities \(P^t\) and \(\mu^t\) are fixed. The randomness within the task comes from repeated visits to the corresponding coordinates:
\begin{itemize}
    \item whenever \((x,\omega,a)\) is visited, the next state is sampled from the categorical distribution \(P^t(\cdot\mid x,\omega,a)\);
    \item whenever \(x\) is visited, the realized outcome is sampled from the categorical distribution \(\mu^t(\cdot\mid x)\).
\end{itemize}
Equivalently, for every fixed entry \(x'\) of \(P^t(\cdot\mid x,\omega,a)\) and every fixed entry \(\omega\) of \(\mu^t(\cdot\mid x)\), the associated one-time observation is Bernoulli with mean equal to that entry's probability. Therefore,
\[
\operatorname{Var}\!\big(\mathbf{1}\{x_{k+1}=x'\}\mid P^t, x,\omega,a\big)
\le \frac14,
\qquad
\operatorname{Var}\!\big(\mathbf{1}\{\omega_k=\omega\}\mid \mu^t, x\big)
\le \frac14.
\]
For reward coordinates, conditioned on the task means \(u^{s,t}(x,\omega,a)\) and \(u^{r,t}(x,\omega,a)\), the within-task observations are independent across episodes, take values in \([0,1]\), and satisfy
\[
\mathbb{E}\!\big[u_i^{s,t}(x,\omega,a)\mid u^{s,t}(x,\omega,a)\big]
=
u^{s,t}(x,\omega,a),
\qquad
\operatorname{Var}\!\big(u_i^{s,t}(x,\omega,a)\mid u^{s,t}(x,\omega,a)\big)
\le \sigma_s^2,
\]
\[
\mathbb{E}\!\big[u_i^{r,t}(x,\omega,a)\mid u^{r,t}(x,\omega,a)\big]
=
u^{r,t}(x,\omega,a),
\qquad
\operatorname{Var}\!\big(u_i^{r,t}(x,\omega,a)\mid u^{r,t}(x,\omega,a)\big)
\le \sigma_r^2.
\]
Finally, due to variance bounds of the observable parameters, for the similarity parameters we choose
\( 
\kappa_P = \frac{1/4}{\iota_P^2},
\;
\kappa_\mu = \frac{1/4}{\iota_\mu^2},
\;
\kappa_{u^s} = \frac{\sigma_s^2}{\iota_s^2},
\;
\kappa_{u^r}= \frac{\sigma_r^2}{\iota_r^2}.
\) 

\vspace{6pt}
\noindent
We now specialize the generic meta-estimator to each primitive. For transition kernel estimator we have 
\[
\hat P_i^t(x' \mid x,\omega,a)
:=
w_{\kappa_P}\!\big(N_i^t(x,\omega,a)\big)\,
\frac{N_i^t(x,\omega,a,x')}{\max\{1,N_i^t(x,\omega,a)\}}
+
\bar w_{\kappa_P}\!\big(N_i^t(x,\omega,a)\big)\,
\bar P_G^{\,t-1}(x' \mid x,\omega,a),
\]
where we define
\( 
\bar P_G^{\,t-1}(x' \mid x,\omega,a)
:=
\frac{1}{\max\{1,M_{t-1}(x,\omega,a)\}}
\sum_{\tau=1}^{t-1}
\mathbf 1\{N_m^\tau(x,\omega,a)>0\}\,
\frac{N_m^\tau(x,\omega,a,x')}{\max\{1,N_m^\tau(x,\omega,a)\}}.
\) 

\vspace{6pt}
\noindent
For prior distribution estimator, we have\[
\hat\mu_i^t(\omega \mid x)
:=
w_{\kappa_\mu}\!\big(N_i^t(x)\big)\,
\frac{\sum_{j=1}^{i-1}\mathbf 1\{x_j^t=x,\ \omega_j^t=\omega\}}{\max\{1,N_i^t(x)\}}
+
\bar w_{\kappa_\mu}\!\big(N_i^t(x)\big)\,
\bar\mu_G^{\,t-1}(\omega \mid x),
\]
where we define, 
\( 
\bar\mu_G^{\,t-1}(\omega \mid x)
:=
\frac{1}{\max\{1,M_{t-1}(x)\}}
\sum_{\tau=1}^{t-1}
\mathbf 1\{N_m^\tau(x)>0\}\,
\frac{\sum_{j=1}^{m}\mathbf 1\{x_j^\tau=x,\ \omega_j^\tau=\omega\}}{\max\{1,N_m^\tau(x)\}}.
\) 

\vspace{6pt}
\noindent
For reward estimators, for \(\ell \in \{s,r\}\)  we have:
\[
\!\hat u_{i}^{\ell,t}(x,\!\omega,\!a)
\!\!:=\!
w_{\kappa_{u^\ell}}\!\big(\!N_i^t(x,\omega,a)\big)
\frac{\!\sum_{j=1}^{i-1} \!u_j^{\ell,t}(x,\!\omega,\!a\!)\,\!\mathbf 1\{x_j^t\!=\!x, \omega_j^t\!=\!\omega, a_j^t\!=\!a\!\}\!\!}{\max\{1,N_i^t(x,\omega,a)\}}
+
\bar w_{\kappa_{u^\ell}}\!\big(\!N_i^t(x,\!\omega,\!a)\big)
\bar u_{G}^{\ell,t-1}\!(x,\!\omega,\!a)
\]
where, 
\( 
\bar u_{G}^{\ell,t-1}(x,\omega,a)
\!:= \!
\frac{1}{\max\{1,M_{t-1}(x,\omega,a)\}}
\sum_{\tau=1}^{t-1}
\mathbf 1\{N_m^\tau(x,\omega,a)>0\}
\frac{\sum_{j=1}^{m} \!u_j^{\ell,\tau}(x,\omega,a)\mathbf 1\{x_j^\tau=x, \omega_j^\tau=\omega, a_j^\tau=a\}}{\max\{1,N_m^\tau(x,\omega,a)\}}.
\) 

\vspace{10pt}
\noindent
For notational simplicity in the sequel, and in particular in the Appendices, we take
\( 
\kappa_P=\kappa_\mu=\kappa_{u^s}=\kappa_{u^r}=\kappa,
\) 
while keeping in mind that the interpretation of the within-task noise differs across primitives.

\vspace{6pt}
\noindent
The corresponding confidence radii are denoted by
\(
\epsilon_i^t(x,\omega,a),\)
\( \zeta_i^t(x),\)  \(
\xi_{i}^{s,t}(x,\omega,a),\)
\(\xi_{i}^{r,t}(x,\omega,a) \). \( 
\)
We provide the concentration proofs in Appendix~\ref{app:meta-conf-bdd}. For a confidence parameter \(\delta \in (0,1)\), we define the good event in which all confidence bounds hold as \(\mathcal E(\delta)\). With the updated concentration lemmas in Appendix~\ref{app:meta-conf-bdd}, the event \(\mathcal E(\delta)\) holds with probability at least \(1-8\delta\), in both the full-feedback and partial-feedback settings.

\begin{algorithm}[t]
\caption{Full-Feedback Meta Optimistic Persuasive Policy Search (Full-Meta-OPPS)}
\label{alg:OPPS-full}
\begin{algorithmic}[1]
\Require $X$, $A$, $\Omega$, $m$, $T$, confidence parameter $\delta \in (0,1)$
\For{task $t=1, \dots T$}
    \For{iteration $i = 1,\ldots,m$}
      \State Update all estimators $\hat{P}_i^t, \hat{\mu}_i^t, \hat{u}_s^{s,t}, \hat{u}_i^{r,t}$ and bounds
      $\epsilon_i^t, \zeta_i^t, \xi_i^{s,t}, \xi_i^{r,t}$ given new observations
      \State $\hat{q}_i^t \gets$ Solve \textsc{Meta-Opt-Opt} 
      \State $\phi^t_i \gets \phi^{\hat{q}_i^t}$
      \State Run Protocol~\ref{protoc2} by committing to $\phi_i^t$
      \State Observe \textit{full} feedback from Protocol~\ref{protoc2}
    \EndFor
\EndFor
\end{algorithmic}
\end{algorithm}

\subsection{Full Feedback Setting}

In this section, we utilize the Optimistic Persuasive Policy Search, Algorithm \ref{alg:OPPS-full},
proposed by \citep{MarkovPersuasionScratch2025} to learn and solve 
Problem~\ref{eq:baseline_obj}. 
At each episode of every task, the algorithm solves a linear optimization problem, 
referred to as \textsc{Meta-Opt-Opt} (see \ref{app:meta-optopt}). 
This program constitutes the meta-learning variant of the original 
\textsc{Opt-Opt} formulation in \citep{MarkovPersuasionScratch2025}, 
where the optimization is performed using the linear constraints set by our meta-estimators and their confidence bounds. Simply, \textsc{Meta-Opt-Opt} is a linear program whose goal is to maximize sender's utility as:
\[
\max_{q^t,\zeta^t,\epsilon^t}\;\;
\sum_{x\in X_k}\sum_{\omega\in\Omega}\sum_{a\in A}\sum_{x'\in X_{k+1}}
q^t(x,\omega,a,x')\bigl(\hat u^{s,t}_i(x,\omega,a)+\xi^{s,t}_i(x,\omega,a)\bigr)
\quad 
\]
subject to the linear constraints on the transition functions, outcomes, occupancy measure and the incentive compatibility. Since we do not know the receiver types' true mean utility $u^{r,t}(x,\omega,a)$, the classical persuasiveness constraint could not be used. Therefore, an \emph{optimistic} incentive compatibility constraint is used ensuring that as our estimations get closer to the true mean $u^{r,t}(x,\omega,a)$. It can be seen that, the incentive compatibility becomes the persuasiveness constraint, as violation goes to $0$. The \emph{optimistic }incentive compatibility is given as follows:
\[
 \! \!\!\! \sum_{\omega\in\Omega}\sum_{x'\in X_{k+1}} q^t(x,\omega,a,x')
\Big(\hat u_i^{r,t}(x,\omega,a)+\xi_i^{r,t}(x,\omega,a)-\hat u_i^{r,t}(x,\omega,a')
+\xi_i^{r,t}(x,\omega,a')\Big)\ge 0 \notag\\
\]
In Algorithm \ref{alg:OPPS-full}, at each iteration the algorithm first updates all
estimators and confidence bounds using the feedback obtained from previous episodes as in Line~3.
It then commits to the signaling policy $\phi_i^t$ induced by an optimal solution
$\hat q_i^t$ of \textsf{Meta-Opt-Opt}, which is computed in Line~4. Notice that, the occupancy measure $q^t_i$ resulting from committing to $\phi^t_i$ is in general
different from computed $\hat q^t_i$, as the former is defined in terms of the true and unknown transition and prior functions, namely $P$ and $\mu$. 

\vspace{6pt}
\noindent
Furthermore, as stated in \citet{MarkovPersuasionScratch2025} under the good event \(\mathcal{E}(\delta)\), there exists a feasible solution to \textsf{Meta-Opt-Opt} program. Then, by bounding the difference between the estimated and true occupancy measures,
and establishing high-probability guarantees for the feasibility of the
\textsf{Meta-Opt-Opt} program under the proposed estimators, which ensures that
the estimated occupancy measures remain close to the true ones,
we arrive at the following theorem.
We provide the detailed analysis in Appendix~\ref{app:meta-optopt} and~\ref{sub:reg-viol}.

\begin{algorithm}[t]
\caption{Partial Feedback Meta Optimistic Persuasive Policy Search (Partial-Meta-OPPS)}
\label{alg:OPPS-part}
\begin{algorithmic}[1]
\Require $X,\Omega,A,m,T$, $\delta\in(0,1)$, $\alpha\in[1/2,1]$
\State $N \gets \lceil m^{\alpha}\rceil$
\For{task $t=1,\ldots,T$}
\State Initialize counter $C(x,\omega,a)$ to $0$ for all $(x,\omega,a)$
    \For{iteration $i = 1,\ldots,m$}
      \State Update all estimators $\hat{P}^t_i,\hat{\mu}^t_i,\hat{u}_i^{s,t},\hat{u}^{r,t}_i$ and bounds
      $\epsilon^t_i,\zeta^t_i,\xi^{s,t}_i,\xi^{r,t}_i$ given new observations
      \If{$i \le N|X||\Omega||A|$}
        \State $(x,\omega,a) \gets \arg\min_{(x,\omega,a)\in X\times\Omega\times A} C(x,\omega,a)$
        \State $\widehat{q}^t_i \gets$ Solve  \textsc{Meta-Opt-Opt}  with its objective modified as $\sum_{x'\in X} q^t(x,\omega,a,x')$
        \State $C(x,\omega,a) \gets C(x,\omega,a) + 1$
      \Else
        \State $\widehat{q}^t_i \gets$ Solve  \textsc{Meta-Opt-Opt}  
      \EndIf
      \State $\phi^t_i \gets \phi^{\widehat{q}_i^t}$
      \State Run Protocol~\ref{protoc2} by committing to $\phi^t_i$
      \State Observe \textit{partial} feedback from Protocol~\ref{protoc2}
    \EndFor
\EndFor
\end{algorithmic}
\end{algorithm}

\begin{theorem}
Given any $\delta\in(0,1)$, with probability at least {$1-11\delta$}, Algorithm \ref{alg:OPPS-full} attains the given cumulative task averaged regret and cumulative task averaged violation:
\[
R^T_m \le 
\mathcal{O}\Bigg(
\frac{L^2\sqrt{m}}{\sqrt{m}+\sqrt{\kappa|X||\Omega||A|}}\;
|X|\sqrt{m|\Omega||A|\ln\!\Bigl(\frac{m|X||\Omega||A|}{\delta}\Bigr)} 
\Bigg)
\]
\[
V^T_m \le 
\mathcal{O}\Bigg(
\frac{L^2\sqrt{m}}{\sqrt{m}+\sqrt{\kappa|X||\Omega||A|}}\;
|X|\sqrt{m|\Omega||A|\ln\!\Bigl(\frac{m|X||\Omega||A|}{\delta}\Bigr)} 
\Bigg)
\]
\end{theorem}

\subsection{Partial Feedback Setting}
In the partial-feedback setting, the main difficulty compared to the full-feedback case 
lies in the limited observability of persuasiveness constraints. Specifically, after 
committing to a signaling policy $\phi_i^t$, the sender does not observe sufficient 
information to directly evaluate whether $\phi_i^t$ satisfies the persuasiveness 
constraints or not. Consequently, obtaining sublinear constraint violation in the 
partial-feedback regime is substantially more challenging than in the 
full-feedback setting. This limited feedback introduces an inherent trade-off between regret minimization 
and constraint violation, governed by the amount of exploration performed. 
To address this challenge, we utilize the exploration included version of OPPS, leading to Algorithm \ref{alg:OPPS-part}, again introduced in \citet{MarkovPersuasionScratch2025}.

\vspace{6pt}
\noindent
The key idea is to partition the episodes of each task into two distinct phases as exploration and exploitation phases.
The first phase is dedicated to the objective of constructing accurate estimates of the persuasiveness constraints to guarantee  sublinear cumulative violation. This phase lasts for the first
\( 
N |X| |\Omega| |A|
\) 
episodes, where $N \coloneqq \lceil m^\alpha \rceil$ and 
$\alpha \in [1/2,1]$ is a parameter provided to the algorithm that controls 
the relative duration of the exploration and exploitation phases. The second phase is devoted to regret minimization. During this phase, 
the algorithm proceeds analogously to Algorithm \ref{alg:OPPS-full}, using the estimates obtained during the exploration phase. This two-phase structure explicitly balances exploration for constraint 
estimation and exploitation for regret minimization, enabling sublinear 
regret while controlling cumulative persuasiveness violations. This leads to the following theorems, for which we provide details in Appendices~\ref{app:meta-optopt} and~\ref{sub:reg-viol}.

\begin{theorem}
Given any $\delta\in(0,1)$, with probability at least $1-11\delta$, Algorithm \ref{alg:OPPS-part} attains cumulative task averaged expected regret and cumulative task averaged violation:
\[
R^T_m \le 
\mathcal{O}\Bigg(NL|X||\Omega||A|+
\frac{L^2\sqrt{m}}{\sqrt{m}+\sqrt{\kappa|X||\Omega||A|}}\;
|X|\sqrt{m|\Omega||A|\ln\!\Bigl(\frac{m|X||\Omega||A|}{\delta}\Bigr)} 
\Bigg)
\]
\text{where} $N \coloneqq \lceil m^\alpha \rceil$ is the length of the exploration phase.
\end{theorem}

\begin{theorem}
Given any $\delta\in(0,1)$, with probability at least $1-13\delta$, Algorithm \ref{alg:OPPS-part} attains cumulative task averaged violation:
\begin{align}
V^T_m & \!\le \tilde{\mathcal{O}}\Bigg[\rho  \Bigg(
\frac{Lm}{\sqrt{m}\!+\!\sqrt{\kappa|X||\Omega||A|}} 
\!+\! \frac{\sqrt{|X||\Omega||A|}N}{\!\sqrt{NL}\!+\!\sqrt{\kappa|X|\Omega||A|}\!}\!+\! \sqrt{N} \! 
+ \!\frac{m}{\!\sqrt{NL}\!+\!\sqrt{\kappa|X|\Omega||A|}\!}\!+\!{\frac{m}{\sqrt N}}  \Bigg) \Bigg] \notag
\end{align}
\( 
\text{where}\
\rho := |X||\Omega||A|^2L\,
\sqrt{\ln\Bigl(\frac{1}{\delta}\Bigr)}
\), and $N \coloneqq \lceil m^\alpha \rceil$ is the length of the exploration phase.
\end{theorem}

\section{Numerical Results}

\subsection{Numerical Results for Online Bayesian Persuasion}

In the OBP experiments, we use the classic judge--prosecutor
example of \citet{KamenicaGentzkow2011BayesianPersuasion}. We consider an environment with two actions,
two outcomes, and two receiver types, $K=2$. For each task, prior, the utilities of the sender and the receiver are sampled from a uniform distribution around their means,  provided below, with width $\tau_1 = 0.05$. To construct persuasive policies $\phi \in \mathcal{P}$, we sample these policies from a uniform probability grid while only those signaling schemes that satisfy the persuasiveness
constraints for both receiver types are retained. The means for the prior, sender utility and 2 receiver types' utilities are given as:
{\small
\[
\mu(w)=\begin{pmatrix}0.2 & 0.8\end{pmatrix},\quad
u^s(\omega,a)=-\begin{pmatrix}0.7 & 0.3\\0.7 & 0.3\end{pmatrix},\quad
u_1^r(\omega,a)=-\begin{pmatrix}0.7 & 0.3\\0.3 & 0.7\end{pmatrix},\quad
u_2^r(\omega,a)=-\begin{pmatrix}0.8 & 0.2\\0.2 & 0.8\end{pmatrix}.
\]
}
\noindent 
Furthermore, the loss vectors of the persuasive policies $\phi \in \mathcal{P}$ are bounded in $[0,1] \subset \mathbb{R}^2$ and
the interval of the
learning rate $\eta$ is chosen to be $[0.05, 0.25]$. The within-task iteration number is set to $m=5$, while the total number of tasks is $T=25$ for the experiments. Finally, we report on the average trajectory over 20 runs for both the non-meta-learning and meta-learning cases. The shaded regions indicate one standard deviation around the mean trajectories, which is given in the Figures \ref{fig:regret0} and \ref{fig:violation0}.

\begin{figure}[H] 
    \centering
    \begin{subfigure}{0.42\linewidth}
        \centering
        \includegraphics[width=\linewidth]{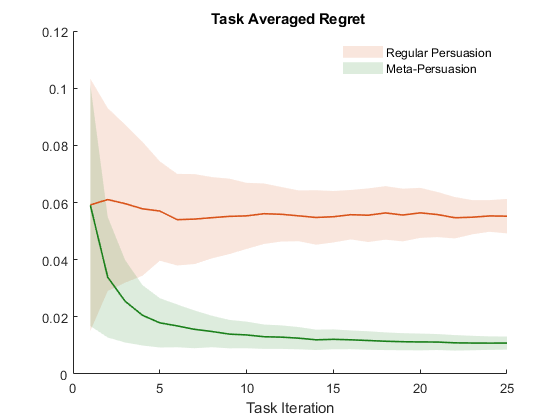}
        \caption{Task-averaged regret, full feedback }
        \label{fig:regret0}
    \end{subfigure}
    \begin{subfigure}{0.42\linewidth}
        \centering
        \includegraphics[width=\linewidth]{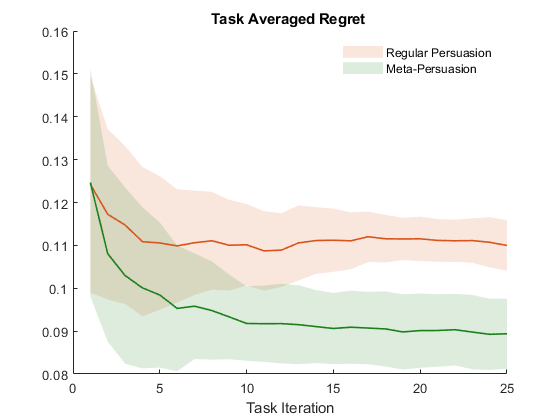}
        \caption{Task-averaged regret, partial feedback }
        \label{fig:violation0}
    \end{subfigure}
\end{figure}

\subsection{Numerical Results for Markov Persuasion Processes}

In the MPP experiments, we again use the  judge--prosecutor
persuasion example.
The environment consists of two states, two actions, two outcomes, and two layers.
The across-task mean parameters are defined as follows. The transition to the second layer is deterministic, and given as \( P_G(x_2 \mid x_1, \omega, a) = 1.\)  The outcome kernel, conditional on the state, on average, with  the sender's and receiver's mean utilities are given by
\begin{equation}
\mu_G(\omega \mid x) =
\begin{pmatrix}
0.2 & 0.8 \\
0.8 & 0.2
\end{pmatrix}, \;
u_G^s(\cdot,\omega,a) =
\begin{pmatrix}
0.7 & 0.3 \\
0.7 & 0.3
\end{pmatrix},
\;
u_G^r(\cdot,\omega,a) =
\begin{pmatrix}
0.7 & 0.3 \\
0.3 & 0.7
\end{pmatrix}.
\; \notag
\end{equation}

\vspace{6pt}
\noindent
Across tasks, the probability distributions of the random variables are drawn
from a uniform distribution around their mean with width $\tau_2 = 0.01$.
Within each task, sampling is again done by a uniform distribution centered around the sampled mean with width $\tau_3 = 0.1$. The number of tasks is set to $T=1000$, and each task consists of $m=200$ iterations. At every iteration, \textsf{Meta-Opt-Opt} is solved with updated estimators. Since incentive compatibility constraints are not known exactly at the beginning, the algorithm exhibits positive violation and negative regret.
However, as more tasks are observed, the estimators improve, and produce
increasingly feasible solutions with lower violation.
Consequently, regret approaches zero from the negative side. Finally, we report on the average regret and violation trajectories over 20 runs for both the OPPS and Meta-OPPS algorithms under partial and full feedback. The shaded regions indicate the one standard deviation around the mean trajectories. The results are shown in Figures~\ref{fig:regret1}--\ref{fig:violation1} and Figures~\ref{fig:regret2}--\ref{fig:violation2}.

\begin{figure}[H]
    \centering
    
    \begin{subfigure}{0.41\linewidth}
        \centering
        \includegraphics[width=\linewidth]{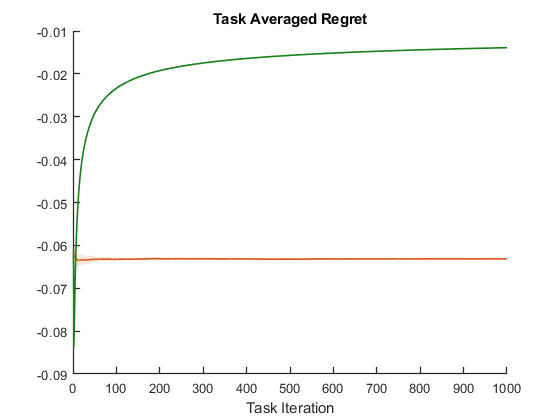}
        \caption{Task-averaged regret, full feedback}
        \label{fig:regret1}
    \end{subfigure}
    \begin{subfigure}{0.41\linewidth}
        \centering
        \includegraphics[width=\linewidth]{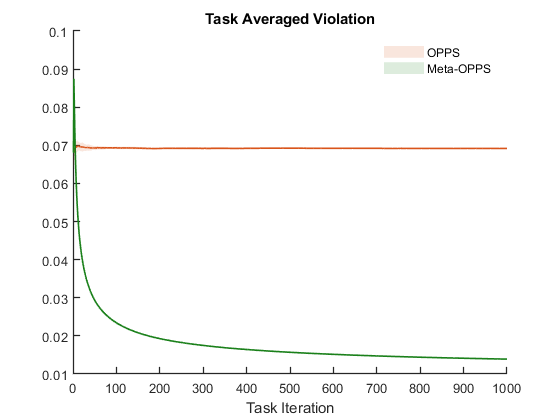}
        \caption{Task-averaged violation, full feedback}
        \label{fig:violation1}
    \end{subfigure}   
\end{figure}

\begin{figure}[ht]
    \centering
    
    \begin{subfigure}{0.41\linewidth}
        \centering
        \includegraphics[width=\linewidth]{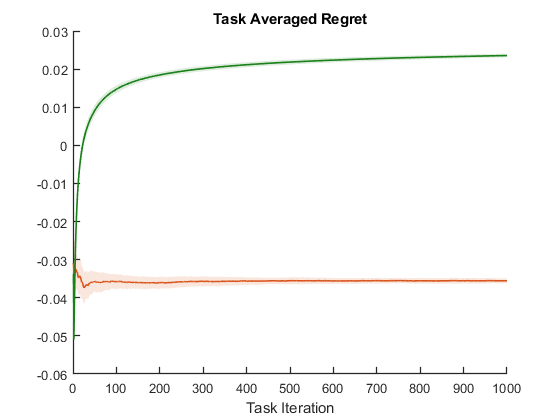}
        \caption{Task-averaged regret, partial feedback}
        \label{fig:regret2}
    \end{subfigure}
    \begin{subfigure}{0.41\linewidth}
        \centering
        \includegraphics[width=\linewidth]{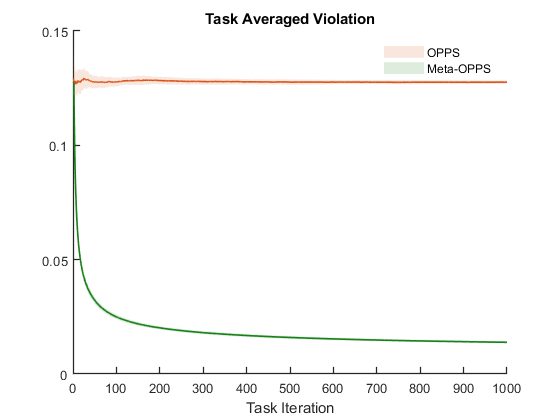}
        \caption{Task-averaged violation, partial feedback}
        \label{fig:violation2}
    \end{subfigure}
\end{figure}

\section{Conclusion}

Building on the classical Bayesian persuasion framework 
and the meta-learning paradigm of leveraging structure across related tasks, 
we have introduced \emph{meta-persuasion} algorithms for repeated Bayesian persuasion 
in both the Online Bayesian Persuasion (OBP) and Markov Persuasion Process (MPP) 
settings under full and partial feedback. 
Our approach establishes that when tasks share a common latent structure, 
the sender can achieve strictly improved \emph{task-averaged regret} guarantees 
relative to learning each task independently, while recovering standard 
worst-case rates under heterogeneous or adversarial task sequences.

\vspace{6pt}
\noindent
Our work opens several promising directions for future research. 
First, extending meta-persuasion to adversarially drifting task 
families would clarify when transfer remains beneficial and when it may 
degrade performance. 
Second, incorporating forward-looking strategic receivers and studying 
sequential persuasion problems within the meta-learning framework would introduce 
dynamic incentive compatibility considerations to the model. 
Finally, extending the meta-learning layer beyond linear loss structures to 
convex--concave or more general nonlinear classes, and analyzing settings in 
which receivers themselves learn over time, would connect meta-persuasion to 
broader themes in learning-in-games and dynamic information design.

\newpage 
\section{Acknowledgments}

Research of the authors was supported in part by the Army Research Office (ARO) Grant Number W911NF-24-1-0085

\bibliographystyle{plainnat} 
\bibliography{references}  

\newpage

\newpage

\appendix

\section{Online Bayesian Persuasion}

\subsection{Proofs for Online Bayesian Persuasion with Full Feedback}

\noindent In this appendix, we first present Lemmas \ref{lem:A1} and \ref{lem:C2}, which will be used in the proof of Theorem \ref{thm:OBP-full-fb}.

\noindent\begin{lemma} \label{lem:A1}
    
(Lemma A.1 \citet{MetaLearningBandits}) \  Let $\mathcal{R} : \bar{\nu}(\mathcal{P}) \mapsto \mathbb{R}_{\geq 0}$ be a strictly-convex function with 
\vspace{6pt}
\noindent
$\max_{z \in \bar{\nu}(\mathcal{P})} \| \nabla^2 \mathcal{R}(z) \|_2 \leq S$ over a convex set $\bar{\nu}(\mathcal{P}) \subset \mathbb{R}^K$ with 
$\max_{z \in \bar{\nu}(\mathcal{P})} \| z \|_2 \leq \sqrt{K}$. Then, for any points 
$\mathbf{z}^1, \ldots, \mathbf{z}^t \in \bar{\nu}(\mathcal{P})$, the actions 
$y^1 = \arg\min_{z \in \bar{\nu}(\mathcal{P})} \mathcal{R}(z)$ and 
$y^t = \tfrac{1}{t-1} \sum_{s < t} \mathbf{z}^s$ 
have regret
\( 
\sum_{t=1}^T D_{\mathcal{R}}(\mathbf{z}^t \| y^t) - D_{\mathcal{R}}(\mathbf{z}^t \| y^{T+1})
\leq 8 SK (1 + \ln T).
\
\) 
\end{lemma}

\noindent\begin{lemma}
\label{lem:C2}    

(Corollary C.2. \citet{AdaGrad_Meta-Kodak2019}) \ Let $\{U^{(t)} : \mathbb{R}_+ \to \mathbb{R}\}_{t \geq 1}$ be a sequence of functions of the form 
\( 
U^{(t)}(\eta) = \left( \frac{(B^{(t)})^2}{\eta} + \eta \right) \gamma^{(t)}
\) 
for any positive scalars $\gamma^{(1)}, \ldots, \gamma^{(T)} \in \mathbb{R}_+$ and adversarially chosen $B_t \in [0,A]$. Then, the $\epsilon - EWOO$ algorithm, with $\beta = \frac{4}{mA}\min\{\frac{\epsilon^2}{A^2},1\}$, for which $\epsilon > 0$, uses the actions of $EWOO$ run on the functions
\( 
\tilde{U}_t(\eta) = \left( \frac{(B^{(t)})^2 + \epsilon^2}{\eta} + \eta \right) \gamma^{(t)}
\) 
over the domain $[\epsilon, \sqrt{A^2 + \epsilon^2}]$ to determine $\eta^{(t)}$ achieves regret
\( 
\min \left\{ \frac{\epsilon^2}{\eta^*}, \, \epsilon \right\} 
\sum_{t=1}^T \gamma^{(t)}
+ \frac{A \gamma_{\max}}{2} \max \left\{ \frac{A^2}{\epsilon^2}, \, 1 \right\} 
\bigl(1 + \ln(T+1)\bigr)
\) 
for all $\eta^* > 0$.

\end{lemma}

{\hypersetup{linkcolor=black}
\noindent\textbf{Theorem \ref{thm:OBP-full-fb}.}
}Algorithm \ref{alg:omd-tuning-fb}  with $\epsilon = \sqrt{\frac{K}{m}}\frac{1}{T^{1/4}}$, $ \rho = \frac{1}{T^{1/4}}$, $A = \sqrt{\frac{K}{m}}$ and $\beta = \frac{4}{mA}\min\{\frac{\epsilon^2}{A^2},1\}$ achieves \vspace{1.5pt}task-averaged regret of 
    \(  R^T_m =  O(\sqrt{m\mathrm{Var}\!\left(\{z^{*(t)}\}_{t=1}^T\right)}) + o_T(\operatorname{poly}(m,A))  \)


\begin{proof}
In the proof, we first argue that enlarging the comparator class to the convex hull only upper bounds the regret notion of interest. Indeed, observe that
\[
\begin{aligned}
R_m 
&=  
\sum_{i=1}^m \mathbb{E}\!\left[\nu(\phi_i)^\top \mathbf{1}_{k_i}\right]
-
\min_{\phi^\ast \in \mathcal{P}}
\sum_{i=1}^m \nu(\phi^\ast)^\top \mathbf{1}_{k_i} \\
&=  \sum_{i=1}^m \tilde z_i^\top \mathbf{1}_{k_i}
-
\min_{z^\ast \in \nu(\mathcal{P})}
\sum_{i=1}^m (z^\ast)^\top \mathbf{1}_{k_i}
 \!\le
\sum_{i=1}^m \tilde z_i^\top \mathbf{1}_{k_i}
-\!\!\!
\min_{z^\ast \in \bar\nu(\mathcal{P})}
\sum_{i=1}^m (z^\ast)^\top \mathbf{1}_{k_i}
\end{aligned}
\]
since $\nu(\mathcal{P}) \subseteq \bar\nu(\mathcal{P})$, and minimizing over a larger set can only decrease the minimum. Next, by Carath\'eodory's theorem, any element of $\bar \nu(\mathcal{P})$ can be written as a convex combination of finitely many elements of $\nu(\mathcal{P})$. In Algorithm~\ref{alg:omd-tuning-fb}, the strategy sampling procedure exactly implements such convex decompositions. Therefore, by linearity of the loss, we can replace the expectation in the first term with evaluation at the mean strategy:
\( 
\mathbb{E}\!\left[\nu(\phi_i)^\top \mathbf{1}_{k_i}\right]
=
\tilde z_i^\top \mathbf{1}_{k_i},
\) \( 
\text{where }
\tilde z_i := \mathbb{E}[\nu(\phi_i)].
\) Finally, since the mapping $\nu(\cdot)$ is linear, it commutes with convexification:
\( 
\nu (\bar{\mathcal{P}})
=
\bar\nu(\mathcal{P}).
\)
Thus, bounding the regret reduces to a standard online linear optimization problem over the convex set $\bar\nu(\mathcal{P})$, and we may directly invoke the regret guarantees of OGD on this convex domain to bound $R_m$.

\vspace{6pt}
\noindent 
For OGD \citep{OnlineConvexProg-Zinkevich}, 
let $z^* = \nu(\phi^*)$  denote the optimal point for the sender, and 
let $\bar z_i = \bar \nu(\phi_i)$ denote the iterates. 
Since the loss functions over $\bar \nu(\mathcal{P})$ are linear and given by 
\( 
\mathcal{L}_k(\phi) = \nu(\phi)^\top \mathbf{1}_k,
\) 
their gradients with respect to the Euclidean norm have unit norm. Thus, we have 
$$ R_m \leq \frac{\| z^* - \bar z_1 \|_2^2}{2\eta} + \frac{\eta}{2}m $$ 
\noindent Then, across the tasks we get 
\[
\begin{aligned} \!R^T_m & \leq \frac{1}{T}\sum_{t=1}^T\frac{\| (z^{*})^{t} - (\bar z_1)^{t} \|_2^2}{2\eta^t} + \frac{1}{T}\sum_{t=1}^T\frac{\eta^t}{2}m \\
    &\!= \!\frac{1}{T}\!\Bigg[\!\sum_{t=1}^T \!
        \frac{\| (z^{*})^{t} \!-\! (\bar z_1)^{t} \|_2^2\!}{2\eta^t}
     \!+\! \!\sum_{t=1}^T 
        \!\frac{\eta^t m}{2} 
     \!+\! \!\min_{\eta > 0} \!
        \left\{
            \!\sum_{t=1}^T
            \!\frac{\!\| (z^{*})^{t} \!- \!(\bar z_1)^{t} \|_2^2\!}{2\eta}
            \!+\! \frac{\!\eta m\!}{2}\!
        \right\} 
     \!\!- \! \min_{\eta > 0} \!
        \left\{
            \!\sum_{t=1}^T\!
            \frac{\!\| (z^{*})^{t} \!- (\bar z_1)^{t} \|_2^2\!}{2\eta}
            \!+ \!\frac{\!\eta m\!}{2}\!
        \right\}\! \Bigg]\\ 
    &\!= \Delta_U + \frac{1}{T}\min_{\eta > 0} \left\{ \sum_{t=1}^T\frac{\|(z^{*})^{t} - (\bar z_1)^{t} \|_2^2}{2\eta} + \frac{\eta m}{2}\right\} \end{aligned}
\] where,
\( 
\Delta_U \coloneqq \frac{1}{T} \sum_{t=1}^T(\frac{\| (z^{*})^{t} - (\bar z_1)^{t} \|_2^2}{2\eta^t} + \frac{\eta^t m}{2})-\frac{1}{T}\min_{\bar\eta>\eta > 0} \left\{ \sum_{t=1}^T\frac{\| (z^{*})^{t} - (\bar z_1)^{t} \|_2^2}{2\eta} + \frac{\eta m}{2}\right\}
\). Then, we have:
\[
\begin{aligned}
\!R^T_m 
&\le\! \Delta_U \!
+ \!\min_{\eta > 0} \!\Bigg\{\!
      \frac{1}{T}\!\!
      \min_{z \in \bar{\nu}(\mathcal{P})}
      \!\left[
          \sum_{t=1}^T
          \!\frac{\!\| (z^{*})^{t} - z \|_2^2}{2\eta}
      \right] 
      \! \!+ \! \frac{1}{T}\!
        \sum_{t=1}^T\!
        \frac{\| (z^{*})^{t} - (\bar z_1)^{t} \|_2^2\!}{2\eta}
      \! + \! \frac{\!\eta m\!}{2} \! - \! \frac{1}{T}\!\!
      \min_{z \in \bar{\nu}(\mathcal{P})}\!\!
      \left[
          \sum_{t=1}^T\!
          \frac{\!\| (z^{*})^{t}\! - \!z \|_2^2\!}{2\eta}
      \right] \!
\Bigg\} 
\end{aligned}
\]
\noindent
One can see that the minimization with respect to $z$ in the second term 
trivially yields 
\( 
z = \frac{1}{T} \sum_{t=1}^{T} (z^{*})^{t},
\) 
which we denote by $\overline{z^*}$. 
Furthermore, the term becomes
\( 
\frac{1}{2\eta} \sum_{t=1}^{T} 
\left\| (z^{*})^{t} - \overline{z^*} \right\|_2^2,
\) 
which corresponds to the empirical variance of the sequence 
$\{(z^{*})^{t}\}_{t=1}^{T}$. 
We denote this quantity by $\mathrm{Var}\!\left(\{z^{*(t)}\}_{t=1}^T\right)$.
The first term after $\Delta_U$ becomes $\frac{\mathrm{Var}\!\left(\{z^{*(t)}\}_{t=1}^T\right)}{2 \eta}$. For the second term, we use Lemma \ref{lem:A1} with $S =1$, $K = K$, $\mathcal{R}(z) = \frac{1}{2}\|z\|_2^2 $, with the identification, 
\(  
z^t \equiv (z^*)^t,
\) \(
y^t \equiv (\bar z_1)^{t} = \frac{1}{t-1} \sum_{s < t} (z^*)^s,
 \)  \(
y^{T+1} = \frac{1}{T} \sum_{t=1}^{T} (z^*)^t = \overline{z^*}.
\) Since we use the Euclidean regularizer, 
we have
\( 
D_{\mathcal R}(a \| b)
=
\frac{1}{2}\|a-b\|_2^2,
\) 
and $\|\nabla^2 \mathcal R(z)\|_2 = 1$, so $S = 1$. Therefore, substituting the Euclidean Bregman divergence and dividing both sides by $\eta T$ yields,
\[
\frac{1}{T}
\sum_{t=1}^{T}
\frac{\|(z^*)^t - (\bar z_1)^{t}\|_2^2}{2\eta}
-
\frac{1}{T}
\sum_{t=1}^{T}
\frac{\|(z^*)^t - \overline{z^*}\|_2^2}{2\eta}
\le
\frac{8K(1+\ln T)}{\eta T}.
\]
Now observe that
\( 
\overline{z^*}
\in
\arg\min_{z}
\sum_{t=1}^{T}
\|(z^*)^t - z\|_2^2,
\) 
and therefore we may directly bound regret as
\begin{equation}  \label{eq:i}
R^T_m  \leq \Delta_U + \min_{\eta>0} \left\{ \tfrac{\eta m}{2} + \frac{\mathrm{Var}\!\left(\{z^{*(t)}\}_{t=1}^T\right)}{2 \eta} +\frac{8K(1+\ln T)}{\eta T} \right\} \  
\end{equation} 
\noindent
Now we bound the $\Delta_U$ term using Lemma \ref{lem:C2} with \( \tilde{U}_t(\eta) =\left( \frac{\| (z^*)^{t} - (\bar z_1)^{t} \|_2^2 + m\rho^2A^2}{2\eta} + \frac{m \eta}{2} \right) \)  
\noindent
where  $\epsilon = \rho A$, $\gamma^{(t)} = \frac{m}{2}$, $\rho = \frac{1}{T^{1/4}}$, $B_t = \frac{\| (z^{*})^{t} - (\bar z_1)^{t} \|_2^2}{\sqrt{m}}$, $A = \frac{\sqrt{K}}{\sqrt{m}}$. Then, by direct application of Lemma \ref{lem:C2} we get, 
\begin{equation} \label{eq:ii}
\Delta_U \!\leq\! \min \left\{\frac{K}{\eta^*\sqrt{T}},\frac{\sqrt{K}}{T^{1/4}} \right\}\frac{m}{2} + \frac{m\sqrt{K}}{2}\frac{(1 + \ln (T+1))}{\sqrt{T}}\!=\!\frac{m\sqrt{K}}{2}\!\left( \!\min\!\left\{\frac{\sqrt{K}}{\eta^*\sqrt{T} },\frac{1}{T^{1/4}} \right\}\! +\! \frac{1 \!+\! \ln (T\!+\!1)}{\sqrt{T}}\!\right)    \end{equation} 
\noindent
for all $\eta^* > 0$. Now, combining equations \eqref{eq:i} and \eqref{eq:ii}:

\begin{align*}
R^T_m
\le &\min_{\eta>0} \!\left\{
    \tfrac{\eta m}{2}
    \!+\! \frac{\mathrm{Var}\!\left(\{z^{*(t)}\}_{t=1}^T\right)}{2 \eta}
    \!+ \!\frac{8K(1+\ln T)}{\eta T}
\right\} 
\!+\! \frac{m\sqrt{K}}{2}\!\left(\!
    \min\left\{
        \frac{\sqrt{K}}{\eta^*\sqrt{T}},
        \frac{1}{T^{1/4}}
    \right\}
    + \frac{1 + \ln (T+1)}{\sqrt{T}}
\!\right)
\end{align*}
\noindent
Then, choosing \( 
\eta = \sqrt{\frac{\mathrm{Var}\!\left(\{z^{*(t)}\}_{t=1}^T\right)}{m}}
\), yields \( R^T_m = O\Big(\sqrt{m\mathrm{Var}\!\left(\{z^{*(t)}\}_{t=1}^T\right)}\;\Big) + 
o_T(\operatorname{poly}(m,K)) \).

\end{proof}

\subsection{Proofs for Online Bayesian Persuasion with Partial Feedback}

\noindent In this appendix, we first present and prove Theorem \ref{thm:ctomd-regret-obp} and Lemma \ref{lem:offset-cost}, and then present the proof of Theorem \ref{thm:obp-bandit-meta1}.

\begin{theorem}[Single-task regret of CTOMD]\label{thm:ctomd-regret-obp}
Fix a task $t\in[T]$ and suppress the superscript $t$. Run Algorithm~\ref{alg:omd-tuning} with $\eta K\le \tfrac{1}{4}$, and let $b:=1/\sqrt{m}$.
Then, for every comparator $u\in\bar\nu_b(P)$, choosing
\( 
\eta = \frac{\sqrt{K\ln m}}{4K\sqrt{m}}
\) 
yields
\[
\mathbb{E}\!\left[\sum_{i=1}^m \nu(\phi_i)^\top \mathbf{1}_{k_i}\right]
\;\le\;
\min_{u\in\bar\nu_{1/\sqrt{m}}(P)}\sum_{i=1}^m u^\top \mathbf{1}_{k_i}
\;+\;16K^{3/2}\sqrt{m\ln m}.
\]
\end{theorem}

\begin{proof}
We write $H_i:=\nabla^2\mathcal{R}(z_i)$. 
First, a standard property of self-concordant barriers implies $W_1(z_i)\subset \operatorname{int}(\bar\nu(P))$
for all $z_i\in\operatorname{int}(\bar\nu(P))$
\citep{CompetingIntheDark}.
Hence the sampling in Algorithm~\ref{alg:omd-tuning}, i.e.
\( 
\bar y_i = \bar z_i + \varepsilon_i v_{j''}^{-1/2}\,e_{j''},
\) 
is feasible as it has the local norm of $1$, with respect to $\bar z_i$.

\vspace{6pt}
\noindent
Since $\bar y_i\in\bar\nu(P)$
, Carath\'eodory's theorem yields points
$y_{i,1},\dots,y_{i,m_i}\in\nu(P)$ and weights $\lambda_{i,j}\ge0$ with $\sum_j\lambda_{i,j}=1$ such that
$\bar y_i=\sum_j \lambda_{i,j}y_{i,j}$.
Algorithm~\ref{alg:omd-tuning} samples $j'\sim \lambda_i$ and plays $\phi_i=\nu^\dagger(y_{i,j'})$,
so that $\nu(\phi_i)=y_{i,j'}$, and therefore
\begin{equation}\label{eq:caratheodory-exp}
\mathbb{E}\bigl[\nu(\phi_i)\mid \bar y_i\bigr]=\bar y_i.
\end{equation}
Since the loss is linear in the lifted vector, conditioning on $y_i$ gives
\begin{equation}\label{eq:loss-cond-y}
\mathbb{E}\bigl[\nu(\phi_i)^\top \mathbf{1}_{k_i}\mid \bar y_i\bigr]
=
\bar y_i^\top \mathbf{1}_{k_i}.
\end{equation}

\vspace{6pt}
\noindent
For the estimator, 
\( 
\tilde\ell_i
:=
K\cdot \bigl(\nu(\phi_i)^\top \mathbf{1}_{k_i}\bigr)\cdot \varepsilon_i\, v_{j''}^{1/2}\,e_{j''}
\in\mathbb{R}^K.
\) 
We first show that
\( 
\mathbb{E}[\tilde\ell_i\mid \bar z_i]=\mathbf{1}_{k_i}.
\) 
Indeed, we have $\mathbb{E}[\bar y_i\mid \bar z_i]=\bar z_i$ since $\mathbb{E}[\varepsilon_i]=0$.
Next, conditioning on $\bar z_i$ and on the event $j''=j$,
using \eqref{eq:loss-cond-y} and $\bar y_i=\bar z_i+\varepsilon_i v_j^{-1/2}e_j$, we get 
\( 
\mathbb{E}\!\left[\nu(\phi_i)^\top \mathbf{1}_{k_i}\ \bigm|\ \bar z_i, j''=j, \varepsilon_i\right]
=
\left(\bar z_i+\varepsilon_i v_j^{-1/2}e_j\right)^\top \mathbf{1}_{k_i}.
\) 
Multiplying by $\varepsilon_i v_j^{1/2}e_j$ and averaging over $\varepsilon_i\in\{\pm1\}$ cancels out the $\bar z_i$ term
and yields
\( 
\mathbb{E}_{\varepsilon_i}\!\left[\tilde\ell_i\mid \bar z_i, j''=j\right]
=
K\,\langle \mathbf{1}_{k_i}, e_j\rangle\,e_j.
\) 
Finally averaging over $j''$ uniform on $[K]$ gives \( 
\mathbb{E}[\tilde\ell_i\mid \bar z_i]=\mathbf{1}_{k_i}.
\) Now, we leverage Lemma~2 from \citet{CompetingIntheDark}, which implies that for any comparator 
$u \in \bar{\nu}(P)$, the following holds under the FTRL/OMD update:
\begin{equation}\label{eq:md-basic}
\sum_{i=1}^m \langle \tilde\ell_i, \bar z_i-u\rangle
\;\le\;
\frac{D_{\mathcal{R}}(u,\bar z_1)}{\eta}
+
\sum_{i=1}^m \langle \tilde\ell_i, \bar z_i-\bar z_{i+1}\rangle.
\end{equation}
For a $\vartheta$-self-concordant barrier, the barrier growth controls the Bregman divergence from $z_1$
to any point at Minkowski distance at most $(1+b)^{-1}$ from the boundary; concretely
\[
D_{\mathcal{R}}(u,\bar z_1)\le \mathcal{R}(u)-\mathcal{R}(\bar z_1)\;\le\; \vartheta\ln\Bigl(1+\frac{1}{b}\Bigr),
\qquad \forall u\in\bar\nu_b(P).
\]
With $b=1/\sqrt{m}$, this implies $D_{\mathcal{R}}(u,\bar z_1)\le 2\vartheta\ln m$ \citep{NesterovNemirovskii1994}. Finally,  any polytope in $\mathbb{R}^K$
has at least $K$-self concordant barrier. Then, we have $D_{\mathcal{R}}(u,\bar z_1)\le 2K\ln m$. 

\vspace{6pt}
\noindent
Now, letting $h_i := \bar z_{i+1}-\bar z_i$ and $r_i := \|h_i\|_{z_i}$,  Lemma 6 from \citet{CompetingIntheDark}, implies $\|h_i\|_{z_i} < 4\eta K.$
Then, we bound the local dual norm of the estimator.
Condition on $\bar z_i$ and take $j''=j$.
Since $H_i^{-1}e_j=v_j^{-1}e_j$ and $\nu(\phi_i)^\top \mathbf{1}_{k_i}\in[0,1]$, we get; 
\[
\|\tilde\ell_i\|_{z_i,*}^2
=
\tilde\ell_i^\top H_i^{-1}\tilde\ell_i
=
K^2\bigl(\nu(\phi_i)^\top \mathbf{1}_{k_i}\bigr)^2\cdot v_j\cdot e_j^\top H_i^{-1}e_j
=
K^2\bigl(\nu(\phi_i)^\top \mathbf{1}_{k_i}\bigr)^2
\le K^2,
\]
and hence $\|\tilde\ell_i\|_{z_i,*}\le K$ almost surely. By Cauchy--Schwarz in the local primal-dual pair we have,
\begin{equation}\label{eq:innerprod-stab}
\langle \tilde\ell_i, \bar z_i-\bar z_{i+1}\rangle
=
\langle \tilde\ell_i, -h_i\rangle
\;\le\;
\|\tilde\ell_i\|_{\bar z_i,*}\,\|h_i\|_{\bar z_i}
\le 32\eta K^2 
\end{equation}

\noindent Summing over $i$ gives
\begin{equation}\label{eq:stab-sum}
\sum_{i=1}^m \langle \tilde\ell_i, \bar z_i-\bar z_{i+1}\rangle \le 32K^2\eta m.
\end{equation}

\vspace{6pt}
\noindent
Combining \eqref{eq:md-basic}, and \eqref{eq:stab-sum} yields, for all $u\in\bar\nu_{1/\sqrt{m}}(P)$,
\( 
\sum_{i=1}^m \langle \tilde\ell_i, \bar z_i-u\rangle
\le
\frac{2\vartheta\ln m}{\eta}+32K^2\eta m.
\) 
Taking expectations and using \( 
\mathbb{E}[\tilde\ell_i\mid z_i]=\mathbf{1}_{k_i} \)  gives
\[
\mathbb{E}\!\left[\sum_{i=1}^m \langle \mathbf{1}_{k_i}, \bar z_i-u\rangle\right]
\le
\frac{2\vartheta\ln m}{\eta}+32K^2\eta m.
\]
Finally, by \eqref{eq:caratheodory-exp}--\eqref{eq:loss-cond-y} and the tower property,
\( 
\mathbb{E}\!\left[\nu(\phi_i)^\top \mathbf{1}_{k_i}\right]
=
\mathbb{E}\!\left[\bar y_i^\top \mathbf{1}_{k_i}\right]
=
\mathbb{E}\!\left[\bar z_i^\top \mathbf{1}_{k_i}\right],
\) 
and thus the left-hand side becomes exactly the expected cumulative loss suffered by CTOMD, proving the stated bound.
The final optimized rate follows by plugging in $\eta=\frac{\sqrt{\vartheta\ln m}}{4K\sqrt{m}}$.

\end{proof}

\begin{lemma}\label{lem:offset-cost}
Assume that losses are \emph{value-bounded} on $\bar\nu(\mathcal P)$ in the sense that for every loss vector
$g\in\mathbb{R}^K$ under consideration,
\( 
0 \le \langle g,z\rangle \le 1
\text{ for all } z\in\bar\nu(\mathcal P).
\) Then, for any sequence $\{g_i\}^m_{i=1}$ we have
\[
\min_{u\in\bar\nu_b(\mathcal P)}\sum_{i=1}^m \langle g_i,u\rangle
\;\le\;
\min_{z\in\bar\nu(\mathcal P)}\sum_{i=1}^m \langle g_i,z\rangle + bm.
\]
\end{lemma}

\begin{proof}
Fix any $z\in\bar\nu(\mathcal P)$ and define the point
\( 
u
:=
z_1+\frac{1}{1+b}(z-z_1).
\) 
We first show that $u\in\bar\nu_b(\mathcal P)$.
Indeed,
\[
z_1+(1+b)(u-z_1)
=
z_1+(1+b)\cdot\frac{1}{1+b}(z-z_1)
=
z\in\bar\nu(\mathcal P),
\]
and thus by definition of $\pi_{z_1}(\cdot)$ we have $\pi_{z_1}(u)\le (1+b)^{-1}$, and hence $u\in\bar\nu_b(\mathcal P)$. Next, since $\langle g,\cdot\rangle$ is linear,
\[
\langle g,u\rangle
=
\frac{1}{1+b}\langle g,z\rangle+\frac{b}{1+b}\langle g,z_1\rangle
\le
\frac{1}{1+b}\langle g,z\rangle+\frac{b}{1+b}\cdot 1
=
\langle g,z\rangle+\frac{b}{1+b}\bigl(1-\langle g,z\rangle\bigr)
\le
\langle g,z\rangle+b.
\]
Now choosing $z^\star\in\arg\min_{z\in\bar\nu(\mathcal P)}\langle g,z\rangle$, and
summing the same argument over $i$ proves the claim.
\end{proof}

{\hypersetup{linkcolor=black}
\noindent\textbf{Theorem \ref{thm:obp-bandit-meta1}.}
}
For each $b$ in $\mathcal G$ with interval \( (\underline{b},\bar b)\) define the constants
\( 
D_b^2 := \max_{x,y\in\mathcal{\bar \nu}_b}D_{\mathcal R}(x\|y),
\) \( 
S_b := \max_{x\in\mathcal{\bar \nu}_b}\|\nabla^2\mathcal{R}(x)\|_2,
\) \( 
\mathsf{K} := \max_{x,y\in\mathcal{\bar \nu}}\|x-y\|_2.
\) 
Define the divergence, at level $b$ by
\( 
\widehat V_b^2
:=
\min_{z\in\mathcal{\bar \nu}}\;
\mathbb{E}\!\left[\frac{1}{T}\sum_{t=1}^T D_{\mathcal R}\!\big(\mathrm{OPT}_b(\tilde\ell^t)\,\|\,z\big)\right]
\), where $\mathrm{OPT}_b(\tilde \ell^t) = \arg\min_{x\in\bar\nu_b}\langle\tilde \ell^t ,x\rangle$. Then, running Algorithm~\ref{alg:omd}, there exist a grid size
$k=\widetilde O(D_{\underline b}^2K\sqrt{mT})$ and a meta step-size $\alpha$ such that the expected task-averaged regret satisfies
\begin{align}
\!\!\!\!\!\mathbb{E}\Big[\frac{1}{T}\sum_{t=1}^T\sum_{i=1}^m
\big(\nu(\phi_i^t)^\top \mathbf{1}_{k_i^t} - \nu(\phi_t^\star)^\top \mathbf{1}_{k_i^t}\big)\Big]
&\le
72K\frac{\sqrt m}{T^{-1/4}}
\Big(
D_{\underline b}\sqrt{\frac{m}{T}\ln k}
+
\frac{S_{\underline b}\mathsf{K}^2}{D_{\underline b}T}(1+\ln T)
\Big) \nonumber\\
&
+ \!\!\!\!\!\!\min_{z\in\mathcal{\bar \nu}, \eta>0, b\in[\underline b,\bar b]}
\!\mathbb{E}\!\left[\!\frac{1}{T}\!\sum_{t=1}^T
\!\frac{D_{\mathcal R}(\mathrm{OPT}_b(\hat\ell^t)\|z)}{\eta} \!+\! (32\eta K^2\!+\!b)m\!\right]\!.\label{eq:meta-main}
\end{align}
Moreover, optimizing over $\eta$ yields the simplified form
\begin{equation}\label{eq:meta-simplified}
\mathbb{E}[R_m^T]
\;\le\;
\widetilde O\!\left(\frac{D_{\underline b}K m}{T^{1/4}}+\frac{S_{\underline b}\mathsf{K}^2K\sqrt m}{D_{\underline b}T^{3/4}}\right)
\;+\;
\min_{b\in[\underline b,\bar b]}
\Big(4K\widehat V_b\sqrt{2m}+bm\Big).
\end{equation}
In particular, as $T\to\infty$ the $\widetilde O(\cdot)$ term vanishes and
$\mathbb{E}[R_m^T]=O(\widehat V_b\sqrt m+bm)$ for the best $b$ in the range.

\begin{proof}

Since $\nu$ is linear, $\nu(\bar{\mathcal{P}}) =\bar \nu(\mathcal P)=\mathcal{\bar \nu}$ and the loss is linear in $\nu(\phi)$.
Thus the per-task comparator can be taken as
\( 
z_t^\star \in \arg\min_{z\in\mathcal{\bar \nu}}\sum_{i=1}^m z^\top \mathbf{1}_{k_i^t},
\) 
which upper-bounds regret against $\phi_t^\star\in\mathcal P$. Let $g_i^t:=\mathbf{1}_{k_i^t}$ and let $\ell^t:=\sum_{i=1}^m g_i^t$.
Therefore, by Lemma~\ref{lem:offset-cost}, for tasks $ t \in [T]$, 
\begin{align}
\mathbb{E}\sum_{t=1}^T\sum_{i=1}^m \langle g_i^t, (\bar z_i)^t - z_t^\star\rangle
&\le
\mathbb{E}\sum_{t=1}^T b_t m
+
\sum_{t=1}^T\sum_{i=1}^m \langle g_i^t, (\bar z_i)^t - \mathrm{OPT}_b^t(\ell^t)\rangle.\label{eq:offset-step}
\end{align}
\noindent
By unbiasedness of CTOMD’s estimator, $\mathbb{E}[\tilde\ell_i^t\mid (\bar z_i)^t]=g_i^t$, and $\mathrm{OPT}_b(\ell^t)$ is deterministic
given the adversary’s losses, hence
\( 
\mathbb{E}\,\langle g_i^t, (\bar z_i)^t-\mathrm{OPT}_b(\ell^t)\rangle
=
\mathbb{E}\,\langle \tilde\ell_i^t, (\bar z_i)^t-\mathrm{OPT}_b(\ell^t)\rangle.
\) 
Summing over $i$ and $t$ and using $\hat\ell^t=\sum_{i=1}^m \tilde\ell_i^t$ gives
\[
\mathbb{E}\sum_{i=1}^m \langle \tilde\ell_i^t, (\bar z_i)^t-\mathrm{OPT}_b(\ell^t)\rangle
=
\mathbb{E}\Big[\sum_{i=1}^m \langle \tilde\ell_i^t, (\bar z_i)^t\rangle - \langle \hat\ell^t,\mathrm{OPT}_b(\ell^t)\rangle\Big]
\le
\mathbb{E}\Big[\sum_{i=1}^m \langle \tilde\ell_i^t, (\bar z_i)^t\rangle - \langle \hat\ell^t,\mathrm{OPT}_b(\hat\ell^t)\rangle\Big],
\]
since $\mathrm{OPT}_b(\hat\ell^t)$ minimizes $\langle \hat\ell^t,\cdot\rangle$ over $\mathcal{\bar \nu}_b$.
Thus,
\begin{equation}\label{eq:swap-opt}
\mathbb{E}\sum_{t=1}^T\sum_{i=1}^m \langle g_i^t, (\bar z_i)^t - \mathrm{OPT}_b(\ell^t)\rangle
\le
\mathbb{E}\sum_{t=1}^T\sum_{i=1}^m \langle \tilde\ell_i^t, (\bar z_i)^t - \mathrm{OPT}_b(\hat\ell^t)\rangle.
\end{equation}

\noindent
Condition on the hyperparameter $g^t=(\eta^t,b^t)$ sampled by the meta-learner on task $t$ and on the initialization
$z^t_1$ it provides.
Applying (\ref{eq:md-basic} and \ref{eq:stab-sum}) with comparator $u=\mathrm{OPT}_b(\hat\ell^t)\in\mathcal{\nu}_b$ yields
\[
\sum_{i=1}^m \langle \tilde\ell_i^t, (\bar z_i)^t-\mathrm{OPT}_b(\hat\ell^t)\rangle
\le
\frac{D_{\mathcal R}(\mathrm{OPT}_b(\hat\ell^t)\|z^t_1)}{\eta_t}+32K^2\eta_t m.
\]
Combining with \eqref{eq:offset-step}--\eqref{eq:swap-opt} gives
\begin{equation}\label{eq:per-task-upper}
\mathbb{E}\sum_{t=1}^T\sum_{i=1}^m \langle g_i^t, (\bar z_i)^t - z_t^\star\rangle
\le
\mathbb{E}\sum_{t=1}^T\left[\frac{D_{\mathcal R}(\mathrm{OPT}_b(\hat\ell^t)\|z^t_1)}{\eta_t}+(32K^2\eta_t+b_t)m\right].
\end{equation}

\noindent
Define the meta-loss for $g=(\eta,b)\in G$ and $z\in\mathcal{\bar \nu}$ by
\( 
U_t(z,g)
:=
\frac{D_{\mathcal R}(\mathrm{OPT}_b(\hat\ell^t)\|z)}{\eta}+(32K^2\eta+b)m.
\) 
It can be seen that Algorithm \ref{alg:omd} is exactly the algorithm stated in \citet{MetaLearningBandits} specialized to BLO setting with no $\beta$ updates.
So, we may apply \citealp[Thm.\ 3.1]{MetaLearningBandits} with the BLO constants determined as;
\( 
d\leftarrow K,\;
G \leftarrow 4K\sqrt2,\;
D\leftarrow D_{\underline b},\;
S\leftarrow S_{\underline b},\;
K\leftarrow \mathsf{K},\;
M\leftarrow 1.
\) With the same discretization size $k$ as in \citealp[Thm.\ 5.1]{MetaLearningBandits},
this yields inequality \eqref{eq:meta-main} after dividing by $T$. Now, fixing $b$ and $z$, let
\( 
A_b(z):=\mathbb{E}\!\left[\frac{1}{T}\sum_{t=1}^T D_{\mathcal R}(\mathrm{OPT}_b(\hat\ell^t)\|z)\right].
\) 
Then,
\[
\min_{\eta>0}\left\{\frac{A_b(z)}{\eta}+32K^2\eta m\right\}
=
2\sqrt{32}\,K\sqrt{mA_b(z)}
=
4K\sqrt{2m}\,\sqrt{A_b(z)}.
\]
Minimizing over $z\in\mathcal{\bar \nu}$ gives $\sqrt{A_b(z)}=\widehat V_b$, proving \eqref{eq:meta-simplified}.
The asymptotic statement follows as the leading $\widetilde O(\cdot)$ term is $o_T(1)$.
\end{proof}

\section{Markov Persuasion Processes}

As discussed earlier, the relevant task-dependent parameters of the repeated games are drawn from distributions supported on $[0,1]$, with across-task means
$P_G(\cdot \mid x,\omega,a)$,
$\mu_G(\cdot \mid x)$,
$u_G^s(x,\omega,a)$, and
$u_G^r(x,\omega,a)$
for each $x \in X$, $\omega \in \Omega$, and $a \in A$, together with their corresponding variances.
Additionally, we assume that all tasks have the same state, outcome, and action-space cardinalities.
For clarity, let $\Psi$ denote a uniform upper bound on the $\ell_1$-deviation of a single task draw from its across-task mean, namely
\( 
\Psi > \Psi_P,\Psi_{\mu},\Psi_{u^{s}},\Psi_{u^{r}}.
\) 

\subsection{Confidence Bounds for Meta-Estimators}\label{app:meta-conf-bdd}

\noindent
The estimated probability of transitioning from $x \in X$ to $x' \in X$  by taking action $a \in A$, when the realized outcome in state $x$ is 
$\omega \in \Omega$, based on the estimations from previous tasks, is given by
\[
\hat P_i^t(x' \mid x,\omega,a)
:=
w_{\kappa_P}\!\big(N_i^t(x,\omega,a)\big)\,
\frac{N_i^t(x,\omega,a,x')}{\max\{1,N_i^t(x,\omega,a)\}}
+
\bar w_{\kappa_P}\!\big(N_i^t(x,\omega,a)\big)\,
\bar P_G^{\,t-1}(x' \mid x,\omega,a).
\]
where we define
\( 
\bar P_G^{\,t-1}(x' \mid x,\omega,a)
:=
\frac{1}{\max\{1,M_{t-1}(x,\omega,a)\}}
\sum_{\tau=1}^{t-1}
\mathbf 1\{N_m^\tau(x,\omega,a)>0\}\,
\frac{N_m^\tau(x,\omega,a,x')}{\max\{1,N_m^\tau(x,\omega,a)\}}.
\) 
\begin{lemma} \label{lem:P-conf}
    Given any $\delta\in(0,1)$, with probability at least $1-2\delta$,  the following inequality holds for every $x \in X$, $\omega\in \Omega$, $a \in A$, $i \in [m]$, and $t \in [T]$ jointly:
    \[
    ||P^t(.|x,\omega,a)-\hat{P}_i^t(.|x,\omega,a)||_1 \leq \epsilon_i^t(x,\omega,a)
    \]
    where \( 
\epsilon_i^t(x,\omega,a)
\!:=\!
w_\kappa\!\big(N_i^t(x,\!\omega,\!a)\big)
\!\sqrt{
\frac{
2|X_{k(x)+1}|
\ln\!\big(\!\frac{m|X||\Omega||A|}{\delta}\big)
}{
\max\{1,N_i^t(x,\omega,a)\}
}
}
+
\bar w_\kappa\!\big(N_i^t(x,\!\omega,\!a)\big) \!\!
\left[
\!\sqrt{
\frac{
\!2|X_{k(x)+1}|
\ln\!\big(\!\frac{|X||\Omega||A|T}{\delta}\big)\!
}{
\max\{1,M_{t-1}(x,\omega,a)\}
}
}
\!+\!\Psi\!
\right].
\) 
    
\end{lemma}

\begin{proof}
\[
\begin{aligned}
\left\lVert P^t(\cdot\mid x,\omega,a)-\hat{P}_i^t(\cdot\mid x,\omega,a)\right\rVert_1
&\le w_p\left\lVert \bar{P}_i^t(\cdot\mid x,\omega,a)-P^t(\cdot\mid x,\omega,a)\right\rVert_1 \\
&\quad + (1-w_p)\left\lVert \bar{P}_G^{t-1}(\cdot\mid x,\omega,a)-P_G(\cdot\mid x,\omega,a)\right\rVert_1 \\
&\quad + (1-w_p)\left\lVert P^t(\cdot\mid x,\omega,a)-P_G(\cdot\mid x,\omega,a)\right\rVert_1 .
\end{aligned}
\]
    We bound the first term by using the Eq.44 in  \citet{NearOptimalRL2008} and employing a union bound over all $x$, $\omega$, $a$, and $i$. The
    second term is bounded using the same inequality and a union bound over all $x$ $\omega$,$a$ and $t$,
    and third term is straightforward from $\Psi_P\leq\Psi$.
\end{proof}

\noindent
Next, we introduce confidence bounds for prior distributions. For every state $x\in X$, we define $\hat{\mu}_i^t(.|x) \in \Delta(\Omega)$  as the estimator of the prior distribution at $x$  built by using observations up to episode $i \in [m]$ and task $t \in [T]$. Formally, the entries of vector $\hat{\mu}_i^t(.|x)$ are such that, for every $\omega \in \Omega$:

\[
\hat\mu_i^t(\omega \mid x)
:=
w_{\kappa_\mu}\!\big(N_i^t(x)\big)\,
\frac{\sum_{j=1}^{i-1}\mathbf 1\{x_j^t=x,\ \omega_j^t=\omega\}}{\max\{1,N_i^t(x)\}}
+
\bar w_{\kappa_\mu}\!\big(N_i^t(x)\big)\,
\bar\mu_G^{\,t-1}(\omega \mid x),
\]
where we define, 
\( 
\bar\mu_G^{\,t-1}(\omega \mid x)
:=
\frac{1}{\max\{1,M_{t-1}(x)\}}
\sum_{\tau=1}^{t-1}
\mathbf 1\{N_m^\tau(x)>0\}\,
\frac{\sum_{j=1}^{m}\mathbf 1\{x_j^\tau=x,\ \omega_j^\tau=\omega\}}{\max\{1,N_m^\tau(x)\}}.
\) 

\begin{lemma}
    Given any $\delta\in(0,1)$, with probability at least $1-2\delta$,  the following inequality holds for every $x \in X$, $i \in [m]$, and $t \in [T]$ jointly:
       \[
    ||\mu^t(.|x)-\hat{\mu}_i^t(.|x)||_1 \leq \zeta_i^t(x)
    \]
    where \( 
\zeta_i^t(x)
:=
w_\kappa\!\big(N_i^t(x)\big)
\sqrt{
\frac{
2|\Omega|
\ln\!\big(m|X|/\delta\big)
}{
\max\{1,N_i^t(x)\}
}
}
+
\bar w_\kappa\!\big(N_i^t(x)\big)
\left(
\sqrt{
\frac{
2|\Omega|
\ln\!\big(|X|T/\delta\big)
}{
\max\{1,M_{t-1}(x)\}
}
}
+\Psi
\right).
\)  

\end{lemma}

\begin{proof}
    The proof follows the lines in the proof of Lemma \ref{lem:P-conf}. This time, we union bound the first term over all $x$ and $i$, and second term over all $x$ and $t$, where separate events have the cardinality of $|\Omega|$. The third term follows from $\Psi_\mu\leq\Psi$.
\end{proof}

\noindent
Finally, we introduce our estimators for the reward functions of the sender and the receiver types. 
In the following, we present the results related to the sender’s and receiver types’ rewards 
under both full and partial feedback. First, for every $x \in X$, $\omega \in \Omega$, and $a \in A$, the estimated sender and receiver 
rewards for the full feedback case is, constructed using observations up to episode $i \in [m]$ in task $t \in [T]$, 
are defined as follows:
\[
\!\hat u_{i}^{s,t}(x,\!\omega,\!a)
\!\!:=\!
w_{\kappa_{u^s}}\!\big(\!N_i^t(x,\omega)\big)
\frac{\!\sum_{j=1}^{i-1} \!u_j^{s,t}(x,\!\omega,a)\,\!\mathbf 1\{x_j^t\!=\!x, \omega_j^t\!=\!\omega\}\!\!}{\max\{1,N_i^t(x,\omega)\}}
+
\bar w_{\kappa_{u^s}}\!\big(\!N_i^t(x,\!\omega)\big)
\bar u_{G}^{s,t-1}\!(x,\!\omega,\!a)
\]
\[
\!\hat u_{i}^{r,t}(x,\!\omega,\!a)
\!\!:=\!
w_{\kappa_{u^r}}\!\big(\!N_i^t(x,\omega)\big)
\frac{\!\sum_{j=1}^{i-1} \!u_j^{r,t}(x,\!\omega,\!a\!)\,\!\mathbf 1\{x_j^t\!=\!x, \omega_j^t\!=\!\omega\}\!\!}{\max\{1,N_i^t(x,\omega)\}}
+
\bar w_{\kappa_{u^r}}\!\big(\!N_i^t(x,\!\omega)\big)
\bar u_{G}^{r,t-1}\!(x,\!\omega,\!a)
\]
where, \vspace{5pt}
\( 
\bar u_{G}^{s,t-1}(x,\omega,a)
\!:= \!
\frac{1}{\max\{1,M_{t-1}(x,\omega)\}}
\sum_{\tau=1}^{t-1}
\mathbf 1\{N_m^\tau(x,\omega)>0\}
\frac{\sum_{j=1}^{m} \!u_j^{s,\tau}(x,\omega,a)\mathbf 1\{x_j^\tau=x, \omega_j^\tau=\omega\}}{\max\{1,N_m^\tau(x,\omega)\}},
\) and receiver types' estimators are defined analogously. Secondly, for every $x \in X$, $\omega \in \Omega$, and $a \in A$, the estimated sender and receiver 
rewards for the partial feedback case is, constructed using observations up to episode $i \in [m]$ in task $t \in [T]$, 
are defined as follows:
\[
\!\hat u_{i}^{s,t}(x,\!\omega,\!a)
\!\!:=\!
w_{\kappa_{u^s}}\!\big(\!N_i^t(x,\omega,a)\big)
\frac{\!\sum_{j=1}^{i-1} \!u_j^{s,t}(x,\!\omega,\!a\!)\,\!\mathbf 1\{x_j^t\!=\!x, \omega_j^t\!=\!\omega, a_j^t\!=\!a\!\}\!\!}{\max\{1,N_i^t(x,\omega,a)\}}
+
\bar w_{\kappa_{u^s}}\!\big(\!N_i^t(x,\!\omega,\!a)\big)
\bar u_{G}^{s,t-1}\!(x,\!\omega,\!a)
\]
\[
\!\hat u_{i}^{r,t}(x,\!\omega,\!a)
\!\!:=\!
w_{\kappa_{u^r}}\!\big(\!N_i^t(x,\omega,a)\big)
\frac{\!\sum_{j=1}^{i-1} \!u_j^{r,t}(x,\!\omega,\!a\!)\,\!\mathbf 1\{x_j^t\!=\!x, \omega_j^t\!=\!\omega, a_j^t\!=\!a\!\}\!\!}{\max\{1,N_i^t(x,\omega,a)\}}
+
\bar w_{\kappa_{u^r}}\!\big(\!N_i^t(x,\!\omega,\!a)\big)
\bar u_{G}^{r,t-1}\!(x,\!\omega,\!a)
\]
where,\vspace{5pt}
\( 
\bar u_{G}^{s,t-1}(x,\omega,a)
\!:= \!
\frac{1}{\max\{1,M_{t-1}(x,\omega,a)\}}
\sum_{\tau=1}^{t-1}
\mathbf 1\{N_m^\tau(x,\omega,a)>0\}
\frac{\sum_{j=1}^{m} \!u_j^{s,\tau}(x,\omega,a)\mathbf 1\{x_j^\tau=x, \omega_j^\tau=\omega, a_j^\tau=a\}}{\max\{1,N_m^\tau(x,\omega,a)\}},
\) and receiver types' estimators are defined analogously. The following lemma establishes confidence bounds on the sender’s rewards under the assumption that full feedback is observed.

\begin{lemma}
    Given any $\delta\in(0,1)$, with probability at least $1-2\delta$,  the following inequality holds for every $x \in X$, $\omega \in \Omega$, $a \in A$, $i \in [m]$, and $t \in [T]$ jointly:

    \[
    | u^{s,t}(x,\omega,a) - \hat{u}_i^{s,t}(x,\omega,a)| \leq \xi_i^{s,t}(x,w,a)
    \]
    where $\xi_i^{s,t}(x,\omega,a) \coloneqq \min\{1,w_\kappa\!\big(N_i^t(x,\omega)\big) \sqrt{\frac{\ln{(3m|X||\Omega|/\delta})}{\max\{1, N^t_i(x,\omega)\}}}+\bar w_\kappa\!\big(N_i^t(x,\omega)\big)(\sqrt{\frac{\ln(3|X||\Omega|T)/\delta)}{\max\{1,M_{t-1}(x,w)\}}}+\Psi)\}$

\end{lemma}
\begin{proof}
    The proof follows the lines of the proof of Lemma \ref{lem:P-conf}. This time, instead of using Eq.44 in  \citet{NearOptimalRL2008} we use the Hoeffding's inequality. We then union bound over all $x$, $w$ and $i$, and second term over all $x$, $w$ and $t$, where separate events have the cardinality of $1$. The third term follows from $\Psi_{u^s}\leq\Psi$.
\end{proof}

\noindent
The following lemma establishes confidence bounds on the receiver’s rewards in the full feedback case.

\begin{lemma} \label{lem:reward-receiver-partinfo-est}
    Given any $\delta\in(0,1)$, with probability at least $1-2\delta$,  the following condition holds for every $x \in X$, $\omega \in \Omega$, $a \in A$, $i \in [m]$, and $t \in [T]$ jointly:

    \[
    | u^{r,t}(x,\omega,a) - \hat{u}_i^{r,t}(x,\omega,a)| \leq \xi_i^{r,t}(x,w,a)
    \]
    where $\xi_i^{r,t}(x,\omega,a) \coloneqq\min\{1,w_\kappa\!\big(N_i^t(x,\omega)\big) \sqrt{\frac{\ln{(3m|X||\Omega|/\delta})}{\max\{1, N^t_i(x,\omega)\}}}+\bar w_\kappa\!\big(N_i^t(x,\omega)\big)(\sqrt{\frac{\ln(3|X||\Omega|T)/\delta)}{\max\{1,M_{t-1}(x,w)\}}}+\Psi)\}$
    
\end{lemma}
\begin{proof}
    The proof follows the lines of the proof of Lemma \ref{lem:P-conf}. This time, instead of using the Eq.44 in  \citet{NearOptimalRL2008} we use the Hoeffding's inequality. We then bound over all $x$, $w$ and $i$, and second term over all $x$, $w$ and $t$, where separate events have the cardinality of $1$. The third term follows from $\Psi_{u^r}\leq\Psi$.
\end{proof}

\noindent
The following lemma establishes confidence bounds on the sender’s rewards for the partial feedback case.

\begin{lemma}
    Given any $\delta\in(0,1)$, with probability at least $1-2\delta$,  the following condition holds for every $x \in X$, $\omega \in \Omega$, $a \in A$ $i \in [m]$ and $t \in [T]$ jointly:

    \[
    | u^{s,t}(x,\omega,a) - \hat{u}_i^{s,t}(x,\omega,a)| \leq \xi_i^{s,t}(x,w,a)
    \]
    where $\xi_i^{s,t}(x,\omega,a) \!\coloneqq \!\min\{1,w_\kappa\!\big(N_i^t(x,\omega,a)\big)\! \sqrt{\frac{\ln{(3m|X||\Omega||A|/\delta})}{\max\{1, N^t_i(x,\omega,a)\}}}+\bar w_\kappa\!\big(N_i^t(x,\omega,a)\big)(\sqrt{\frac{\ln(3|X||\Omega||A|T/\delta)}{\max\{1,M_{t-1}(x,w,a)\}}}+\Psi)\}$
    
\end{lemma}
\begin{proof}
    The proof follows the lines of the proof of Lemma \ref{lem:P-conf}. This time, instead of using Eq.44 in  \citet{NearOptimalRL2008} we use the Hoeffding's inequality. We then union bound over all $x$, $w$, $a$ and $i$, and second term over all $x$, $w$, $a$ and $t$, where separate events have the cardinality of $1$. The third term follows from $\Psi_{u^s}\leq\Psi$.
\end{proof}

\noindent
The following lemma establishes confidence bounds on the receiver’s rewards for the partial feedback case.

\begin{lemma}
    Given any $\delta\in(0,1)$, with probability at least $1-2\delta$,  the following condition holds for every $x \in X$, $\omega \in \Omega$, $a \in A$ $i \in [m]$ and $t \in [T]$ jointly:

    \[
    | u^{r,t}(x,\omega,a) - \hat{u}_i^{r,t}(x,\omega,a)| \leq \xi_i^{r,t}(x,w,a)
    \]
    where $\xi_i^{r,t}(x,\omega,a) \coloneqq  \min\{1,w_\kappa\!\big(N_i^t(x,\omega,a)\big) \sqrt{\frac{\ln{(3m|X||\Omega||A|/\delta})}{\max\{1, N^t_i(x,\omega,a)\}}}+\bar w_\kappa\!\big(N_i^t(x,\omega,a)\big)(\sqrt{\frac{\ln(3|X||\Omega||A|T/\delta)}{\max\{1,M_{t-1}(x,w,a)\}}}+\Psi)\}$
    
\end{lemma}
\begin{proof}
    The proof follows the lines of the proof of Lemma \ref{lem:P-conf}. This time, instead of using Eq.44 in  \citet{NearOptimalRL2008} we use the Hoeffding's inequality. We then union bound over all $x$, $w$, $a$ and $i$, and second term over all $x$, $w$, $a$ and $t$, where separate events have the cardinality of $1$. The third term follows from $\Psi_{u^r}\leq\Psi$.
\end{proof}

\noindent 
Now that we have each estimator well defined, we introduce a generic notation that covers
all coordinates used in Appendix B. For the transition coordinates, let
\( 
\mathcal C_P := X\times \Omega \times A.
\) 
For the prior coordinates, let
\( 
\mathcal C_\mu := X.
\) 
For the reward coordinates, let
\( 
\mathcal C_{\mathrm{rew}}^{\mathrm{ff}} := X\times \Omega
\;\text{and}\;
\mathcal C_{\mathrm{rew}}^{\mathrm{pf}} := X\times \Omega \times A,
\) 
corresponding respectively to the full-feedback and partial-feedback settings.

\vspace{6pt}
\noindent 
For a coordinate family $\mathcal C$ and a coordinate $c\in \mathcal C$, let
$I_i^t(c)$ denote the indicator that coordinate $c$ is observed at episode $i$ of task $t$. Define the within-task count, terminal task indicator, and active-task count by
\( 
N_i^t(c):=\sum_{j=1}^{i-1} I_j^t(c),\;
N_m^t(c):=\sum_{j=1}^{m} I_j^t(c),\;
I_t(c):=\mathbf 1\{N_m^t(c)>0\},\;\text{and}\;
M_t(c):=\sum_{\tau=1}^t I_\tau(c).
\) Next, we provide the following two technical lemmas, we have leveraged while proving our regret and violation bounds.  

\begin{lemma}\label{lem:active_task_sum}
Define
\[
B_m(\kappa):=
\begin{cases}
1+\kappa\ln\!\bigl(1+\frac{m}{\kappa}\bigr), & \kappa>0,\\[0.5em]
0, & \kappa=0.
\end{cases}
\]
Then, for every coordinate family $\mathcal C$ and every $c\in \mathcal C$, the following hold:

\begin{align}
\sum_{i=1}^m \bar w_\kappa\!\bigl(N_i^t(c)\bigr)\,I_i^t(c)
&\le B_m(\kappa)\,I_t(c),
\label{eq:taskwise_shrinkage_bound}
\\
\sum_{t=1}^T \frac{I_t(c)}{\sqrt{\max\{M_{t-1}(c), 1\}}}
&\le 2\sqrt{M_T(c)+1},
\label{eq:active_task_harmonic}
\\
\frac{1}{T}\sum_{t=1}^T\sum_{i=1}^m
\bar w_\kappa\!\bigl(N_i^t(c)\bigr)\,I_i^t(c)\,
\sqrt{\frac{\beta}{\max\{M_{t-1}(c),1\}}}
&\le 2B_m(\kappa)\sqrt{\frac{\beta}{T}},
\qquad \forall \beta>0,
\label{eq:meta_sqrt_average}
\\
\frac{1}{T}\sum_{t=1}^T\sum_{i=1}^m
\bar w_\kappa\!\bigl(N_i^t(c)\bigr)\,I_i^t(c)\,\Psi
&\le B_m(\kappa)\Psi\frac{M_T(c)}{T}
\le B_m(\kappa)\Psi,
\label{eq:meta_Psi_average}
\end{align}
\end{lemma}

\begin{proof}
If \(\kappa=0\), then \(\bar w_\kappa(\cdot)\equiv 0\), and thus (\ref{eq:taskwise_shrinkage_bound}), (\ref{eq:meta_sqrt_average}), and (\ref{eq:meta_Psi_average}) are immediate. Assume therefore that $\kappa>0$. We first prove \eqref{eq:taskwise_shrinkage_bound}. If $I_t(c)=0$, then the left-hand side is zero.
Suppose $I_t(c)=1$, and let
\( 
1\le i_1< i_2<\cdots < i_{n_t(c)}\le m
\) 
be the episodes of task $t$ at which coordinate $c$ is observed, where $n_t(c):=N_m^t(c)\ge 1$.
Then, by construction, $N_{i_j}^t(c)=j-1$ for every $j\in[n_t(c)]$, and hence
\( 
\sum_{i=1}^m \bar w_\kappa\!\bigl(N_i^t(c)\bigr)\,I_i^t(c)
=
\sum_{j=1}^{n_t(c)} \frac{\kappa}{j-1+\kappa}.
\) 
Using integral comparison,
\[
\sum_{j=1}^{n_t(c)} \frac{\kappa}{j-1+\kappa}
\le
1+\int_0^{n_t(c)-1}\frac{\kappa}{u+\kappa}\,du
=
1+\kappa\ln\Bigl(1+\frac{n_t(c)-1}{\kappa}\Bigr).
\]
Since $n_t(c)\le m$, we get
\[
\sum_{i=1}^m \bar w_\kappa\!\bigl(N_i^t(c)\bigr)\,I_i^t(c)
\le
1+\kappa\ln \Bigl(1+\frac{m}{\kappa}\Bigr)
=
B_m(\kappa).
\]
This proves \eqref{eq:taskwise_shrinkage_bound}. Next we prove \eqref{eq:active_task_harmonic}. If $M_T(c)=0$, the claim is trivial. Otherwise, let
\( 
1\le t_1< t_2<\cdots < t_{M_T(c)}\le T
\) 
be the tasks for coordinate $c$ that have been seen. Then $M_{t_j-1}(c)=j-1$ for every $j$, and therefore
\[
\sum_{t=1}^T \frac{I_t(c)}{\sqrt{\max\{M_{t-1}(c), 1\}}}
=
1+\sum_{j=2}^{M_T(c)}\frac{1}{\sqrt{j-1}}
\le
2+\int_0^{M_T(c)-2}\frac{du}{\sqrt{u+1}}
\le
2\sqrt{M_T(c)-1} \le2\sqrt{M_T(c)+1} .
\]
This proves \eqref{eq:active_task_harmonic}. To prove \eqref{eq:meta_sqrt_average}, combine \eqref{eq:taskwise_shrinkage_bound} and
\eqref{eq:active_task_harmonic}:
\[
\frac{1}{T}\sum_{t=1}^T\sum_{i=1}^m
\bar w_\kappa\!\bigl(N_i^t(c)\bigr)\,I_i^t(c)\,
\sqrt{\frac{\beta}{\max\{M_{t-1}(c), 1\}}}
\le
\frac{B_m(\kappa)\sqrt{\beta}}{T}
\sum_{t=1}^T
\frac{I_t(c)}{\sqrt{\max\{M_{t-1}(c), 1\}}}
\]
\[
\le
\frac{2B_m(\kappa)\sqrt{\beta}}{T}\sqrt{M_T(c)-1}
\le
2B_m(\kappa)\sqrt{\frac{\beta}{T}},
\]
since $M_T(c)\le T$. For \eqref{eq:meta_Psi_average}, again using \eqref{eq:taskwise_shrinkage_bound},
\[
\frac{1}{T}\sum_{t=1}^T\sum_{i=1}^m
\bar w_\kappa\!\bigl(N_i^t(c)\bigr)\,I_i^t(c)\,\Psi
\le
\frac{B_m(\kappa)\Psi}{T}\sum_{t=1}^T I_t(c)
=
B_m(\kappa)\Psi\frac{M_T(c)}{T}
\le
B_m(\kappa)\Psi.
\]
\end{proof}

\begin{lemma}\label{lem:active_task_sum_sharp}
For every \(c\in C\) and every task \(t\in[T]\), the following holds:
\begin{equation}\label{eq:sharp_within_single}
\sum_{i=1}^m
w_\kappa\!\bigl(N_i^t(c)\bigr)\,
\frac{I_i^t(c)}{\sqrt{\max\{1,N_i^t(c)\}}}
\;\le\;
\frac{2N_m^t(c)}{\sqrt{N_m^t(c)}+\sqrt{\kappa}},
\end{equation}
\noindent 
with the convention that the right-hand side equals \(0\) when \(N_m^t(c)=0=\kappa\). More generally, for every finite subset \(D\subseteq C\), if
\( 
S_t(D):=\sum_{c\in D}N_m^t(c),
\) 
then
\begin{equation}\label{eq:sharp_within_aggregate}
\sum_{c\in D}\sum_{i=1}^m
w_\kappa\!\bigl(N_i^t(c)\bigr)\,
\frac{I_i^t(c)}{\sqrt{\max\{1,N_i^t(c)\}}}
\;\le\;
\frac{2S_t(D)\sqrt{|D|}}{\sqrt{S_t(D)}+\sqrt{\kappa|D|}},
\end{equation}
again with the convention that the right-hand side equals \(0\) when \(S_t(D)=0=\kappa\).
\end{lemma}

\begin{proof}
We first prove \eqref{eq:sharp_within_single}. If \(N_m^t(c)=0\), then the left-hand side is zero, so there is nothing to show. Assume \(N_m^t(c)=n\ge 1\), and let
\( 
1\le i_1<\cdots<i_n\le m
\) 
be the episodes in which \(c\) is observed. Then \(N_{i_j}^t(c)=j-1\), so
\[
\sum_{i=1}^m
w_\kappa\!\bigl(N_i^t(c)\bigr)\,
\frac{I_i^t(c)}{\sqrt{\max\{1,N_i^t(c)\}}}
=
\sum_{j=1}^n
w_\kappa(j-1)\,
\frac{1}{\sqrt{\max\{1,j-1\}}}.
\]
\emph{Case 1: \(\kappa=0\).}
Then \(w_0(\cdot)\equiv 1\), hence
\[
\sum_{j=1}^n
w_0(j-1)\,
\frac{1}{\sqrt{\max\{1,j-1\}}}
=
1+\sum_{j=2}^n \frac{1}{\sqrt{j-1}}
=
1+\sum_{u=1}^{n-1}\frac{1}{\sqrt{u}}.
\]
Using integral comparison,
\[
1+\sum_{u=1}^{n-1}\frac{1}{\sqrt{u}}
\le
2+\int_0^{n-1}\frac{du}{\sqrt{u+1}}
=
2\sqrt{n}
=
\frac{2n}{\sqrt{n}+\sqrt{0}}.
\]
\emph{Case 2: \(\kappa>0\).}
Then \(w_\kappa(0)=0\), so the first term vanishes and
\[
\sum_{j=1}^n
w_\kappa(j-1)\,
\frac{1}{\sqrt{\max\{1,j-1\}}}
=
\sum_{j=2}^n
\frac{j-1}{j-1+\kappa}\cdot \frac{1}{\sqrt{j-1}}
=
\sum_{u=1}^{n-1}\frac{\sqrt{u}}{u+\kappa}.
\]
Then, we have the following inequality for this case
\[
\frac{\sqrt{u}}{u+\kappa}
\le
\frac{2}{\sqrt{u+\kappa}+\sqrt{u-1+\kappa}}
=
2\bigl(\sqrt{u+\kappa}-\sqrt{u-1+\kappa}\bigr),
\qquad \forall u\ge 1.
\]
Therefore,
\[
\sum_{u=1}^{n-1}\frac{\sqrt{u}}{u+\kappa}
\le
2\sum_{u=1}^{n-1}\bigl(\sqrt{u+\kappa}-\sqrt{u-1+\kappa}\bigr)
=
2\bigl(\sqrt{n-1+\kappa}-\sqrt{\kappa}\bigr).
\]
Since the map \(x\mapsto x/(\sqrt{x}+\sqrt{\kappa})\) is increasing on \([0,\infty)\),
\[
2\bigl(\sqrt{n-1+\kappa}-\sqrt{\kappa}\bigr)
=
\frac{2(n-1)}{\sqrt{n-1+\kappa}+\sqrt{\kappa}}
\le
\frac{2n}{\sqrt{n}+\sqrt{\kappa}}.
\]
Hence,
\[
\sum_{u=1}^{n-1}\frac{\sqrt{u}}{u+\kappa}
\le
\frac{2n}{\sqrt{n}+\sqrt{\kappa}}.
\]
Therefore,
\[
\sum_{i=1}^m
w_\kappa\!\bigl(N_i^t(c)\bigr)\,
\frac{I_i^t(c)}{\sqrt{\max\{1,N_i^t(c)\}}}
\le
\frac{2n}{\sqrt{n}+\sqrt{\kappa}},
\]
which proves \eqref{eq:sharp_within_single}. We now prove \eqref{eq:sharp_within_aggregate}. Define
\[
g_a(x):=\frac{x}{\sqrt{x}+a},
\qquad x\ge 0,\ a\ge 0.
\]
A direct computation gives
\[
g_a''(x)
=
-\frac{\sqrt{x}+3a}{4\sqrt{x}\,(\sqrt{x}+a)^3}
\le 0
\qquad \forall x>0,
\]
and hence \(g_a\) is concave on \([0,\infty)\). Applying \eqref{eq:sharp_within_single} coordinate-wise and then Jensen’s inequality,
\[
\sum_{c\in D}\sum_{i=1}^m
w_\kappa\!\bigl(N_i^t(c)\bigr)\,
\frac{I_i^t(c)}{\sqrt{\max\{1,N_i^t(c)\}}}
\le
2\sum_{c\in D} g_{\sqrt{\kappa}}\!\bigl(N_m^t(c)\bigr)
\]
\[
\le
2|D|\,g_{\sqrt{\kappa}}\!\left(\frac{1}{|D|}\sum_{c\in D}N_m^t(c)\right)
=
\frac{2S_t(D)\sqrt{|D|}}{\sqrt{S_t(D)}+\sqrt{\kappa|D|}},
\]
which proves \eqref{eq:sharp_within_aggregate}.
\end{proof}

\subsection{Meta-Opt-Opt} \label{app:meta-optopt}

The difference between \textsc{Meta-Opt-Opt} and \textsc{Opt-Opt} in \citep{MarkovPersuasionScratch2025} 
is that we employ meta-estimators and, consequently, meta confidence bounds within the algorithm. 
In particular, Appendix C, Lemma~2 of \citet{MarkovPersuasionScratch2025} implies that, for any $\delta \in (0,1)$, 
under the good event $\mathcal{E}(\delta)$, \textsc{Meta-Opt-Opt} admits a feasible solution 
for every $i \in [m]$ and every task $t \in [T]$. The \textsc{Meta-Opt-Opt} procedure, executed at each iteration $i \in [m]$ and for each task 
$t \in [T]$, is as follows:

\begin{align}
\max_{q^t,\zeta^t,\epsilon^t}\quad
& \sum_{x\in X_k}\sum_{\omega\in\Omega}\sum_{a\in A}\sum_{x'\in X_{k+1}}
q^t(x,\omega,a,x')\Big(\hat u_i^{s,t}(x,\omega,a)+\xi_i^{s,t}(x,\omega,a)\Big)
\quad \text{s.t.} \tag{2a}\\[2pt]
& \!\!\!\!\!\!\!\!\!\!\!\!\!\!\!\!\!\!\!\!\!\!\!\! \sum_{x\in X_k}\sum_{\omega\in\Omega}\sum_{a\in A}\sum_{x'\in X_{k+1}}
q^t(x,\omega,a,x') = 1
\qquad \forall k\in[0\ldots L-1] \tag{2b}\\[2pt]
& \!\!\!\!\!\!\!\!\!\!\!\!\!\!\!\!\!\!\!\!\!\!\!\!\!\! \sum_{x'\in X_{k-1}}\sum_{\omega\in\Omega}\sum_{a\in A} q^t(x',\omega,a,x) 
= \sum_{\omega\in\Omega}\sum_{a\in A}\sum_{x'\in X_{k+1}} q^t(x,\omega,a,x') \quad
\qquad \; \forall k\in[0\ldots L-1],\ \forall x\in X_k \tag{2c}\\[2pt]
& \!\!\!\!\!\!\!\!\!\!\!\!\!\!\!\!\!\!\!\!\!\!\!\!q^t(x,\omega,a,x')-\hat P^t_i(x'|x,\omega,a)\!\!\!\!\!\sum_{x''\in X_{k+1}} \!\!\!\!q^t(x,\omega,a,x'') ]
\le \epsilon^t(x,\omega,a,x') \notag\\[1pt]
&\qquad \qquad\qquad\qquad\qquad\qquad\quad \;\; \forall k\in[0\ldots L-1],\ \forall(x,\omega,a,x')\in X_k\times\Omega\times A\times X_{k+1} \tag{2d}\\[2pt]
&\!\!\!\!\!\!\!\!\!\!\!\!\!\!\!\!\!\!\!\!\!\!\!\! \hat P^t_i(x'|x,a,\omega)\sum_{x''\in X_{k+1}} q^t(x,\omega,a,x'')-q^t(x,\omega,a,x')
\le \epsilon^t(x,\omega,a,x') \notag\\[1pt]
& \qquad \qquad\qquad\qquad\qquad\qquad\quad \;\; \forall k\in[0\ldots L-1],\ \forall(x,\omega,a,x')\in X_k\times\Omega\times A\times X_{k+1} \tag{2e}\\[2pt]
& \!\!\!\!\!\!\!\!\!\!\!\!\!\!\!\!\!\!\!\!\!\!\!\!\!\!\!\! \sum_{x'\in X_{k+1}} \!\!\!\!\!\epsilon^t(x,\omega,a,x')
\!\le \!\epsilon^t_i(x,\omega,a)\!\!\!\!\sum_{x'\in X_{k+1}} \!\!\!\!q^t(x,\omega,a,x')
\qquad\; \forall k\in[0\ldots L-1],\ \forall(x,\omega,a)\in X_k\times\Omega\times A \tag{2f}\\[4pt]
& \!\!\!\!\!\!\!\!\!\!\!\!\!\!\!\!\!\!\!\!\!\!\!\!\!\! q^t(x,\omega)-\hat\mu^t_i(\omega|x)\sum_{\omega'\in\Omega} q^t(x,\omega')
\le \zeta^t(x,\omega)
\qquad\qquad\qquad\;\;\;\;\; \forall k\in[0\ldots L-1],\ \forall(x,\omega)\in X_k\times\Omega \tag{2g}\\[2pt]
& \!\!\!\!\!\!\!\!\!\!\!\!\!\!\!\!\!\!\!\!\!\!\!\!\!\! \hat\mu^t_i(\omega|x)\sum_{\omega'\in\Omega} q^t(x,\omega')-q^t(x,\omega)
\le \zeta^t(x,\omega)
\qquad\qquad\qquad \;\;\;\;\; \forall k\in[0\ldots L-1],\ \forall(x,\omega)\in X_k\times\Omega \tag{2h}
\end{align}
\begin{align}
& \!\!\!\! \sum_{\omega\in\Omega} \zeta^t(x,\omega)
\le \zeta^t_i(x)\sum_{\omega\in\Omega} q^t(x,\omega)
\qquad\qquad\qquad\qquad\qquad\qquad\qquad\qquad \forall k\in[0\ldots L-1],\ \forall x\in X_k \tag{2i}\\[4pt]
&  \! \!\!\! \sum_{\omega\in\Omega}\sum_{x'\in X_{k+1}} q^t(x,\omega,a,x')
\Big(\hat u_i^{r,t}(x,\omega,a)+\xi_i^{r,t}(x,\omega,a)-\hat u_i^{r,t}(x,\omega,a')
+\xi_i^{r,t}(x,\omega,a')\Big)\ge 0 \notag\\
& \qquad\qquad\qquad\qquad\qquad\qquad\qquad\qquad\qquad\quad \forall k\in[0\ldots L-1],\ \forall(x,a)\in X_k\times A,\ \forall a'\in A \tag{2j}\\[2pt]
& \!\!\!q^t(x,\omega,a,x')\ge 0
\qquad\qquad\qquad\qquad\qquad \forall k\in[0\ldots L-1],\ \forall(x,\omega,a,x')\in X_k\times\Omega\times A\times X_{k+1} \tag{2k}
\end{align}

\subsubsection{Occupancy Measure Bounds} \label{occ:measure:bound}

We begin by showing that the estimated occupancy measures, concentrate around the true occupancy measures, as both definitions of regret and violation in our setting directly leverage this quantity.

\begin{lemma}\label{lem:task_avg_occ_sharp}
Fix $\delta\in(0,1)$ and assume that the good event $\mathcal E(\delta)$ holds. Then, with probability at least $1-2\delta$,
\[
\begin{aligned}
\frac{1}{T}\sum_{t=1}^T\sum_{i\in[m]}\bigl\lVert q_i^t-\hat{q}_i^t\bigr\rVert_1
\;\le\;
\mathcal{O}\Bigg(&
\frac{L^2\sqrt{m}}{\sqrt{m}+\sqrt{\kappa|X||\Omega||A|}}\;
|X|\sqrt{m|\Omega||A|\ln\Bigl(\frac{m|X||\Omega||A|}{\delta}\Bigr)}
\Bigg).
\end{aligned}
\]
\end{lemma}

\begin{proof}
Let
\( 
\ell_P:=\ln\Bigl(\frac{m|X||\Omega||A|}{\delta}\Bigr),
\ell_\mu:=\ln\Bigl(\frac{m|X|}{\delta}\Bigr),
\) 
and 
\( 
\beta_P:=2|X|\ln\Bigl(\frac{|X||\Omega||A|T}{\delta}\Bigr),
\beta_\mu:=2|\Omega|\ln\Bigl(\frac{|X|T}{\delta}\Bigr)
\). Write for $c=(x,\omega,a)\in\mathcal C_P$ and $d=x\in\mathcal C_\mu$,
\[
\epsilon_i^t(c)=\epsilon_{i,\mathrm{w}}^t(c)+\epsilon_{i,\mathrm{m}}^t(c),
\qquad
\zeta_i^t(d)=\zeta_{i,\mathrm{w}}^t(d)+\zeta_{i,\mathrm{m}}^t(d),
\]
where
\[
\epsilon_{i,\mathrm{w}}^t(c)
:=
w_\kappa\!\bigl(N_i^t(c)\bigr)
\sqrt{\frac{2|X_{k(x)+1}|\,\ell_P}{\max\{1,N_i^t(c)\}}},
\quad
\epsilon_{i,\mathrm{m}}^t(c)
:=
\bar w_\kappa\!\bigl(N_i^t(c)\bigr)
\left(
\sqrt{\frac{\beta_P}{\max\{M_{t-1}(c), 1\}}}+\Psi
\right),
\]
and
\[
\zeta_{i,\mathrm{w}}^t(d)
:=
w_\kappa\!\bigl(N_i^t(d)\bigr)
\sqrt{\frac{2|\Omega|\,\ell_\mu}{\max\{1,N_i^t(d)\}}},
\quad
\zeta_{i,\mathrm{m}}^t(d)
:=
\bar w_\kappa\!\bigl(N_i^t(d)\bigr)
\left(
\sqrt{\frac{\beta_\mu}{M_{t-1}(d)\vee 1}}+\Psi
\right).
\]
From Appendix D, Lemma 3 in \citep{MarkovPersuasionScratch2025} we have with probability at least $1-2\delta$,
\begin{equation}\label{eq:occ_reduction_sharp}
\frac{1}{T}\sum_{t=1}^T \sum_{i=1}^m \|q_i^t-\hat q_i^t\|_1
\le
\frac{2L}{T}\sum_{c\in\mathcal C_P}\sum_{t=1}^T\sum_{i=1}^m \epsilon_i^t(c)\,I_i^t(c)
+
\frac{L}{T}\sum_{d\in\mathcal C_\mu}\sum_{t=1}^T\sum_{i=1}^m \zeta_i^t(d)\,I_i^t(d)
+
4L|X|\sqrt{2m\ln\Bigl(\frac{L}{\delta}\Bigr)}.
\end{equation}
Thus it remains to control the transition and prior contributions. Fix a task $t$ and a layer $s\in\{0,\dots,L-1\}$. Let
\( 
\mathcal D_s^P:=X_s\times\Omega\times A.
\) 
Then, by Lemma~\ref{lem:active_task_sum_sharp},
\begin{align*}
\sum_{c\in\mathcal D_s^P}\sum_{i=1}^m \epsilon_{i,\mathrm{w}}^t(c)\,I_i^t(c)
&\le
\sqrt{2|X|\,\ell_P}
\sum_{c\in\mathcal D_s^P}\sum_{i=1}^m
\frac{N_i^t(c)}{N_i^t(c)+\kappa}
\frac{I_i^t(c)}{\sqrt{\max\{1,N_i^t(c)\}}}
\\ & \le
\sqrt{2|X|\,\ell_P}
\cdot
\frac{2S_t(\mathcal D_s^P)\sqrt{|\mathcal D_s^P|}}
{\sqrt{S_t(\mathcal D_s^P)}+\sqrt{\kappa|\mathcal D_s^P|}}.
\end{align*}
Since in a loop-free MPP each episode visits at most one triplet in a fixed layer, we have 
\( 
S_t(\mathcal D_s^P)\le m,
\) and it follows that \(
|\mathcal D_s^P|=|X_s||\Omega||A|\le |X||\Omega||A|.
\) 
Hence
\[
\sum_{c\in\mathcal D_s^P}\sum_{i=1}^m \epsilon_{i,\mathrm{w}}^t(c)\,I_i^t(c)
\le
\frac{2\sqrt m}{\sqrt m+\sqrt{\kappa|X||\Omega||A|}}
\,|X_s|\sqrt{2m|\Omega||A|\,\ell_P}.
\]
Summing over at most $L$ preceding layers inside each $k$-sum and then over $k=0,\dots,L-1$ yields
\begin{equation} \label{eq:transition_within_bound}
\frac{2L}{T}\sum_{c\in\mathcal C_P}\sum_{t=1}^T\sum_{i=1}^m \epsilon_{i,\mathrm{w}}^t(c)\,I_i^t(c)
\le
\frac{4L^2\sqrt m}{\sqrt m+\sqrt{\kappa|X||\Omega||A|}}
\,|X|\sqrt{2m|\Omega||A|\,\ell_P}.
\end{equation}
Next we treat the within-task prior contribution. For a fixed layer $s$, let
\( 
\mathcal D_s^\mu:=X_s.
\) 
Applying Lemma~\ref{lem:active_task_sum_sharp} again,
\[
\sum_{d\in\mathcal D_s^\mu}\sum_{i=1}^m \zeta_{i,\mathrm{w}}^t(d)\,I_i^t(d)
\le
\sqrt{2|\Omega|\,\ell_\mu}
\cdot
\frac{2S_t(\mathcal D_s^\mu)\sqrt{|\mathcal D_s^\mu|}}
{\sqrt{S_t(\mathcal D_s^\mu)}+\sqrt{\kappa|\mathcal D_s^\mu|}}.
\]
Since each episode visits at most one state in a fixed layer, 
\( 
S_t(\mathcal D_s^\mu)\le m,
\;
|\mathcal D_s^\mu|=|X_s|\le |X|.
\) 
Therefore
\[
\sum_{d\in\mathcal D_s^\mu}\sum_{i=1}^m \zeta_{i,\mathrm{w}}^t(d)\,I_i^t(d)
\le
\frac{2\sqrt m}{\sqrt m+\sqrt{\kappa|X|}}
\sqrt{2m|X_s||\Omega|\,\ell_\mu}.
\]
Summing over the at most $L$ preceding layers and then over $k=0,\dots,L-1$ gives
\begin{equation} \label{eq:prior_within_bound}
\frac{L}{T}\sum_{d\in\mathcal C_\mu}\sum_{t=1}^T\sum_{i=1}^m \zeta_{i,\mathrm{w}}^t(d)\,I_i^t(d)
\le
\frac{2L^2\sqrt m}{\sqrt m+\sqrt{\kappa|X|}}
\sqrt{2m|X||\Omega|\,\ell_\mu}.
\end{equation}
By Lemma~\ref{lem:active_task_sum},
\begin{align*}
\frac{2L}{T}\sum_{c\in\mathcal C_P}\sum_{t=1}^T\sum_{i=1}^m \epsilon_{i,\mathrm{m}}^t(c)\,I_i^t(c)
&\le
\frac{2L}{T}\sum_{c\in\mathcal C_P}\sum_{t=1}^T\sum_{i=1}^m
\frac{\kappa}{N_i^t(c)+\kappa}\,I_i^t(c)
\sqrt{\frac{\beta_P}{\max\{M_{t-1}(c), 1\}}}
\\&\quad +
\frac{2L\Psi}{T}\sum_{c\in\mathcal C_P}\sum_{t=1}^T\sum_{i=1}^m
\frac{\kappa}{N_i^t(c)+\kappa}\,I_i^t(c)
\\& 
\le
4L|\mathcal C_P|\,B_m(\kappa)\sqrt{\frac{\beta_P}{T}}
+
\frac{2LB_m(\kappa)\Psi}{T}\sum_{c\in\mathcal C_P}M_T(c).
\end{align*}
Since $|\mathcal C_P|=|X||\Omega||A|$, this becomes
\begin{equation}\label{eq:transition_meta_bound}
\frac{2L}{T}\sum_{c\in\mathcal C_P}\sum_{t=1}^T\sum_{i=1}^m \epsilon_{i,\mathrm{m}}^t(c)\,I_i^t(c)
\le
4L|X||\Omega||A|\,B_m(\kappa)\sqrt{\frac{\beta_P}{T}}
+
\frac{2LB_m(\kappa)\Psi}{T}\sum_{c\in\mathcal C_P}M_T(c).
\end{equation}
Again by Lemma~\ref{lem:active_task_sum},
\begin{align*}
\frac{L}{T}\sum_{d\in\mathcal C_\mu}\sum_{t=1}^T\sum_{i=1}^m \zeta_{i,\mathrm{m}}^t(d)\,I_i^t(d)
& \le
\frac{L}{T}\sum_{d\in\mathcal C_\mu}\sum_{t=1}^T\sum_{i=1}^m
\frac{\kappa}{N_i^t(d)+\kappa}\,I_i^t(d)
\sqrt{\frac{\beta_\mu}{M_{t-1}(d)\vee 1}}
\\ & \quad +
\frac{L\Psi}{T}\sum_{d\in\mathcal C_\mu}\sum_{t=1}^T\sum_{i=1}^m
\frac{\kappa}{N_i^t(d)+\kappa}\,I_i^t(d)
\\&\le
2L|\mathcal C_\mu|\,B_m(\kappa)\sqrt{\frac{\beta_\mu}{T}}
+
\frac{LB_m(\kappa)\Psi}{T}\sum_{d\in\mathcal C_\mu}M_T(d).
\end{align*}
Since $|\mathcal C_\mu|=|X|$, this becomes
\begin{equation}\label{eq:prior_meta_bound}
\frac{L}{T}\sum_{d\in\mathcal C_\mu}\sum_{t=1}^T\sum_{i=1}^m \zeta_{i,\mathrm{m}}^t(d)\,I_i^t(d)
\le
2L|X|\,B_m(\kappa)\sqrt{\frac{\beta_\mu}{T}}
+
\frac{LB_m(\kappa)\Psi}{T}\sum_{d\in\mathcal C_\mu}M_T(d).
\end{equation}
Substituting \eqref{eq:transition_within_bound}, \eqref{eq:prior_within_bound},
\eqref{eq:transition_meta_bound}, and \eqref{eq:prior_meta_bound} into
\eqref{eq:occ_reduction_sharp} yields
\[
\frac{1}{T}\sum_{t=1}^T \sum_{i=1}^m \|q_i^t-\hat q_i^t\|_1
\le
C_{\mathrm{occ}}(m,\delta)
+
4L|X||\Omega||A|\,B_m(\kappa)\sqrt{\frac{\beta_P}{T}}
+
2L|X|\,B_m(\kappa)\sqrt{\frac{\beta_\mu}{T}}
+
H_{\mathrm{occ}}(T),
\]
where
\[ 
C_{\mathrm{occ}}(m,\delta;\kappa)
\!:= \! 
\frac{4L^2\sqrt m|X|}{\!\sqrt m\!+\!\sqrt{\kappa|X||\Omega||A|}\!}
\sqrt{2m|\Omega||A|\,\ell_P}
\!+\!
\frac{2L^2\sqrt m}{\!\sqrt m+\sqrt{\kappa|X|}\!}
\sqrt{2m|X||\Omega|\,\ell_\mu}
\!+\!
4L|X|\sqrt{2m\ln\Bigl(\frac{L}{\delta}\Bigr)},
\]
and,
\[
H_{\mathrm{occ}}(T)
=
\frac{2LB_m(\kappa)\Psi}{T}\sum_{c\in\mathcal C_P}M_T(c)
+
\frac{LB_m(\kappa)\Psi}{T}\sum_{d\in\mathcal C_\mu}M_T(d).
\]
Finally, since $M_T(c)\le T$ for every $c\in\mathcal C_P$ and $M_T(d)\le T$ for every
$d\in\mathcal C_\mu$,
\[
H_{\mathrm{occ}}(T)
\le
2LB_m(\kappa)\Psi\,|\mathcal C_P|
+
LB_m(\kappa)\Psi\,|\mathcal C_\mu|
=
LB_m(\kappa)\Psi\bigl(2|X||\Omega||A|+|X|\bigr).
\]
Taking the limit completes the proof.
\end{proof}

\begin{lemma}\label{lem:meta-xi-bound}
Fix $\delta\in(0,1)$ and assume that the good event $\mathcal E(\delta)$ holds. Then, for both $(\xi_i^{r,t})$, $(\xi_i^{s,t})$, with probability at least $1-\delta$, we have,
\[
\frac{1}{T}\sum_{t=1}^T\sum_{i=1}^{m} (\xi^{r,t}_i)^\top q_i^t
\le
\tilde{\mathcal{O}}\!\left(
\frac{\sqrt{Lm}}
{\sqrt{Lm}+\sqrt{\kappa|X||\Omega||A|}}
\sqrt{
Lm|X||\Omega||A|
\ln\!\Bigl(
\frac{m|X||\Omega||A|}{\delta}
\Bigr)
}
\right).
\]
\end{lemma}

\begin{proof} In the full-feedback case, define
\( 
\mathcal C_{\mathrm{rew}}:=X\times \Omega,
\;
q_i^t(x,\omega):=\sum_{a\in A} q_i^t(x,\omega,a),
\) 
and
\( 
\ell_{\mathrm{rew}}
:=
\ln\Bigl(\frac{3m|X||\Omega|}{\delta}\Bigr),
\;
\beta_{\mathrm{rew}}
:=
\ln\Bigl(\frac{3|X||\Omega|T}{\delta}\Bigr).
\) Then, in the partial-feedback case, define
\( 
\mathcal C_{\mathrm{rew}}:=X\times \Omega \times A,
\) 
and
\( 
\ell_{\mathrm{rew}}
:=
\ln\Bigl(\frac{3m|X||\Omega||A|}{\delta}\Bigr),
\;
\beta_{\mathrm{rew}}
:=
\ln\Bigl(\frac{3|X||\Omega||A|T}{\delta}\Bigr).
\)
We prove the sender-reward bound. The receiver-reward bound follows by the same argument, replacing
$\xi_i^{s,t}$ with $\xi_i^{r,t}$ throughout. For each reward coordinate $c\in\mathcal C_{\mathrm{rew}}$, define
\[ 
\xi_{i,\mathrm{w}}^{s,t}(c)
:=
w_\kappa\!\bigl(N_i^t(c)\bigr)
\sqrt{\frac{\ell_{\mathrm{rew}}}{\max\{1,N_i^t(c)\}}},
\quad 
\xi_{i,\mathrm{m}}^{s,t}(c)
:=
\bar w_\kappa\!\bigl(N_i^t(c)\bigr)
\left(
\sqrt{\frac{\beta_{\mathrm{rew}}}{\max\{M_{t-1}(c), 1\}}}
+\Psi
\right).
\]
Then, we have
\( 
\xi_i^{s,t}(c)\le \xi_{i,\mathrm{w}}^{s,t}(c)+\xi_{i,\mathrm{m}}^{s,t}(c).
\)  For each pair $(t,i)\in[T]\times[m]$,
\( 
Y_{t,i}
:=
(\xi_i^{s,t})^\top q_i^t
-
\sum_{c\in\mathcal C_{\mathrm{rew}}}\xi_i^{s,t}(c)\,I_i^t(c)
\) 
is a martingale-difference sequence with respect to the natural filtration. Since, every reward is at most $1$,
and in each episode at most one coordinate is visited per layer, we have
\( 
|Y_{t,i}|\le L
\; \text{almost surely for all }(t,i)\in[T]\times[m].
\) 
Applying Azuma--Hoeffding inequality yields the following w.p. $1-\delta$,
\begin{equation}\label{eq:reward_task_avg_reduction_sharp}
\frac{1}{T}\sum_{t=1}^T \sum_{i=1}^m (\xi_i^{s,t})^\top q_i^t
\le
\frac{1}{T}\sum_{t=1}^T\sum_{c\in\mathcal C_{\mathrm{rew}}}\sum_{i=1}^m
\xi_i^{s,t}(c)\,I_i^t(c)
+
L\sqrt{\frac{2m\ln(1/\delta)}{T}}.
\end{equation}
Using the decomposition of $\xi_i^{s,t}(c)$,
\[
\frac{1}{T}\sum_{t=1}^T\sum_{c\in\mathcal C_{\mathrm{rew}}}\sum_{i=1}^m
\xi_i^{s,t}(c)\,I_i^t(c)
\le
\frac{1}{T}\sum_{t=1}^T\sum_{c\in\mathcal C_{\mathrm{rew}}}\sum_{i=1}^m
\xi_{i,\mathrm{w}}^{s,t}(c)\,I_i^t(c)
+
\frac{1}{T}\sum_{t=1}^T\sum_{c\in\mathcal C_{\mathrm{rew}}}\sum_{i=1}^m
\xi_{i,\mathrm{m}}^{s,t}(c)\,I_i^t(c).
\]
We first bound the within-task term. For each fixed task $t$, by Lemma~\ref{lem:active_task_sum_sharp},
applied with $\mathcal D=\mathcal C_{\mathrm{rew}}$,
\[
\sum_{c\in\mathcal C_{\mathrm{rew}}}\sum_{i=1}^{m}
\xi_{i,\mathrm{w}}^{s,t}(c)\,I_i^t(c)
\le
\sqrt{\ell_{\mathrm{rew}}}\,
\frac{2S_t(\mathcal C_{\mathrm{rew}})\sqrt{|\mathcal C_{\mathrm{rew}}|}}
{\sqrt{S_t(\mathcal C_{\mathrm{rew}})}+\sqrt{\kappa|\mathcal C_{\mathrm{rew}}|}},
\]
where
\( 
S_t(\mathcal C_{\mathrm{rew}})
:=
\sum_{c\in\mathcal C_{\mathrm{rew}}}N_m^t(c).
\) 
Since the process is loop-free and each episode visits at most one reward coordinate per layer,
\( 
S_t(\mathcal C_{\mathrm{rew}})\le Lm.
\) 
Using again that $x\mapsto x/(\sqrt{x}+a)$ is increasing on $[0,\infty)$, we obtain
\[
\sum_{c\in\mathcal C_{\mathrm{rew}}}\sum_{i=1}^{m}
\xi_{i,\mathrm{w}}^{s,t}(c)\,I_i^t(c)
\le
\frac{2\sqrt{Lm}}
{\sqrt{Lm}+\sqrt{\kappa|\mathcal C_{\mathrm{rew}}|}}
\sqrt{Lm\,|\mathcal C_{\mathrm{rew}}|\,\ell_{\mathrm{rew}}}.
\]
Since the function
\( 
u\mapsto \frac{2\sqrt{Lm}}{\sqrt{Lm}+\sqrt{u\,|\mathcal C_{\mathrm{rew}}|}}
\sqrt{Lm\,|\mathcal C_{\mathrm{rew}}|\,\ell_{\mathrm{rew}}}
\) 
is decreasing in $u\ge 0$, and $\kappa\ge \kappa$, it follows that
\[
\sum_{c\in\mathcal C_{\mathrm{rew}}}\sum_{i=1}^{m}
\xi_{i,\mathrm{w}}^{s,t}(c)\,I_i^t(c)
\le
\frac{2\sqrt{Lm}}
{\sqrt{Lm}+\sqrt{\kappa|\mathcal C_{\mathrm{rew}}|}}
\sqrt{Lm\,|\mathcal C_{\mathrm{rew}}|\,\ell_{\mathrm{rew}}}.
\]
Averaging over $t$ gives
\begin{equation}\label{eq:reward_within_sharp}
\frac{1}{T}\sum_{t=1}^T\sum_{c\in\mathcal C_{\mathrm{rew}}}\sum_{i=1}^{m}
\xi_{i,\mathrm{w}}^{s,t}(c)\,I_i^t(c)
\le
\frac{2\sqrt{Lm}}
{\sqrt{Lm}+\sqrt{\kappa|\mathcal C_{\mathrm{rew}}|}}
\sqrt{Lm\,|\mathcal C_{\mathrm{rew}}|\,\ell_{\mathrm{rew}}}.
\end{equation}
By the definition of $\xi_{i,\mathrm{m}}^{s,t}(c)$ and
Lemma~\ref{lem:active_task_sum},
\begin{align*}
\frac{1}{T}\sum_{t=1}^T\sum_{c\in\mathcal C_{\mathrm{rew}}}\sum_{i=1}^{m}
\xi_{i,\mathrm{m}}^{s,t}(c)\,I_i^t(c)
& \le
\frac{1}{T}\sum_{c\in\mathcal C_{\mathrm{rew}}}\sum_{t=1}^T\sum_{i=1}^{m}
\frac{\kappa}{N_i^t(c)+\kappa}\,I_i^t(c)
\sqrt{\frac{\beta_{\mathrm{rew}}}{\max\{M_{t-1}(c), 1\}}}
\\ & \quad +
\frac{\Psi}{T}\sum_{c\in\mathcal C_{\mathrm{rew}}}\sum_{t=1}^T\sum_{i=1}^{m}
\frac{\kappa}{N_i^t(c)+\kappa}\,I_i^t(c)
\\ & \le
2|\mathcal C_{\mathrm{rew}}|\,B_m(\kappa)\sqrt{\frac{\beta_{\mathrm{rew}}}{T}}
+
\frac{B_m(\kappa)\Psi}{T}\sum_{c\in\mathcal C_{\mathrm{rew}}}M_T(c).
\end{align*}
Thus
\begin{equation}\label{eq:reward_meta_sharp}
\frac{1}{T}\sum_{t=1}^T\sum_{c\in\mathcal C_{\mathrm{rew}}}\sum_{i=1}^{m}
\xi_{i,\mathrm{m}}^{s,t}(c)\,I_i^t(c)
\le
2|\mathcal C_{\mathrm{rew}}|\,B_m(\kappa)\sqrt{\frac{\beta_{\mathrm{rew}}}{T}}
+
\frac{B_m(\kappa)\Psi}{T}\sum_{c\in\mathcal C_{\mathrm{rew}}}M_T(c).
\end{equation}
Substituting \eqref{eq:reward_within_sharp} and \eqref{eq:reward_meta_sharp} into
\eqref{eq:reward_task_avg_reduction_sharp} yields
\begin{align*}
\frac{1}{T}\sum_{t=1}^T \sum_{i=1}^m (\xi_i^{s,t})^\top q_i^t
&\le
C_{\mathrm{rew}}^{\mathrm{sim}}(m,T,\delta;\kappa)
+
2|\mathcal C_{\mathrm{rew}}|\,B_m(\kappa)\sqrt{\frac{\beta_{\mathrm{rew}}}{T}}
+
\frac{B_m(\kappa)\Psi}{T}\sum_{c\in\mathcal C_{\mathrm{rew}}}M_T(c),
\end{align*}
where
\( 
C_{\mathrm{rew}}^{\mathrm{sim}}(m,T,\delta;\kappa)
:=
\frac{2\sqrt{Lm}}
{\sqrt{Lm}+\sqrt{\kappa|\mathcal C_{\mathrm{rew}}|}}
\sqrt{Lm\,|\mathcal C_{\mathrm{rew}}|\,\ell_{\mathrm{rew}}}
+
L\sqrt{\frac{2m\ln(1/\delta)}{T}}
\).  Finally, since $M_T(c)\le T$ for every $c\in\mathcal C_{\mathrm{rew}}$,
\( 
\frac{B_m(\kappa)\Psi}{T}\sum_{c\in\mathcal C_{\mathrm{rew}}}M_T(c)
\le
B_m(\kappa)\Psi |\mathcal C_{\mathrm{rew}}|.
\) Taking the limit, 
this is the desired bound. 

\vspace{6pt}
\noindent
The proof for
\( 
\frac{1}{T}\sum_{t=1}^T \sum_{i=1}^m (\xi_i^{r,t})^\top q_i^t
\) 
is identical, replacing $\xi_i^{s,t}$ with $\xi_i^{r,t}$ throughout. This yields
\[
\frac{1}{T}\sum_{t=1}^T \sum_{i=1}^m (\xi_i^{r,t})^\top q_i^t
\le
C_{\mathrm{rew}}^{\mathrm{sim}}(m,T,\delta;\kappa)
+
2|\mathcal C_{\mathrm{rew}}|\,B_m(\kappa)\sqrt{\frac{\beta_{\mathrm{rew}}}{T}}
+
\frac{B_m(\kappa)\Psi}{T}\sum_{c\in\mathcal C_{\mathrm{rew}}}M_T(c),
\]
with
\( 
\frac{B_m(\kappa)\Psi}{T}\sum_{c\in\mathcal C_{\mathrm{rew}}}M_T(c)
\le
B_m(\kappa)\Psi |\mathcal C_{\mathrm{rew}}|,
\) which completes the proof.
\end{proof}

\subsection{Regret and Violation Bounds} \label{sub:reg-viol}

\begin{theorem}
Given any $\delta\in(0,1)$, with probability at least $1-11\delta$, Algorithm~\ref{alg:OPPS-full} attains the following cumulative task averaged expected regret:
\[
R^T_m \le 
\tilde{\mathcal{O}}\Bigg(
\frac{L^2\sqrt{m}}{\sqrt{m}+\sqrt{\kappa|X||\Omega||A|}}\;
|X|\sqrt{m|\Omega||A|\ln\!\Bigl(\frac{m|X||\Omega||A|}{\delta}\Bigr)} 
\Bigg)
\]
\end{theorem}
\begin{proof}
    Using Lemma \ref{lem:task_avg_occ_sharp} to bound the occupancy measure and Lemma \ref{lem:meta-xi-bound} to bound the average rewards within the proof of Appendix D.2 Theorem 1 in \citet{MarkovPersuasionScratch2025} concludes the proof. 
\end{proof}

\begin{theorem}
Given any $\delta\in(0,1)$, with probability at least $1-11\delta$, Algorithm~\ref{alg:OPPS-full} attains the following cumulative task averaged expected violation:
\[
V^T_m \le 
\tilde{\mathcal{O}}\Bigg(
\frac{L^2\sqrt{m}}{\sqrt{m}+\sqrt{\kappa|X||\Omega||A|}}\;
|X|\sqrt{m|\Omega||A|\ln\!\Bigl(\frac{m|X||\Omega||A|}{\delta}\Bigr)} 
\Bigg)
\]
\end{theorem}
\begin{proof}
    Using Lemma \ref{lem:reward-receiver-partinfo-est} for the full feedback receiver confidence bound, Lemma \ref{lem:task_avg_occ_sharp} to bound the occupancy measure and Lemma \ref{lem:meta-xi-bound} to bound the average rewards within the proof of Appendix D.3 Theorem 2 in \citet{MarkovPersuasionScratch2025} concludes the proof. 
\end{proof}

\begin{theorem}
Given any $\delta\in(0,1)$, with probability at least $1-11\delta$, Algorithm~\ref{alg:OPPS-part} attains the following cumulative task averaged expected regret:
\[
R^T_m \le 
\tilde{\mathcal{O}}\Bigg(NL|X||\Omega||A|+
\frac{L^2\sqrt{m}}{\sqrt{m}+\sqrt{\kappa|X||\Omega||A|}}\;
|X|\sqrt{m|\Omega||A|\ln\!\Bigl(\frac{m|X||\Omega||A|}{\delta}\Bigr)} 
\Bigg)
\]
where $N \coloneqq \lceil m^\alpha \rceil$ is the length of exploration phase 
\end{theorem}
\begin{proof}
    Using Lemma \ref{lem:task_avg_occ_sharp} to bound the occupancy measure and Lemma \ref{lem:meta-xi-bound} to bound the average rewards, within the proof of Appendix E.1 Theorem 3 in \citet{MarkovPersuasionScratch2025} concludes the proof. 
\end{proof}

\begin{lemma} \label{lem:partial-violation}
Under the event $\mathcal{E}(\delta)$, with probability at least $1-3\delta$, the following holds for task $t \in [T]$:
\[
V^T_m \!
\le\!
\tilde{\mathcal{O}}\Bigg(
\frac{L^2m|X|}{\sqrt{m}+\sqrt{\kappa|X||\Omega||A|}}\;
\!\!\sqrt{|\Omega||A|\ln\!\Bigl(\frac{m|X||\Omega||A|}{\delta}\Bigr)} 
\Bigg) +\frac{1}{T} \sum_T\sum_{m}\sum_{X, \Omega, A}\!\!\!q^t_i(x,\omega,a)\xi^{r,t}_i(x,\omega,b^t_i(a,x))
\]
\end{lemma}

\begin{proof}
    Using Lemma \ref{lem:task_avg_occ_sharp} to bound the occupancy measure and Lemma \ref{lem:meta-xi-bound} to bound the average rewards within the proof of Appendix E.1 Lemma 11 in \citet{MarkovPersuasionScratch2025} concludes the proof.
\end{proof}

\begin{theorem}
Given any $\delta\in(0,1)$, with probability at least $1-13\delta$, Algorithm~\ref{alg:OPPS-part} attains the following cumulative task averaged expected violation:
\begin{align}
V^T_m & \!\le \tilde{\mathcal{O}}\Bigg[\rho  \Bigg(
\frac{Lm}{\sqrt{m}\!+\!\sqrt{\kappa|X||\Omega||A|}} 
\!+\! \frac{\sqrt{|X||\Omega||A|}N}{\!\sqrt{NL}\!+\!\sqrt{\kappa|X|\Omega||A|}\!}\!+\! \sqrt{N} \! 
+ \!\frac{m}{\!\sqrt{NL}\!+\!\sqrt{\kappa|X|\Omega||A|}\!}\!+\!\sqrt{\frac{m^2}{N}} \!+\! \frac{m\kappa}{NL} \! \Bigg) \Bigg] \notag
\end{align}
\( 
\text{where}\
\rho \;:=\; |X||\Omega||A|^2L\,
\sqrt{\ln\!\Bigl(\frac{1}{\delta}\Bigr)}
\), and $N \coloneqq \lceil m^\alpha \rceil$ is the length of exploration phase.
\end{theorem}

\begin{proof}
    Using Lemma~\ref{lem:partial-violation} together with the properties of our reward estimators and confidence bounds under partial feedback, and following the steps in Appendix E.2, Theorem~4 \citet{MarkovPersuasionScratch2025}, we conclude the proof. 
\end{proof}

\end{document}